\documentclass[times,review]{elsarticle}
\usepackage[labelfont=bf]{caption}

\usepackage{amsmath,amsfonts,amssymb,amsthm,booktabs,color,epsfig,graphicx,hyperref,url}

\usepackage{fullpage}

\theoremstyle{plain}
\newtheorem{Thm}{Theorem}
\newtheorem{Pro}{Proposition} 
\newtheorem{Lem}{Lemma}

\newtheorem{Cor}{Corollary}

\theoremstyle{definition}

\newtheorem{Rem}{Remark}
\newtheorem{Def}{Definition}

\newcommand{\cl}{\mathcal}

\newcommand{\wh}{\widehat}
\newcommand{\wt}{\widetilde}

\newcommand{\mbf}{\mathbf}
\newcommand{\bb}{\mathbb}

\newcommand{\E}{\mathrm{E}}

\newcommand{\EqD}{\overset{d}{=}}
\newcommand{\ConvD}{\overset{d}{\rightarrow}}

 \newcommand{\ba}{\mbf{a}}
 \newcommand{\bw}{\mbf{w}}

\newcommand{\dsph}{\bb{S}_+^{d-1}}

\def\pp#1{ \left(#1\right) }
\def\pb#1{ \left[#1\right] }
\def\pc#1{ \left\{#1\right\} }

\def\Prt#1{ \Pr \pp{#1} }
 
\def\ind#1{  \mbf{1}  \pc{#1}}

\DeclareMathOperator\supp{supp}

\begin{document}

\begin{frontmatter}

\title{On estimation and order selection for  multivariate extremes via clustering}

\author[1]{Shiyuan Deng\corref{a}}
\author[1]{He Tang\corref{a}}
\author[1]{Shuyang Bai\corref{mycorrespondingauthor}}

\address[1]{Department of Statistics,  310 Herty Dr., University of Georgia, Athens, GA 30602}

\cortext[a]{The first two authors contributed equally to this work.}
 
\cortext[mycorrespondingauthor]{Corresponding author. Email address: \url{bsy9142@uga.edu}}

\begin{abstract}
We investigate the estimation of multivariate extreme models with a discrete spectral measure using spherical clustering techniques. The primary contribution involves devising a method for selecting the  order, that is, the number of clusters. The method consistently identifies the true order, i.e., the number of spectral atoms, and enjoys intuitive implementation in practice. Specifically, we introduce an extra penalty term to the well-known simplified average silhouette width, which penalizes small cluster sizes and small dissimilarities between cluster centers. Consequently,  we provide a consistent method for determining the order of a max-linear factor model, where a typical information-based approach is not viable. 
Our second contribution is a large-deviation-type analysis for estimating the discrete spectral measure through clustering methods, which serves as an assessment of the convergence quality of clustering-based estimation for multivariate extremes. Additionally, as a third contribution, we discuss how estimating the discrete measure can lead to parameter estimations of heavy-tailed factor models. We also present simulations and real-data studies that demonstrate order selection and factor model estimation.

\end{abstract}

\begin{keyword} 
clustering \sep
factor models \sep
max-linear models \sep
multivariate extremes \sep
order selection \sep
silhouettes

\MSC[2020] Primary 62G32 \sep
Secondary 60G70
\end{keyword}

\end{frontmatter}


\section{Introduction\label{sec:intro}}

The multivariate extreme value theory concerns the statistical pattern of concurrent extreme values of multiple variables \cite{beirlant2006statistics,haan2006extreme}.  
As a common approach to investigating this pattern, after standardizing the marginal distributions of the variables, one examines the angular distribution of the extreme samples, that is, data points with the largest norms. This angular distribution, under a natural assumption in the theory of multivariate extremes (i.e., the multivariate maximum domain of attraction), approximates a limit distribution on the unit sphere, known as the  {spectral (or angular) measure}.  See Section \ref{sec:background} below for more details.

Given that extremes inherently correspond to a reduced sample size, the challenge of handling high   dimensionality is of heightened importance in this context. As noted in the review article \cite{engelke2021sparse}, many efforts have focused on   employing parsimonious modeling assumptions and techniques to reduce complexity.    A particular parsimonious structure is a discrete spectral measure with a finite number of atoms; that is, the angular distribution of the extreme data points is approximately concentrated on a finite number of directions.      Despite its simplicity,   \cite{fougeres2013dense}     showed that any extremal dependence structure can be arbitrarily well approximated by such a discrete spectral measure.
In addition, a number of parametric models, including heavy-tailed max-linear and sum-linear models (see, e.g., \cite{einmahl2012m}), as well as the recently introduced transformed-linear    model of \cite{cooley2019decompositions}, are known to have a discrete spectral measure.

Recently, as attempts that can also be viewed as providing a parsimonious summary of the angular distribution of multivariate extremes, several authors considered applying clustering algorithms  over the sphere on which the spectral measure resides.   \citet{einmahl2012m} and \citet{janssen2020k} applied the spherical $k$-means algorithm based on cosine dissimilarity \cite{dhillon2001concept} and addressed its relation to the estimation of max-linear factor models.    \citet{fomichov2023spherical} proposed  the spherical $k$-principal-component ($k$-pc) algorithm which is based on a modified cosine dissimilarity, and discussed its superiority in terms of detecting the concentration of the spectral measure on lower-dimensional faces.     \citet{medina2024spectral} considered applying the spectral clustering algorithm \cite{ng2001spectral} to the $k$-nearest neighbor graph constructed from the angular part of the extreme samples, and related it to sum-linear factor models.

As readily observed in the aforementioned works, there is a natural connection between a discrete spectral measure and spherical clustering: Each atom in the spectral measure can be viewed as a cluster center (prototype), and the angular part of the extreme samples form clusters around these atoms. In fact, this intuition has been rigorously explored by \cite{janssen2020k, medina2024spectral}, where consistent recovery of the spectral measure based on their clustering algorithms was established (the consistency result of \cite{janssen2020k} also applies to the $k$-pc clustering of \cite{fomichov2023spherical}).  Since models such as the heavy-tailed max-linear and sum-linear factor models are essentially characterized by the spectral measure, the consistent estimation of spectral measure can be, in principle,  converted to the consistent estimation of parameters of the factor models.

So far, in all the theoretical analysis of the works linking a discrete spectral measure to a clustering algorithm,  the number of atoms, or equivalently speaking, the number of clusters,  is assumed to be known.  We refer to this number as the  {order}, since it also relates to the order of the factor models mentioned above. 
In \cite{janssen2020k,fomichov2023spherical,medina2024spectral},    ad hoc methods such as elbow plot and  scree plot  were used to guide the selection of the order in their real data analysis.  These ad hoc methods are based solely on human visuals to locate the vaguely defined ``elbow'' point, and lack  theoretical justification.

In this work, we further explore clustering-based estimation of multivariate extreme models with a discrete spectral measure.  The contributions of this work are threefold.
 The main contribution involves the development of an order selection method that, on the theoretical side, consistently recovers the true order, and on the practical side, enjoys intuitive and simple implementation. Our method is based on a variant of the well-known silhouette method \cite{rousseeuw1987silhouettes, hruschka2004evolutionary}.  In particular, we introduce an additional penalty term to the so-called simplified average silhouette width,  which discourages small cluster sizes and small dissimilarity between cluster centers.  As a consequence, we provide a method to consistently estimate the order of a max-linear factor model, for which a usual information-based method is not applicable due to the unavailability of likelihood (e.g., \cite{einmahl2012m,yuen2014crps}). 
Our second contribution concerns a large-deviation-type result on the discrete spectral measure estimation via the clustering methods such as  the spherical $k$-means and $k$-pc. This constitutes an attempt to address the quality of convergence for clustering-based estimation in the context of multivariate extremes. 
As a third contribution, we discuss how the discrete spectral measure estimation can be translated into parameter estimates of the heavy-tailed max-linear and sum-linear factor models. Simulation and real-data studies illustrating order selection and factor model estimation are also provided.

The paper is organized as follows. Section \ref{sec:background} provides background and preliminary results on multivariate extremes, spherical clustering, and their connection.  The penalized silhouette method for order selection is introduced in Section \ref{sec:order}. Section \ref{sec:rate} offers some large-deviation-type analysis of convergence of clustering-based spectral estimation.   Section \ref{sec:clust factor models} relates clustering-based spectral estimation to the estimation of certain heavy-tailed factor models. Section \ref{sec:simdata} presents simulation and real-data demonstrations of order selection and factor model estimation.  {Section \ref{sec:Summary} provides a summary and a discussion on potential future directions.} 
By default, all vectors are column vectors. The notation $\delta_\bw$ stands for a delta measure with unit mass at the point $\mbf{w}$ in an appropriate measurable space.

\section{Background}\label{sec:background}

In this section, we  provide some background information on multivariate extreme value theory, spherical clustering,  and their connection.    

\subsection{Multivariate extreme value theory}\label{sec:MEV}

In this section, we review some important elements of multivariate extreme value theory. We refer to \cite{beirlant2006statistics,  haan2006extreme, resnick2007heavy} for more details.

 Suppose that $\mathbf{X}=(X_1,\ldots,X_d)$ is a $d$-dimensional random vector taking values in $[0,\infty)^d$, where $d\ge 2$.  Many discussions in this paper can be extended to the case of $\bb{R}^d$-valued $\mbf{X}$, although for simplicity, we  restrict ourselves to the nonnegative orthant $[0,\infty)^d$, which is also most commonly encountered in practice.  As a conventional practice in the analysis of multivariate extremes, the modeling of marginals and extremal dependence is often separate. We assume that $\mathbf{X}$ has been marginally standardized to share a standard $\alpha$-Pareto-like tail asymptotically: 
 \begin{equation}\label{eq:equiv tail}
 \lim_{x\rightarrow\infty }x^\alpha\Pr(X_1> x)=  \cdots =    \lim_{x\rightarrow\infty }x^\alpha\Pr(X_d> x)=1,
 \end{equation}
  where $\alpha>0$ is known, often chosen as $\alpha=1$ or $\alpha=2$ in literature.  One may generalize \eqref{eq:equiv tail} to include slowly varying functions (see \cite{bingham:1989:regular}), whereas we choose not to do so for simplicity.
  The so-called multivariate regular variation (MRV) assumption  on $\mbf{X}$   requires
\begin{equation}\label{eq:Lambda}
u \Pr\left(  u^{-1/\alpha} \mathbf{X}   \in  \cdot  \right) \overset{v}{\rightarrow}   \Lambda(\cdot), \quad \text{as }u\rightarrow\infty,
\end{equation}
where $\overset{v}{\rightarrow}$ denotes vague convergence of measures  on $\bb{E}_d:= [0,\infty)^d\setminus \{\mathbf{0}\}$, and  $\Lambda$ is an infinite measure on   $ \bb{E}_d$ known as the  {exponent measure}.  For the notion of vague convergence, we follow the formulation of \cite{hult2006regular}   (termed $\cl{M}_0$-convergence there) that does not involve a compactification of $[0,\infty)^d$; see also \cite{kulik2020heavy} (termed $\text{vague}^{\#}$-convergence there).  In particular, convergence \eqref{eq:Lambda} is characterized by convergence at any Borel  subset $E$ of $\bb{E}_d$ whose boundary $\partial E$ is not charged by $\Lambda$ (i.e., $E$ is a  $\Lambda$-continuity set), and which is bounded away from the origin $\mbf{0}$.    

The MRV assumption on  $\mbf{X}$ is equivalent to $\mathbf{X}$ being in the multivariate max-domain of attraction, i.e.,   convergence in distribution of  the normalized component-wise maximum of i.i.d.\ samples of $\mathbf{X}$ towards a multivariate $\alpha$-Fr\'echet distribution with joint distribution function 
\begin{equation}\label{eq:Frechet}
 F_\alpha(\mbf{x}):=\exp\left[-\Lambda( [0,\infty)^d \setminus [\mbf{0},\mbf{x}] )\right],   
\end{equation}
where $[\mbf{0},\mbf{x}]=[0,x_1]\times
\cdots  \times[0,x_d]$, $x_i>0$, $i\in \pc{1,\ldots,d}$.   Moreover, the measure $\Lambda$ satisfies the homogeneity property $\Lambda(c \, \cdot )=c^{-\alpha}\Lambda(\cdot)$, $c>0$, and therefore admits a polar decomposition into a product of radial and angular parts. We shall follow the formulation in \cite[Section 8.2.5]{beirlant2006statistics}, which allows the use of different norms for the radial and angular components.   Suppose that $\| \cdot \|_{(r)}$ and $\| \cdot \|_{(s)}$ denote two arbitrary norms on $\bb{R}^d$.  Slightly abusing the notation, we still  use $\Lambda$  denote the push-forward measure of $\Lambda$ under the one-to-one mapping $\bb{E}_d\mapsto (0,\infty)\times \mathbb{S}_+^{d-1}$,   $\mathbf{x}\mapsto (r,\mathbf{w}=(w_1,\ldots,w_d)):= \left(\|\mathbf{x}\|_{(r)}, \mathbf{x}/\|\mathbf{x}\|_{(s)} \right)$, where
\begin{equation}\label{eq:unit sphere}
\mathbb{S}_+^{d-1}=\{\mbf{x}\in [0,\infty)^d:\  \|\mbf{x}\|_{(s)}=1 \},
\end{equation} 
we have 
\begin{equation}\label{eq:Lambda polar}
\Lambda(dr,d\mbf{w})
= c_{(r)}  \alpha r^{-\alpha-1}dr \times H(d\mathbf{w}),
\end{equation}
where  
\begin{equation}\label{eq:cr}
c_{(r)}=\Lambda(\{\mbf{x}\in  [0,\infty)^d:\ \|\mbf{x}\|_{(r)}\ge 1  \}), 
\end{equation}
and $H$ is a probability measure on $\mathbb{S}_+^{d-1}$ known as  the (normalized)  {spectral measure}. The measure $H$   describes the angular distribution of the concurrence of the extreme values and characterizes the extremal dependence of $\mbf{X}$.   As a consequence of the  marginal standardization, we have 
\begin{equation}\label{eq:stand cond}
\int_{\mathbb{S}_+^{d-1}} \pp{\frac{w_1} {\|\mbf{w}\|_{(r)}}}^\alpha  H(d\mbf{w})=\cdots=\int_{\mathbb{S}_+^{d-1}} \pp{\frac{w_d} {\|\mbf{w}\|_{(r)}}}^\alpha  H(d\mbf{w})=\frac{1}{c_{(r)}}.  
\end{equation} 
In practice, commonly used norms are  the  $p$-norm $\|\mbf{x}\|_p=\pp{\sum_{j=1}^d |x_j|^p}^{1/p}$, $p\in (0,\infty)$, and the sup-norm $\|\mbf{x}\|_\infty=\max\pp{|x_1|,\ldots,|x_d|}$.
 In view of \eqref{eq:Lambda}, the following weak convergence on $\bb{S}_{+}^{d-1}$ holds, as $u\rightarrow\infty$,
\begin{equation}\label{eq:conv spec}
\Pr\pp{  \mbf{X}/\|\mbf{X}\|_{(s)} \in \cdot \mid \|\mbf{X}\|_{(r)}\ge u }  \ConvD H( \cdot)= c_{(r)}^{-1} \Lambda  \pp{ \{\mbf{x}\in \bb{E}_d:\ \mbf{x}/\|\mbf{x}\|_{(s)}\in \cdot \ ,\   \|\mbf{x}\|_{(r)}\ge 1\}} .
\end{equation}

\subsection{Spherical clustering}\label{sec:sph clust}

The spherical clustering algorithms that have been considered so far are performed exclusively on the unit sphere $\bb{S}_+^{d-1}$ with respect to the $2$-norm (Euclidean norm), that is, take $\|\cdot \|_{(s)}$ in \eqref{eq:unit sphere} as $\|\cdot \|_2$. We do not make this assumption for generality unless discussing specific examples.
We equip $\bb{S}_+^{d-1}$ with the subspace topology inherited from $\bb{R}^d$.  
Next, we introduce a  {dissimilarity measure} $D$ on $\bb{S}_+^{d-1}$ that is essential for clustering. 

 \begin{Def}\label{Def:dis}
A  {dissimilarity measure} $D$ on $\bb{S}_+^{d-1}$ is a bivariate function $D:\bb{S}_+^{d-1}\times \bb{S}_+^{d-1}\rightarrow [0,1]$ that is continuous and satisfies the following properties:   
(i) $D(\mbf{w}_1,\mbf{w}_2)=0$ if and only if $\mbf{w}_1=\mbf{w}_2$, and  (ii) $D(\mbf{w}_1,\mbf{w}_2)=D(\mbf{w}_2,\mbf{w}_1)$, where $\mbf{w}_i\in \bb{S}_+^{d-1}$, $i\in \pc{1,2}$.
\end{Def}

\begin{Rem}\label{Rem:D property}
 Without loss of generality, we shall assume that $D$ is properly normalized so that $D$ is surjective over $[0,1]$. 
   A nonnegative function $D$ satisfying (i) and (ii) is often referred to as a  {semimetric}, which  lacks the triangular inequality axiom  of a metric.
   With the assumptions imposed,  we have $\mbf{w}_n\rightarrow \mbf{w}$  on $\dsph$ if and only if $D(\mbf{w}_n,\mbf{w}) \rightarrow 0$ as $n\rightarrow\infty$, and
  the $D$-neighborhoods 
  $$B_D(\mbf{w},r):=\{\mbf{u}\in \dsph: \ D(\mbf{w},\mbf{u})<r   \},$$ $\mbf{w}\in \dsph, r>0$,  form a topological basis of $\dsph$; see, e.g.,  \cite{wilson1931semi,galvin1984completeness}.  Note that due to the compactness of $\dsph$ and the  continuity of $D$, the function    
  \begin{equation}\label{eq:D*}
  D^{\dagger}(\mbf{w}_1,\mbf{w}_2):=\sup_{\mbf{w}\in \dsph}|D(\mbf{w}, \mbf{w}_1 )-D(\mbf{w}, \mbf{w}_2 )| 
\end{equation}
  is also a semimetric that is continuous on $\dsph\times \dsph$ and maps surjectively   to $[0,1]$, which we refer to as the  {dual} of $D$.   Following from its definition,  we have $D^{\dagger}\ge D$, and   a   triangular-like inequality holds:
\begin{equation}\label{eq:ps tri}
 D(\mbf{w}_1, \mbf{w}_3 )\le D(\mbf{w}_1, \mbf{w}_2 ) + D^{\dagger} \pp{\mbf{w}_2,\mbf{w}_3}.
\end{equation}
Additionally, we have the following property: if $D(\mbf{w}_{n},\mbf{w})\to 0$ as $n\rightarrow\infty$, then $D^{\dagger}(\mbf{w}_{n},\mbf{w})\to 0$ as well.
\end{Rem}

  Some common dissimilarity measures are only semimetrics but not metrics.  Below,  we consider $\|\cdot \|_{(s)}=\|\cdot\|_{2}$ so that $\dsph$ is the $2$-norm sphere. The cosine dissimilarity adopted in the  spherical $k$-means of \cite{dhillon2001concept,janssen2020k} is given by 
\begin{equation}\label{eq:cos dis}
    D_{\cos}(\mbf{w}_1,\mbf{w}_2) =1-\mbf{w}_1^{\top} \mbf{w}_2,  
\end{equation}
where $\mbf{w}_1, \mbf{w}_2\in \dsph \subset \bb{R}^d$ .   The dissimilarity measure corresponding to the $k$-pc algorithm of \cite{fomichov2023spherical}  is given by
\begin{equation}\label{eq:pc dis}
D_{\mathrm{pc}}(\mbf{w}_1,\mbf{w}_2) =  1-  \pp{ \mbf{w}_1^{\top}\mbf{w}_2}^2. 
\end{equation}
These two dissimilarity measures enjoy computational advantages, although  neither of them is a metric. 
Note that  since $\left|\pp{\mbf{w}_1^{\top} \mbf{w}_2}^2-\pp{\mbf{w}_1^{\top} \mbf{w}_3}^2\right|\le 2|\mbf{w}_1^{\top} \mbf{w}_2-\mbf{w}_1^{\top} \mbf{w}_3|\le 2\|\mbf{w}_2-\mbf{w}_3\mbf\|_2$,  $\mbf{w}_i\in \dsph \subset \bb{R}^d$, one obtains a bound for the dual semimetric as $D^\dagger\pp{\mbf{w}_2,\mbf{w}_3}\le c \|\mbf{w}_2-\mbf{w}_3\|_2 $ for $D=D_{\mathrm{cos}}$ or $D_{\mathrm{pc}}$, with constant $c=1$ or $2$  respectively.

To simplify the mathematical description of clustering of sample data, it is convenient to use the notion of  {multiset} (e.g., \cite{kettleborough2013optimising}).
 Recall that a multiset $W$ on $\dsph$ is a set that allows repetition of its elements, whose support, denoted as $\supp{W}$, is a subset of $\dsph$  in the usual sense that eliminates repetitions in $W$. For instance,  with two distinct points $\bw_1$ and $\bw_2$ on $\dsph$, one can have $W=\{\mbf{w}_1,\mbf{w}_1,\mbf{w}_2\}$  with $\supp{W}=\{\bw_1,\bw_2\}$. A multiset $W$ can be characterized by the multiplicity function $m_W:\dsph\mapsto\{0,1,\ldots\}$,  where $m_W(\mbf{w})$ equals the number of repetitions of element $\mbf{w}\in \dsph$ ($m_W(\mbf{w})=0$ if $\mbf{w}\notin \supp{W}$).     A subset of $\dsph$ in the usual sense can be understood as a multiset with the multiplicity taking value either $0$ or $1$, with the empty set corresponding to a multiplicity function that is identically $0$.   When the notation $\mbf{w}\in W$ is used for a multiset $W$, it means that $\mbf{w}$ is an element in $\supp{W}$. 
 For multisets $W_1,W_2$ with multiplicity functions $m_1$ and $m_2$ respectively, their union $W_1\cup W_2$ is given by the multiset characterized by the multiplicity function $m_{1}\vee m_{2}$, and their intersection $W_1\cap W_2$ is given by the multiset characterized by $m_{1}\wedge m_{2}$. The relation $W_1\subset W_2$ is understood as $m_1\le m_2$. Furthermore, if $\supp{W}$ is a finite set, a summation $\sum_{\mbf{w}\in W}f(\bw )$ for a suitable function $f$ is understood as $\sum_{\mbf{w}\in \supp{W}}f(\bw ) m_W(\mbf{w})$. For example, the cardinality of $W$ is defined as 
 $$|W|= \sum_{\bw \in \supp{W}} m_W(\mbf{w}).
 $$ 
Also we write $D(\mbf{w},W)= \inf_{\mbf{s}\in \supp{W}} D(\mbf{w},\mbf{s})$.  

Now suppose $W$ is a multiset on $\bb{S}_+^{d-1}$ with cardinality $|W|<\infty$.   Suppose $k\in \bb{Z}_+$ and $k\le |W|$.   Let $A_k^*=\pc{\mbf{a}_1^* ,\ldots,\mbf{a}_k^*}$  be a multiset  on  $\dsph$ with cardinality $k$,  which satisfies 
\begin{equation}\label{eq:k-clust opt}
\sum_{\mbf{w}\in W} D\pp{\mbf{w},A_k^*}= \inf\left\{\sum_{\mbf{w}\in W} D(\mbf{w},A): \  \supp{A}\subset \bb{S}_+^{d-1}, \  |A|=k  \right\}.
\end{equation}
 The existence of $A_k^*$ is guaranteed by the continuity of $D$ and the compactness of $\bb{S}_+^{d-1}$, although it may not be unique. Notice that when $|\supp{W}|\ge k$, the infimum in \eqref{eq:k-clust opt} must be achieved with a distinct set of $\mbf{a}_i^*$'s.
  Below when multisets $C_1,\ldots,C_k$ with multiplicity functions $m_1,\ldots,m_k$  are said to form a partition of a multiset $W$ with multiplicity function $m$, it means that $m=m_1+\cdots+m_k$, and $m_i\neq 0$ for all $i\in\{1,\ldots,k\}$. 
\begin{Def}\label{Def:k-clust}
A   {$k$-clustering} of a multiset  $W$, $1\le k\le |W|$, with respect to the dissimilarity measure $D$ 
 refers to a pair $(A_k^*, \mathfrak{C}_k)$. Here $A_k^*=\pc{\mbf{a}_1^* ,\cdots,\mbf{a}_k^*}$ is as described above, and $\mathfrak{C}_k=\{C_1,\ldots, C_k\}$ is a partition of $W$ into a collection of       
 multisets $C_i$'s  such that $D\pp{\mbf{w}, A_k^* }=D\pp{\mbf{w}, \mbf{a}_i^*}$ for all $\mbf{w}\in C_i$, $i\in \pc{1,\ldots,k}$.  We refer to  $A_k^*$ as the  {set of centers}  and each $C_i$ as a cluster.
\end{Def}
\begin{Rem}\label{Rem:nonemp clust}
A $k$-clustering of $W$ always exists, although it may not be unique even when $A_k^*$ is unique: there may be points in $\supp{W}$ with the same $D$-dissimilarity to multiple centers. On the other hand, it is always possible to ensure non-emptiness of each cluster $C_i$ when $k\le |W|$. 
\end{Rem}
With the choices $D=D_{\cos}$ and $D_{\rm{pc}}$ in \eqref{eq:cos dis} and \eqref{eq:pc dis}, respectively, a $k$-clustering corresponds to the spherical $k$-means and $k$-pc clustering of \cite{dhillon2001concept} and \cite{fomichov2023spherical}, respectively.   Solving a $k$-clustering problem can be computationally hard,  and typically, the solution can only be approximated by a heuristic algorithm such as a Lloyd-type iterative algorithm as in \cite{dhillon2001concept} and \cite{fomichov2023spherical}. In the theoretical analysis of this paper, we assume that a $k$-clustering can be found accurately. 
In addition,   when $W$ is a random multiset,  for a $k$-clustering of $W$, we assume that  the  elements in $A_k^*$ and the labels $\ind{\mbf{w}\in C_j}$, $\mbf{w}\in W$, $j\in\{1,\ldots,k\}$,  are   measurable.

\subsection{Spherical clustering for multivariate extremes}\label{eq:sph clust MEV}\label{sec:sph clust MEV}

We follow \cite{janssen2020k} and \cite{fomichov2023spherical} to relate the spherical clustering to the analysis of multivariate extremes.  Suppose that $\pp{\mbf{X}_1,\ldots,\mbf{X}_n}$, $n\in \bb{Z}_+$, are i.i.d.\ samples of $\mbf{X}$, which is marginally standardized and regularly varying on $\bb{E}_d$ with spectral measure $H$ on $\bb{S}_+^{d-1}$ as assumed in Section \ref{sec:MEV}. We shall also follow the notation introduced in the same section.  Let $\ell_n$ be an intermediate sequence satisfying $\ell_n\rightarrow \infty$ and $\ell_n/n\rightarrow 0$ as $n\rightarrow\infty$.   Introduce  a  multiset on $\bb{S}_+^{d-1}$ representing the extremal subsample:
\begin{equation}\label{eq:W_n}
{W}_n=\left\{\mbf{X}_i/\|\mbf{X}_i\|_{(s)}: \  \|\mbf{X}_i\|_{(r)}\ge 
 \pp{n/ \ell_n}^{1/\alpha}, \ i\in \pc{1,\ldots,n}  \right\}.
\end{equation}
In  words, the  extremal subsample is selected by sample points with largest $\|\cdot \|_{(r)}$ norms   projected  onto the $\|\cdot \|_{(s)}$-norm sphere $\dsph$.
The choice of $\ell_n$ and the regular variation assumption  together imply 
\begin{equation}\label{eq:W_n exp val}
\E |W_n| = n\Prt{ \|\mbf{X}_n\|_{(r)} \ge \pp{ n/\ell_n}^{1/\alpha} }\sim \ell_n  c_{(r)}\rightarrow\infty
\end{equation}
as $n\rightarrow\infty$, where $ c_{(r)}$ is in \eqref{eq:cr}.  Notice that the set of the form $ \{\mbf{x}\in \bb{E}_d:\  \|\mbf{x}\|_{(r)}\ge x\}$, $x>0$, is always a $\Lambda$-continuity set due to the homogeneity of $\Lambda$.  Then by a triangular-array version of  the Strong Law of Large Numbers (see, e.g., \cite{hsu1947complete}), we have
\begin{equation}\label{eq:LLN W_n}
|W_n|/\ell_n  = \frac{1}{\ell_n} \sum_{i=1}^n \ind{  \|\mbf{X}_i\|_{(r)}\ge  \pp{n/\ell_n}^{1/\alpha} }  \rightarrow \Lambda\pp{ \{\mbf{x}\in \bb{E}_d:\  \|\mbf{x}\|_{(r)}\ge 1\} } 
 =c_{(r)} 
\end{equation}  
almost surely as $n\rightarrow\infty$.

Next,  define the following empirical spectral measure on $\bb{S}_+^{d-1}$ as
\begin{equation}\label{eq:emp spec}
H_n=  \frac{1}{|W_n|}\sum_{\bw\in W_n} \delta_\bw,
\end{equation}
where $H_n$ is understood as a zero measure if $|W_n|=0$.
 Then we have the following basic consistency result.
 \begin{Pro}\label{Pro:basic consist}
   Suppose $\mbf{X}$ satisfies conditions \eqref{eq:equiv tail} and    \eqref{eq:Lambda} with a spectral measure $H$ on $\bb{S}_+^{d-1}$ as defined in \eqref{eq:Lambda polar}. Let $W_n$ denote the extremal subsample as in \eqref{eq:W_n} and $H_n$ denote the empirical spectral measure as in \eqref{eq:emp spec}. Then for any $S$ that is a $H$-continuity Borel subset of $\bb{S}_{+}^{d-1}$, we have $H_n(S)\rightarrow H( S)$ almost surely as $n\rightarrow\infty$. 
 \end{Pro}
 \begin{proof}[\textbf{\upshape Proof:}]
  It follows from a triangular-array  Strong Law of Large Numbers with the relations
\eqref{eq:Lambda},   \eqref{eq:conv spec},     \eqref{eq:W_n exp val} and \eqref{eq:LLN W_n}. 
 \end{proof}

Now we consider applying the $k$-clustering in Definition \ref{Def:k-clust} to the random subsample ${W}_n$.  
In view of Proposition \ref{Pro:basic consist}, the distribution of $W_n$ on $\bb{S}_{+}^{d-1}$ serves as a good approximation to $H$. When $H$ is a discrete measure with finitely many atoms, clustering can be applied to $W_n$ in such a way that the cluster centers closely approximate the locations of the atoms, and the cluster proportions provide accurate estimates of the associated probabilities. Furthermore, these estimators are shown to be consistent, as demonstrated in the following Corollary.
 
 \begin{Cor}\label{Cor:disc consist}
  Suppose $\mbf{X}$  is as in Proposition \ref{Pro:basic consist}, and furthermore,  has  a   spectral measure of the following form:
\begin{equation}\label{eq:disc spec}
    H= \sum_{i=1}^k p_i \delta_{\mbf{a}_i},
\end{equation}
where
$\mbf{a}_i$'s are distinct points on $\dsph$,  and $p_i>0$, $p_1+\cdots+p_k=1$. 
Let $W_n$ denote the extremal subsample as in \eqref{eq:W_n},  and  $\pp{{A}_{k,n}=\pp{{\mbf{a}}_{1,n}^k,\ldots,{\mbf{a}}_{k,n}^k} ,\mathfrak{C}_{k,n}=\pc{{C}_{1,n}^k,\ldots, {C}_{k,n}^k}}$ form a $k$-clustering of $W_n$ as defined in Definition \ref{Def:k-clust} with respect to a dissimilarity measure $D$ defined in Definition \ref{Def:dis}. Let 
\begin{equation}\label{eq:est p}
  p_{i,n}^k=\frac{|C_{i,n}^k|}{|W_n|},
\end{equation}  
if $|W_n|>0$, and set $p_{i,n}^k$ as $0$ if $|W_n|=0$. 
Then there exist  bijections $\pi_n: \{1,\ldots,k\}\mapsto \{1,\ldots,k\}$, $n\in \bb{Z}_+$, such that 
\[
\mbf{a}_{\pi_n(i),n}^k\rightarrow \mbf{a}_i,\quad  p_{\pi_n(i),n}^k \rightarrow p_{i},\quad i\in \{1,\ldots,k\},
\]
almost surely.
\end{Cor}
\begin{proof}[\textbf{\upshape Proof:}]
The convergence of $\mbf{a}_{\pi_n(i),n}^k$ follows from   \cite[Theorem 3.1]{janssen2020k} (stated as convergence in Hausdorff distance between ${A}_{k,n} $ and $\{\mbf{a}_1,\ldots,\mbf{a}_k\}$)  and Proposition \ref{Pro:basic consist} above; see also the discussion in   \cite[Section 4]{janssen2020k}). 
It remains to show  the convergence of  $p_{\pi_n(i),n}^k$, $i\in \{1,\ldots,k\}$.  Set 
\begin{equation}\label{eq:r_A}
r_A=\sup \pc{ r>0:    B(\mbf{a}_i,r),\  i\in \pc{1,\ldots,k},   \text{ are  disjoint} }>0.
\end{equation}    
Fix  $\epsilon\in (0,r_A/3)$. By what has been proved and the continuity of $D^\dagger$, at almost every outcome $\omega$ of the sample space $\Omega$,   for $n$ sufficiently large,  we have the dual dissimilarity $D^\dagger(\mbf{a}_{\pi_n(i),n}^k,\mbf{a}_{i})<\epsilon $ , $i\in \{1,\ldots,k\}$.  Fix for now such an $\omega$ and let $n$ be sufficiently large (possibly depending on $\omega$).  Then  by the triangular inequality \eqref{eq:ps tri}, we have   $B_D(\mbf{a}_i,\epsilon)\subset B_D\pp{\mbf{a}_{\pi_n(i),n}^k,2\epsilon} \subset B_D(\mbf{a}_i,3\epsilon)$, $i\in \{1,\ldots,k\}$. Note that $ B_D\pp{\mbf{a}_{\pi_n(i),n}^k,2 \epsilon}\cap W_n$ are disjoint for $i\in \{1,\ldots,k\}$ due to the choice of $\epsilon$. So in view of Definition \ref{Def:k-clust},  we have  
$ B_D(\mbf{a}_i,\epsilon) \cap W_n \subset B_D\pp{\mbf{a}_{\pi_n(i),n}^k,2 \epsilon}\cap W_n  \subset  C_{\pi_n(i),n}^k\subset \pp{ B_D(\mbf{a}_i,\epsilon) \cup  \cap_{j\neq i} B_D(\mbf{a}_j,\epsilon)^c} \cap  W_n $, and hence
\[
H_n\pp{B_D(\mbf{a}_i,\epsilon)}\le p_{\pi_n(i),n}^k\le H_n\pp{B_D(\mbf{a}_i,\epsilon) \cup  \cap_{j\neq i} B_D(\mbf{a}_j, \epsilon)^c}, \quad i\in \{1,\ldots,k\}.
\]
Recall the inequality above is derived for almost every $\omega$ with $n$ sufficiently large.
The conclusion then follows from Proposition \ref{Pro:basic consist} since both sides above converges almost surely to $p_i$ as $n\rightarrow\infty$.
\end{proof}

\begin{Rem}
Comparing   \cite[Proposition 3.3]{janssen2020k} with   Proposition \ref{Pro:basic consist} and Corollary \ref{Cor:disc consist} here, we have chosen to work directly under the marginal standardization assumption in \eqref{eq:equiv tail} and not to treat the empirical marginal transformations as in \cite[Eq.\ (3.5)]{janssen2020k} for simplicity.  Nevertheless, the consistency result of the order selection below (Theorem \ref{Thm:sil cons})  can be extended to the setup of \cite{janssen2020k} based on  the results there.
 \end{Rem}

\section{Order selection via penalized silhouette \label{sec:order}}
\subsection{The method}

Following the notation and setup in Section \ref{sec:sph clust}, suppose $W$ is a multiset on $ \bb{S}_+^{d-1}$ and $1\le k\le |  W|<\infty$.   Let $\pp{A_k^*=\pc{\mbf{a}_1^* ,\ldots, \mbf{a}_k^*}, 
 \mathfrak{C}_k=\pc{C_1,\ldots,C_k}}$ be a $k$-clustering of $W$   with respect to 
a dissimilarity measure $D$ as in Definition \ref{Def:k-clust}.  
  Define for $\mbf{w}\in W$ that  
\[
  a(\mbf{w}) =D\pp{\mbf{w}, A_k^*} , \quad   b(\mbf{w})= \bigvee_{i=1}^k D\pp{\mbf{w}, A_k^*\setminus\pc{\mbf{a}_i^*}},  \]    which are respectively the dissimilarities of $\mbf{w}$  to the closest center (i.e., the center of the cluster it belongs to) and  to the second closest center.  When $k=1$. we understand $b(\mbf{w})=1$.  
  The (simplified) average silhouette width  (ASW)  \cite{hruschka2004evolutionary} of this $k$-clustering is then defined as
 \begin{equation}\label{eq:ASW}
\bar{S}= \bar{S}\left(W;  A_k^* \right) = 
\frac{1}{|W|}\sum_{\mbf{w}\in W} \frac{b(\mbf{w})-a(\mbf{w})}{ b(\mbf{w}) }= 1-  \frac{1}{|W|} \sum_{\mbf{w}\in W} \frac{a(\mbf{w})}{b(\mbf{w})}.
\end{equation}
 
A well-clustered dataset is expected to have small  $ a(\mbf{w})$ values relative to  $b(\mbf{w})$ across the majority of $\mbf{w}$ points.   
    Hence, one often uses    $\bar{S}$ to guide the selection of the number of clusters, that is, to choose $k$ which maximizes $\bar{S}$.  However, when   experimenting applying the ASW to multivariate extremes with a discrete spectral measure as described in Section \ref{sec:sph clust MEV}, the performance is   unsatisfactory: it tends to respond insensitively when the number of clusters exceeds the true $k$, i.e., the number of atoms of the spectral measure;   see, for example, the curve corresponding to $t=0$ in  Fig.\ \ref{kmeanelbow}. In particular, we observe   two behaviors of ASW that lead to the issue: 1) it tends to treat a tiny fraction of isolated points as a cluster;   2)  it sometimes splits a single cluster center into multiple centers that are close to each other.   

Motivated by these observations, we propose to introduce a penalty term that discourages small cluster size and small $D$ dissimilarity between centers. 
Recall that for a $k$-clustering $\pp{A_k^*=\pc{\mbf{a}_1^* ,\ldots, \mbf{a}_k^*},\mathfrak{C}_k=\pc{C_1,\ldots,C_k}}$ of a multiset $W$, the set $\pc{\mbf{a}_1^* ,\ldots, \mbf{a}_k^*}$ represents the cluster centers and $\pc{C_1,\ldots,C_k}$ denotes the partition of $W$. Then $\min_{i=1,\ldots,k} \pp{  |C_i|/( |W|/k) }$ is the smallest cluster size  relative to the average cluster size $ |W|/k$, and  $\min_{1\le i< j\le k}D\pp{\mbf{a}_i^*,\mbf{a}_j^*}$ is the smallest pairwise dissimilarity between centers. Note that both quantities are between $0$ and $1$ (recall that $D$ maps onto $[0,1]$), and we want to introduce a criterion that penalizes either of them being small.
There is arguably some arbitrariness in designing this penalty criterion. Through some mathematical heuristics and extensive experiments, we find that the following penalty works relatively well. 
Let $t\ge 0$ be a tuning parameter.   Set   
\begin{equation}\label{eq:pen}
P_t=P_t\left(W;   A_k^*, \mathfrak{C}_k\right)=1- \pp{  \min_{i=1,\ldots,k} \pp{   \frac{   |C_i|}{ |W|/k} } }^t\pp{ \min_{1\le i< j\le k}D\pp{\mbf{a}_i^*,\mbf{a}_j^*} }^{t},
\end{equation}
where   $\min_{1\le i< j\le k}D\pp{\mbf{a}_i^*,\mbf{a}_j^*}$ is understood as $1$ when $k=1$.
Then we form the  {penalized ASW} defined by \begin{equation}
   S_t= S_t(W; A_k^*, \mathfrak{C}_k)= \bar{S}-P_t  =\pp{  \min_{i=1,\ldots,k}  \pp{ \frac{ |C_i|}{  |W|/k }  }}^t\pp{\min_{1\le i< j\le k}D\pp{\mbf{a}_i^*,\mbf{a}_j^*) }}^{t}- \frac{1}{|W|} \sum_{\mbf{w}\in W} \frac{a(\mbf{w})}{b(\mbf{w})}.
\end{equation}
 Notice that when $t=0$, we have $P_0=0$ and hence $S_0=\bar{S}$. As $t>0$ increases, the penalty $P_t$ increases and hence $S_t$ decreases.

 We have the following consistency result regarding applying the penalized ASW for order selection for a multivariate extreme model with a discrete spectral measure. We follow the notation in Section \ref{sec:sph clust MEV}.
 \begin{Thm}\label{Thm:sil cons}
Suppose $\mbf{X}$ satisfying \eqref{eq:equiv tail} and \eqref{eq:Lambda} has a discrete spectral measure of the form $ H= \sum_{i=1}^k p_i \delta_{\mbf{a}_i}$, $k\in\bb{Z}_+$
where $\mbf{a}_i$'s are distinct points on $\dsph$,  and $p_i>0$, $p_1+\cdots+p_k=1$. Let $W_n$ denote the extremal subsample as in \eqref{eq:W_n} with $\ell_n\rightarrow\infty$, $\ell_n/n\rightarrow 0$,  and  $\pp{A_{m,n},\mathfrak{C}_{m,n}}$, $m\in \bb{Z}_+$, form an $m$-clustering of $W_n$ as defined in Definition \ref{Def:k-clust} with respect to a dissimilarity measure $D$ defined in Definition \ref{Def:dis}. 
Let $r_A$ be defined as in \eqref{eq:r_A}
and define
\begin{equation}\label{eq:pmin}
 p_{\min}=
 \min_{1\le i \le k} p_i.
 \end{equation} 
Then for any $t\in (0,t_0)$, where 
\begin{equation}\label{eq:t0}
 t_0:=\ln\pp{1-r_A p_{\min}}/\ln\pp{r_Akp_{\min}},
 \end{equation} 
we have
\[
 \liminf_n \pc{ S_t\pp{W_n; A_{k,n},\mathfrak{C}_{k,n} } - S_t\pp{W_n; A_{m,n}, \mathfrak{C}_{m,n}}}\ge \Delta_t 
\]
almost surely for any $m\neq k$, where $\Delta_t:= \pp{r_A k p_{\min}}^t- 1+r_A p_{\min} >0 $ when $t\in (0,t_0)$. 
\end{Thm}

The theorem implies that as long as the tuning parameter is in an appropriate range, with probability tending to $1$ as $n\rightarrow\infty$,  the true order $m=k$  uniquely maximizes the penalized ASW.  
The proof of Theorem \ref{Thm:sil cons} can be found in Section \ref{sec:proof consist}.  In Proposition \ref{Pro:sil con rate} below, we will provide a rate of how fast the probability of false order selection decays to zero.

In practice, we suggest plotting the penalized ASW $S_t$  as a function of $m\in\{1,2,\ldots\}$ for a range of small $t$ values. {The idea is to start with $t$ near $0$, gradually increase it, and see how the penalized ASW curve responds. If some obvious spurious clusters (i.e., those with tiny size or centers that are too close together) are present, the curve tends to respond sensitively and bends at the appropriate order.  We then identify the turning point $m$  as the choice of the order $k$.}
   As a quick illustration, we follow a simulation setup of $(d=6,k=6)$ described in Section \ref{sec:sim} below to simulate a max-linear factor model (Section \ref{sec:factor models}).  See Fig.\ \ref{kmeanelbow}.  {Increasing $t$ to a very large value is not informative and is  not recommended in practice.} On the other hand, it would be desirable to develop a   data-driven method for choosing $t$, which we leave for a future work to explore.

\begin{figure}[h]
\centering
  \includegraphics[width=0.9\linewidth]{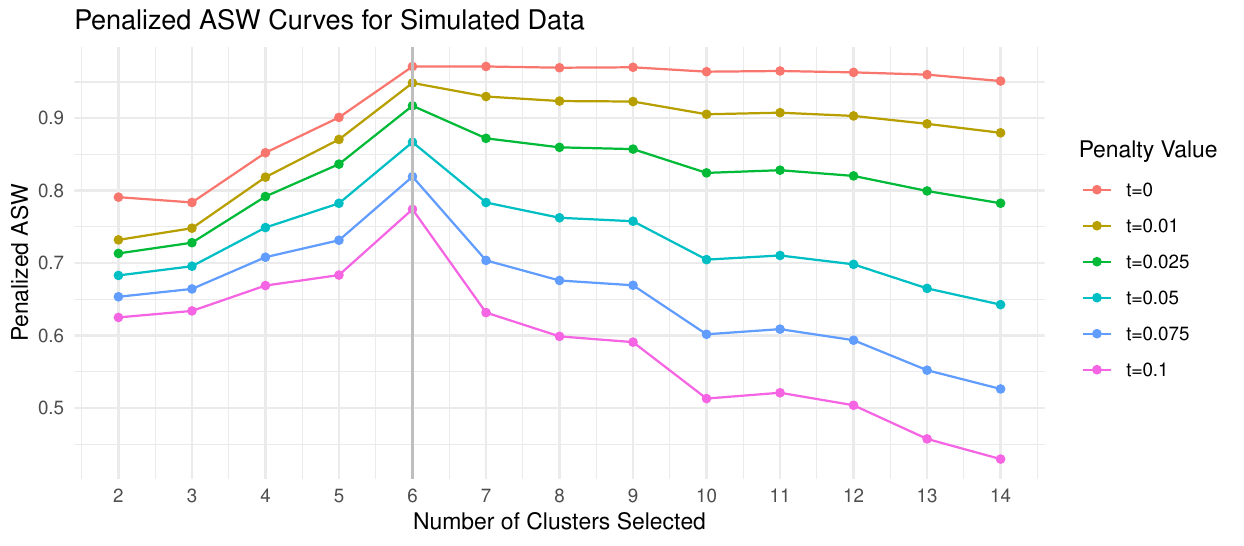}
  \caption{A simulation instance taken from  Section \ref{sec:sim} $d=6$, $k=6$ setup. Penalized Average Silhouette Width (ASW)  $S_t$ (vertical axis) for  spherical $k$-means clustering is plotted as a function of test order $m$ (horizontal axis). The different penalty values of $t$ are illustrated by different colors.  The true discrete  spectral measure in \eqref{eq:disc spec} is  given by $(\mbf{a}_1, p_1 )=((0.29, 0.21, 0.50, 0.45, 0.43, 0.49)^\top, 0.22)$, $(\mbf{a}_2, p_2 ) =((0.74, 0.00, 0.59, 0.00, 0.32, 0.00)^\top, 0.10)$, $(\mbf{a}_3, p_3 ) =((0.00, 0.27, 0.00, 0.47, 0.00, 0.84)^\top, 0.13)$, $(\mbf{a}_4, p_4 )= ((0.33, 0.70, 0.63, 0.00, 0.00, 0.00)^\top,0.14 )$, $(\mbf{a}_5, p_5 ) =((0.00, 0.00, 0.00, 0.81, 0.47, 0.34)^\top, 0.09)$, $(\mbf{a}_6, p_6 ) =((0.48, 0.49, 0.25, 0.33, 0.53, 0.29)^\top, 0.32)$.}\label{kmeanelbow}
\end{figure}

\subsection{Consistency of order selection via penalized silhouette}\label{sec:proof consist}

In this section, we prove Theorem \ref{Thm:sil cons}.

\subsubsection{Some deterministic estimates}
We first prepare some deterministic estimates regarding the $k$-clustering in Definition \ref{Def:k-clust} and ASW in \eqref{eq:ASW}.  
Based on Proposition \ref{Pro:basic consist}, if the true spectral measure $H$ consists of finitely many atoms, and we select the subset $S$ in the proposition as a union of neighborhoods around each atom, then, for sufficiently large $n$, almost all points in the extremal subsample $W_n$ in \eqref{eq:W_n} will lie near one of these atoms with respect to the dissimilarity measure $D$. This motivates us to address a scenario where we aim to cluster a multiset where the majority of points are concentrated around a finite number of centers.


 Fix throughout this section  distinct points $\pc{\mbf{a}_1,\ldots,\mbf{a}_k}=:A$  on $\dsph$,  $k\in \bb{Z}_+$,  $p_i>0$ with $p_1+\cdots+p_k=1$, and  a dissimilarity measure $D$ defined in Definition \ref{Def:dis}.
We say a multiset $W$ on $ \dsph$ with $|W|<\infty$ satisfies the concentration condition $\mathcal{A}(\epsilon,\delta )$   with $\epsilon,\delta>0$, if  $k\le |W|\wedge \delta^{-1}$ such that $\epsilon\in (0,r_A)$ with $r_A$ in \eqref{eq:r_A}, $\delta \in \pp{0, p_{\min}  }$ with $p_{\min}$ in \eqref{eq:pmin}, and 
\begin{equation}\label{eq:Ass concen bound}
 \frac{|W \cap B_D(\mbf{a}_i,\epsilon)|}{|W|}\ge  p_i-\delta  ,  \quad  i\in \pc{1,\ldots,k}.
 \end{equation}


In proving the following lemmas, we need to control the dual dissimilarity $D^{\dagger}$  of points within the same $D$-neighborhood. Therefore, we introduce the following uniform upper bound.
  Set for the aforementioned $A=\pc{\mbf{a}_1,\ldots,\mbf{a}_k}$ and $s>0$ that 
   \begin{equation} \label{rad}
   r_A^\dagger(s)=\sup \pc{ D^{\dagger}(\mbf{a}_i,\mbf{w}):  i\in \pc{1,\ldots,k},\  \mbf{w}\in B_D(\mbf{a}_i,s ) }.
 \end{equation}  
\begin{Rem}\label{Rem:r_A dagger}
Note that $r_A^\dagger(s)>0$ for any $s>0$, and $r_A^\dagger(s)\rightarrow 0$ as $s\rightarrow 0$; see Remark \ref{Rem:D property}. For sufficiently small $\epsilon$, we always have $r_A^\dagger(\epsilon)<r_A$. 
\end{Rem}


If $W$ satisfies the concentration condition $\mathcal{A}(\epsilon,\delta )$ but is partitioned into fewer clusters than $k$, then at least two of the true clusters will be merged into a single cluster. This merge can be effectively identified by the ASW based on the reduced separability, as demonstrated in the next lemma.


\begin{Lem}\label{Lem:m<k ASW}
  Suppose a multiset $W$ satisfies the condition $\mathcal{A}(\epsilon,\delta )$ and $1\le  m<k$. Let  $(A^*_m,\frak{C}_m)$ be an $m$-clustering of $W$ as defined in Definition \ref{Def:k-clust}. Then the (unpenalized) ASW $\bar{S}$ satisfies
    \[
\bar{S}=\bar{S}(W;A^*_m,\frak{C}_m)\le 1- \pp{  p_{\min} -\delta}\pp{r_A-r_A^\dagger(\epsilon)},
    \]
    where $r_A^\dagger$ is as in \eqref{rad}.
\end{Lem}

\begin{proof}[\textbf{\upshape Proof:}]
Since $m<k$ and  $B_D(\mbf{a}_i,r_A)$'s are disjoint, $i\in \pc{1,\ldots,k}$,  there exists $\ell \in \{1,\ldots,k\}$ such that $B_D(\mbf{a}_\ell ,r_A )\cap A_m^* =\emptyset$. Hence for any $\mbf{w}\in W\cap B_D(\mbf{a}_\ell,\epsilon)$, we have by the triangular inequality \eqref{eq:ps tri} that 
\[
a(\mbf{w})=D\pp{\mbf{w}, A_m^*} \ge D\pp{\mbf{a}_\ell, A_m^*} - D^{\dagger}\pp{\mbf{w},\mbf{a}_\ell}\ge r_A - r_A^\dagger(\epsilon). 
\]  
Then since $b(\mbf{w})\le 1$, we have
\begin{align*}
    \frac{1}{|W|} \sum_{\mbf{w}\in W} \frac{a(\mbf{w})}{b(\mbf{w})} \ge \frac{1}{|W|} \sum_{\mbf{w}\in W}  a(\mbf{w})    \ge    \frac{|W\cap B_D(\mbf{a}_\ell,\epsilon) |}{|W|}  \pp{r_A-r_A^\dagger(\epsilon)}  \ge  \pp{  p_{\min}  -\delta}\pp{r_A-r_A^\dagger(\epsilon)},    
\end{align*}
which implies the desired result.
\end{proof}

The next lemma states that if $W$ satisfies the concentration condition $\mathcal{A}(\epsilon,\delta )$ and is partitioned into a  number of clusters that is greater than or equal to $k$, then for each true center of $W$, there exists at least one cluster center that is close to it with respect to the dissimilarity measure $D$.
\begin{Lem}\label{Lem:m>=k center}
      Suppose a multiset $W$ satisfies the condition $\mathcal{A}(\epsilon,\delta )$. Let $\pp{A^*_m=\{\mbf{a}_1^*,\ldots,\mbf{a}_m^*\},\frak{C}_m}$ be an $m$-clustering of $W$ as defined in Definition \ref{Def:k-clust}, $m\ge k$.  Then for any $i\in \{1,\ldots,k\}$, there exists $j\in \{1,\ldots,m\}$, such that $D\left(\mbf{a}_j^*,\mbf{a}_i\right)<\epsilon'$, where
\begin{equation}\label{eq:epsilon prime}
      \epsilon'=\epsilon'(\epsilon,\delta)=  \frac{ (1-k\delta ) \epsilon +k\delta  }{p_{\min} -\delta} +  r_A^\dagger(\epsilon).
\end{equation}
  In particular, when $m=k$ and $\epsilon'<r_A$, there exists a bijection $\pi:\{1,\ldots,k\}\mapsto \{1,\ldots,k\}$, such that $D\left(\mbf{a}_{\pi(i)}^*,\mbf{a}_i\right) <\epsilon'$ for all $i\in \{1,\ldots,k\}$.
\end{Lem}

\begin{proof}[\textbf{\upshape Proof:}]
    We prove the first claim by contradiction. Suppose there exists  $i\in \{1,\ldots,k\}$ such that $D\pp{\ba_j^*,\ba_i}\ge\epsilon'$ for all $j\in \{1,\ldots,m\}$.  Then for any $\bw\in W\cap B_D(\ba_i,\epsilon)$, we have by the triangular inequality \eqref{eq:ps tri} that
    \[
  D\pp{\bw,A_m^*}\ge D\pp{\ba_i, A_m^*} - D^\dagger\pp{\bw,\ba_i} \ge\epsilon'  -   r_A^\dagger(\epsilon).
    \]
Hence combining this and \eqref{eq:Ass concen bound},  
\begin{equation}\label{eq:D A_m* ave}
\frac{1}{|W|} \sum_{\bw \in W} D\pp{\bw, A_m^*}\ge\frac{1}{|W|} \sum_{\bw \in W\cap B_D(\ba_i,\epsilon)} D\pp{\bw, A_m^*} \ge  \pp{p_i-\delta}\pp{\epsilon'- r_A^\dagger(\epsilon)}\ge  \pp{ p_{\min}-\delta  } \pp{\epsilon'-r_A^\dagger(\epsilon)}.
\end{equation}
Next, suppose that a multiset $S$  on
$\dsph$ contains $A$ and  $|S|=m$,  which is only possible when $m\ge k$ as assumed.  Then we  have  $D(\bw,S)\le D(\bw,A)$. Set $U_\epsilon:=W\cap \pp{\cup_{i=1}^k B_D(\ba_i,\epsilon)}$, we have that
\begin{equation}\label{eq:D B ave}
 \frac{1}{|W|} \sum_{\bw \in W} D(\bw, S)\le  \frac{1}{|W|} \pp{ \sum_{\bw \in U_\epsilon} D(\bw, A)  +\sum_{\bw \in W\setminus U_\epsilon} 1 }  < \pp{1-k\delta}\epsilon + k\delta,
\end{equation}
where the last inequality is obtained by maximizing $|W\setminus U_\epsilon|$ with the constraint  \eqref{eq:Ass concen bound}.
Now in view of  \eqref{eq:k-clust opt}, the first expression in \eqref{eq:D A_m* ave} is   less than  or equal to the first expression in \eqref{eq:D B ave},  and hence these two inequalities imply:
\[
 \epsilon'< \pc{(1-k\delta ) \epsilon +k\delta}/\pp{p_{\min}-\delta} +  r_A^\dagger(\epsilon),
\]
which contradicts the choice of $\epsilon'$.

For the second claim, note that $B_D(\mbf{a}_i,r_A)$'s are disjoint, $i\in \pc{1,\ldots,k}$. So if $\epsilon'<r_A$, it is impossible that $D\left(\mbf{a}_j^*,\mbf{a}_i\right)<\epsilon'$ and $D\left(\mbf{a}_j^*,\mbf{a}_{i'}\right)<\epsilon'$ hold simultaneously when $i\neq i'$. The conclusion 
 then follows.
\end{proof}


As a consequence of the previous lemma, if $W$ is concentrated around $k$ centers but is partitioned into more than $k$ clusters, then either some clusters will have a small size, or at least two of the centers will be close to each other, as articulated in the next lemma.
\begin{Lem}\label{Lem:m>k either}
    Suppose a multiset $W$ satisfies the condition $\mathcal{A}(\epsilon,\delta )$. Assume  additionally that $\epsilon'$ in \eqref{eq:epsilon prime} satisfies $\epsilon'<r_A$. 
    Let  $\pp{A^*_m=\{\mbf{a}_1^*,\ldots,\mbf{a}_m^*\},\frak{C}_m=\{C_1,\ldots,C_m\}}$ be an  $m$-clustering of $W$ as defined in Definition \ref{Def:k-clust}, $m> k$.  
    Then either of the following happens: 
    \[
    \min_{i=1,\ldots,m} \frac{|C_i|}{|W|}\le k\delta \quad   \text{ or } \quad  \min_{ 1\le i<j\le m} D\pp{\ba_i^*,\ba_j^*}\le \epsilon'+ 2 r_A^\dagger(\epsilon) + r_A^{\dagger}(\epsilon'). 
    \]
\end{Lem}
\begin{proof}[\textbf{\upshape Proof:}]
Since  $B_D\pp{\ba_i,\epsilon'}$, $i\in \pc{1,\ldots,k}$, are disjoint (because $\epsilon'<r_A$),   by  Lemma   \ref{Lem:m>=k center},   we can, without loss of generality,  assume that $\ba_i^* \in B_D\pp{\ba_i,\epsilon'}$, $i\in \pc{1,\ldots,k}$ .
We now divide into two cases as follows.

\noindent
\emph{Case 1:} there exists one $j\in \{k+1,\ldots,m\}$ (fixed below in the discussion of this case) which satisfies  $D\pp{\ba_j^*,A}>\epsilon'+ 2r_A^\dagger(\epsilon)$. 
Then for any $i\in \{1,\ldots,k\}$ and any $\bw\in W\cap B_D\pp{\ba_i,\epsilon}$, we have by the triangular inequality \eqref{eq:ps tri} that 
$D\pp{\bw,\ba_i^*}\le D\pp{\ba_i,\ba_i^*}+D^\dagger\pp{\ba_i,\bw}\le \epsilon'+r_A^{\dagger}(\epsilon)$,
and hence
\[
D\pp{\bw,\ba_j^*}\ge D\pp{\ba_j^*,\ba_i}-D^\dagger\pp{\bw,\ba_i}> \epsilon' + 2 r_A^\dagger(\epsilon) - r_A^\dagger(\epsilon)\ge D\pp{\bw,\ba_i^*}. 
\]
This in view of  Definition \ref{Def:k-clust} implies that  $W\cap B_D\pp{\ba_i,\epsilon}\subset W\cap  C_j^c$ for all $i\in \pc{1,\ldots,k}$. Therefore, we have by \eqref{eq:Ass concen bound} that 
\[
\min_{i=1,\ldots,m} \frac{|C_i|}{|W|}\le  \frac{|C_j|}{|W|} \le   \frac{|W\cap  \bigcap_{i=1,\ldots,k} B_D\pp{\ba_i,\epsilon}^c |}{|W|}\le k\delta.
\]

\noindent \emph{Case 2}:  for any $j\in \{k+1,\ldots,m\}$, we have $D\pp{\ba_j^*, \ba_i}\le \epsilon' + 2r_A^{\dagger}(\epsilon)$ for some $i\in \{1,\ldots,k\}$. Then for any such pair of $j$ and $i$, we have
\[
 D\pp{\ba_i^*,\ba_j^*} \le 
  D\pp{\ba_i,\ba_j^*}+ D^{\dagger}(\ba_i,\ba_i^*) \le \epsilon'+2r_A^\dagger(\epsilon) + r_A^{\dagger}(\epsilon').
\]
\end{proof}

The next lemma says, if $W$ is concentrated around $k$ centers and is partitioned into exactly $k$ clusters, the (unpenalized) ASW $\bar{S}$ will favor this clustering by providing a high score that is close to $1$. Additionally, none of the clusters will have a size too small, and no two cluster centers will be too close to each other.

\begin{Lem}\label{Lem:m=k and}  Suppose a multiset $W$ satisfies the condition $\mathcal{A}(\epsilon,\delta )$. Let $\pp{A^*_k=\{\mbf{a}_1^*,\ldots,\mbf{a}_k^*\},\frak{C}_k=\{C_1,\ldots,C_k\}}$ be a  $k$-clustering of $W$ as defined in Definition \ref{Def:k-clust}. Suppose in addition   \begin{equation}\label{eq:r_A lb} r_A> \epsilon'+ 2 r_A^\dagger(\epsilon) + r_A^{\dagger}(\epsilon')
   \end{equation}
   with $\epsilon'$ in \eqref{eq:epsilon prime}. Then the (unpenalized) ASW $\bar{S}$ satisfies 
  \[
  \bar{S}=\bar{S}(W;A^*_k,\frak{C}_k)\ge  1- \pp{1-k\delta } \frac{  \epsilon' +  r_A^\dagger(\epsilon) }{r_A - r_A^\dagger(\epsilon)   - r_A^\dagger(\epsilon')}   -      k\delta.
  \]
    In addition,  with the same permutation $\pi:\{1,\ldots,k\}\mapsto\{1,\ldots,k\}$ found in Lemma \ref{Lem:m>=k center}, we have
    \[\frac{|C_{\pi(i)}|}{|W|}\ge p_i-\delta \text{ for each }i\quad   \text{ and } \quad  \min_{ 1\le i<j\le k} D\pp{\ba_i^*,\ba_j^*}\ge  r_A- 2 r_A^{\dagger}(\epsilon'),
    \]
    where when $k=1$, $\min_{ 1\le i<j\le k} D\pp{\ba_i^*,\ba_j^*}$ is understood as 1, and  the inequalities still hold.
\end{Lem}
\begin{proof}[\textbf{\upshape Proof:}]
  Since $r_A>\epsilon'$, by Lemma \ref{Lem:m>=k center}, there exists a permutation $\pi: \{1,\ldots,k\}\mapsto\{1,\ldots,k\}$, such that $D\pp{\ba_i,\ba_{\pi(i)}^*}< \epsilon'$, $i\in\pc{1,\ldots,k}$.  Then for each $i$ and any $\mbf{w}\in B_D(\mbf{a}_i,\epsilon)$, we have by the triangular inequality \eqref{eq:ps tri} that
   \begin{equation}\label{eq:D ineq 1}
D\pp{\mbf{w},\ba_{\pi(i)}^*}\le D\pp{\mbf{a}_i,\mbf{a}_{\pi(i)}^*}+D^{\dagger}\pp{\mbf{w},\mbf{a}_i} <\epsilon' +  r_A^\dagger(\epsilon),
   \end{equation}
and for $j\neq i$  that 
\begin{equation}\label{eq:D ineq 2}
D\left(\mbf{w},\mbf{a}_{\pi(j)}^*\right)\ge D\left( \mbf{a}_i,\mbf{a}_j \right) -D^{\dagger}(\ba_j,\ba_{\pi(j)}^*) -D^{\dagger}(\bw,\ba_i)\ge r_A - r_A^\dagger(\epsilon')-r_A^\dagger(\epsilon),
\end{equation}
where if $k=1$, the left-hand side $D\left(\mbf{w},\mbf{a}_{\pi(j)}^*\right)$ in \eqref{eq:D ineq 2} is understood as 1, and the inequality still holds.
Writing as before $U_\epsilon=\bigcup_{1 \le i\le k} B_D(\mbf{a}_i,\epsilon) \cap W $. In view of \eqref{eq:Ass concen bound} and the inequalities above, we have
$$
\bar{S}= \frac{1}{|W|} \pp{ \sum_{\mbf{w}\in U_\epsilon } +\sum_{\mbf{w}\in W \setminus U_\epsilon } } 
 \pp{1-\frac{a(\mbf{w})}{b(\bw)} }\ge  \pp{1- \frac{  \epsilon' +  r_A^\dagger(\epsilon) }{r_A - r_A^\dagger(\epsilon)   - r_A^\dagger(\epsilon')} } \pp{1-k\delta }  +0,
$$
which implies the  first claim.

For the second claim, in view of Definition \ref{Def:k-clust}, \eqref{eq:r_A lb},  \eqref{eq:D ineq 1}, \eqref{eq:D ineq 2}, we have $W\cap B_D\pp{\ba_i,\epsilon}\subset  C_{\pi(i)}$, $i\in \pc{1,\ldots,k}$. Hence by \eqref{eq:Ass concen bound},
\[\frac{|C_{\pi(i)}|}{|W|}\ge \frac{|W\cap B_D\pp{\ba_i,\epsilon}|}{|W|}\ge  p_{i} - \delta. 
\]
Furthermore, for any $1\le i<j\le k$ and $k>1$,
\[D\pp{\ba_{\pi(i)}^*,\ba_{\pi(j)}^*}\ge D\pp{\ba_i,\ba_j}- D^\dagger\left(\ba_i,\ba_{\pi(i)}^*\right)-D^\dagger\left(\ba_j,\ba_{\pi(j)}^*\right)\ge r_A-2r_A^{\dagger}(\epsilon').
\]
\end{proof}

\subsubsection{Proof of Theorem \ref{Thm:sil cons}}

We first state a result regarding the ASW $\bar{S}$ in \eqref{eq:ASW} when the number of clusters is less than or equal to $k$, the true order of the discrete spectral measure \eqref{eq:disc spec}.   
\begin{Pro}\label{Pro:ASW m<=k}
Suppose $\mbf{X}$ satisfying \eqref{eq:equiv tail}  and \eqref{eq:Lambda} has a discrete spectral measure of the form $ H= \sum_{i=1}^k p_i \delta_{\mbf{a}_i}$,
where
$\mbf{a}_i$'s are distinct points on $\dsph$,  and $p_i>0$, $p_1+\cdots+p_k=1$. Let $W_n$ denote the extremal subsample as in \eqref{eq:W_n}, and  $\pp{A_{m,n},\mathfrak{C}_{m,n}}$ form an $m$-clustering of $W_n$ as defined in Definition \ref{Def:k-clust} with respect to a dissimilarity measure $D$ defined in Definition \ref{Def:dis}. 
If $m<k$,  then almost surely,
    \[
  \limsup_n \bar{S}\pp{W_n; A_{m,n}, \frak{C}_{m,n}  }\le  1- r_A p_{\min},
    \]
where $r_A$ is as in \eqref{eq:r_A} and $p_{\min}$ is as in \eqref{eq:pmin}. If $m=k$, then almost surely,
\[
\lim_n \bar{S}\pp{W_n; A_{k,n}, \frak{C}_{k,n}  } = 1.
\]
\end{Pro}
\begin{proof}[\textbf{\upshape Proof:}]
Note that $\epsilon'\to 0$ as $\epsilon,\delta\rightarrow 0$, we may choose them small enough such that \eqref{eq:r_A lb} is satisfied. Define the event
\begin{equation}\label{eq:E_eps del}
E_n(\epsilon,\delta) =\left\{|W_n\cap B_D(\ba_i,\epsilon)|\ge |W_n|(p_i-\delta),\ i\in \pc{1,\ldots,k}\right\}.
\end{equation}
By Proposition \ref{Pro:basic consist} with $S=B_D(\ba_i,\epsilon)$ and the choice $\epsilon<r_A$, we have each $H_n(S)=|W_n\cap B_D(\ba_i,\epsilon)|/|W_n|$ converges almost surely to $H(S)=p_i$, $i\in \{1,\ldots,k\}$.
Hence, with probability $1$,  the event $E_n(\epsilon,\delta)$ happens eventually as $n\rightarrow\infty$, namely, 
$\Pr\pp{ \liminf_n \ind{E_n(\epsilon,\delta)}=1 }=1$.
Since $W_n$ satisfies the condition $\mathcal{A}(\epsilon,\delta )$ on $E_n(\epsilon,\delta)$, by Lemmas \ref{Lem:m<k ASW} and \ref{Lem:m=k and}, for almost every outcome $\omega$ in the sample space $\Omega$,    when  $n$ is sufficiently large, we have  when $m<k$ that 
\[
\bar{S}(W_n;A_{m,n},\frak{C}_{m,n}) \ind{E_n(\epsilon,\delta)}\le \pc{ 1- \pp{ p_{\min} -\delta}\pp{r_A-r_A^\dagger(\epsilon)}}\ind{E_n(\epsilon,\delta)}
\]
and
\[
\bar{S}(W_n;A_{k,n},\frak{C}_{k,n}) \ind{E_n(\epsilon,\delta)}\ge  
\pc{1- \pp{1-k\delta } \frac{  \epsilon' +  r_A^\dagger(\epsilon) }{r_A - r_A^\dagger(\epsilon)   - r_A^\dagger(\epsilon')}   -      k\delta}\ind{E_n(\epsilon,\delta)}.
\]
The desired results follow  if one takes $\limsup_n$ and $\liminf_n$ respectively in the two inequalities above, and then lets $\delta,\epsilon\rightarrow 0$ ({see also Remark \ref{Rem:r_A dagger}}).
\end{proof}

 Next, we state a result on the penalty $P_t$ in \eqref{eq:pen}  when the number of clusters exceeds or equals $k$.
\begin{Pro}\label{Pro:p_t m>=k}
Suppose $\mbf{X}$ satisfying \eqref{eq:equiv tail}  and \eqref{eq:Lambda} has a discrete spectral measure of the form $ H= \sum_{i=1}^k p_i \delta_{\mbf{a}_i}$,
where
$\mbf{a}_i$'s are distinct points on $\dsph$, and $p_i>0$, $p_1+\cdots+p_k=1$. Let $W_n$ denote the extremal subsample as in \eqref{eq:W_n},  and  $\pp{A_{m,n},\mathfrak{C}_{m,n}}$ form an $m$-clustering of $W_n$ as defined in Definition \ref{Def:k-clust} with respect to a dissimilarity measure $D$ defined in Definition \ref{Def:dis}. 
Suppose $t>0$.
If $m>k$, we have  almost surely
\[
\lim_n P_t(W_n;   A_{m,n}, \mathfrak{C}_{m,n})=1.
\]
If $m=k$,  we have almost surely  
\[
\limsup_n P_t(W_n;   A_{k,n}, \mathfrak{C}_{k,n})\le 1- \pp{  r_A  k p_{\min} }^t,
\]
where $r_A$ is as in \eqref{eq:r_A} and $p_{\min}$ is as in \eqref{eq:pmin}.
\end{Pro}
\begin{proof}[\textbf{\upshape Proof:}]
 The argument is similar to that of Proposition \ref{Pro:ASW m<=k}. In particular, under the restriction to the event $E_n(\epsilon,\delta)$ in \eqref{eq:E_eps del},  we have  by Lemma  \ref{Lem:m>k either} that for $m>k$    
\[
P_t(W_n;   A_{m,n}, \mathfrak{C}_{m,n}) \ge 1- \pp{ k^2\delta}^t \vee \pp{\epsilon'+2r_A^{\dagger}(\epsilon)+r_A^{\dagger}(\epsilon')}^t,
\]
and by Lemma \ref{Lem:m=k and} that
\[
P_t(W_n;   A_{k,n}, \mathfrak{C}_{k,n})\le 1-[k(p_{\min}-\delta) (r_A-2r_A^\dagger(\epsilon'))]^t. 
\]
We omit the rest of the details.  
\end{proof}

Now we are ready to prove Theorem \ref{Thm:sil cons}.
\begin{proof}[\textbf{\upshape Proof of Theorem \ref{Thm:sil cons}:}] 
Putting together Propositions \ref{Pro:ASW m<=k} and \ref{Pro:p_t m>=k},   and using the facts that $\bar{S}\in [0,1]$ and $P_t\in [0,1]$, we have almost surely that
\[
\begin{cases} 
    \limsup_n  S_t(W_n; A_{m,n}, \mathfrak{C}_{m,n})\le 1- r_A p_{\min},   
 & \text{ if } m<k;\\
     \liminf_n S_t(W_n; A_{k,n}, \mathfrak{C}_{k,n}) \ge   \pp{r_A k p_{\min} }^t, & \text{ if } m=k; \\
      \limsup_n S_t(W_n; A_{m,n}, \mathfrak{C}_{m,n}) \le 0,  &\text{ if } m>k. 
\end{cases}
\]
Therefore, the desired claim follows. 
\end{proof}

\section{Large deviation analysis of clustering-based spectral estimation \label{sec:rate}}

In this section, we provide a quantitative assessment of the consistency result in Corollary \ref{Cor:disc consist} through large-deviation-type bounds. This analysis is made possible through certain estimates used in the proof of Theorem \ref{Thm:sil cons} (see Section \ref{sec:proof consist}).

 First, we prepare a Chernoff-Hoeffding-type bound for the sum of a Binomial random number of Bernoulli random variables, which may be of some independent interest.
\begin{Lem}\label{Lem:hoef}
Suppose $B_i$, $i\in \bb{Z}_+$, are independent Bernoulli random variables with $\Pr(B_i=1)=q_1\in (0,1)$ and $N$ is a Binomial$(n,q_2)$ random variable which is independent of $B_i$'s, $n\in \bb{Z}_+$.  Then  we have for any $r\in (0,1-q_1)$,
\begin{align}
\Pr\pp{ \frac{1}{N} \sum_{i=1}^N B_i  > q_1+r }  &\le \exp \pc{ n q_2 \pb{ e^{ -\cl{D}\pp{q_1+r  \parallel  q_1   } }-1}}\le \exp\pc{nq_2 \pp{e^{-2r^2}-1} },   \label{eq:upper tail Hoef} 
\end{align}
and for any $r\in (0,q_1)$,
\begin{align}
\Pr\pp{ \frac{1}{N} \sum_{i=1}^N B_i  < q_1-r }  &\le \exp \pc{ n q_2 \pb{e^{ -\cl{D}\pp{q_1-r  \parallel  q_1   } }-1}}\le \exp\pc{nq_2 \pp{e^{-2r^2}-1} } , \label{eq:lower tail Hoef}
\end{align}
where $\cl{D}(x \parallel y)= x \ln(x/y) + (1-x)\ln\pc{\pp{1-x}/\pp{1-y}}$ if $x,y\in (0,1)$ (the Kullback–Leibler divergence between two Bernoulli distributions). Here  $\sum_{i=1}^m B_i/m$  is understood as $0$ when $m=0$.
\end{Lem} 
\begin{proof}[\textbf{\upshape Proof:}]
We only prove the \eqref{eq:upper tail Hoef} and the proof of \eqref{eq:lower tail Hoef} is similar. It follows from a version of Hoeffding's inequality for Binomial \cite[Equation (2.1)]{hoeffding1963probability} that for any $m \ge 0 $,
\[
\Pr\pp{ \frac{1}{m} \sum_{i=1}^m B_i  > q_1+r }\le  e^{-m \cl{D}\pp{q_1+r \parallel q_1 }}.
\]
Hence
\begin{align*}
    \Pr\pp{ \frac{1}{N} \sum_{i=1}^N B_i  > q_1+r }  \le &  \sum_{m=0}^n {n\choose m } q_2^m  e^{-m \cl{D}\pp{q_1+r \parallel q_1 }}  (1-q_2)^{n-m} 
    \\= &  \pb{q_2  \pc{  e^{-  \cl{D}\pp{q_1+r \parallel q_1 }} -1} +1}^n \le \exp \pc{ n q_2 \pb{e^{ -\cl{D}\pp{q_1+r  \parallel  q_1   } }-1}},
\end{align*}
where in the last inequality  we have used the inequality $x+1\le \exp\pp{x}$, $x\in \bb{R}$. To obtain the second inequality in \eqref{eq:upper tail Hoef}, it suffices to note that  in view of \cite[Equation (2.3)]{hoeffding1963probability} one has $\cl{D}(q_1 - r |q_1)\ge 2 r^2$.   
\end{proof}
\begin{Rem}
 Note that   when $r$ is small, this simplified bound is approximately $\exp(-2nq_2 r^2)$, a form identical to the usual Hoeffding's inequality (recall $nq_2$ is the effective sample size here).    
\end{Rem}

Let $H=\sum_{i=1}^k p_i \delta_{\ba_i}$ be as defined in \eqref{eq:disc spec}. Let $\pp{{A}_{k,n} =\pp{{\mbf{a}}_{1,n}^k,\ldots,{\mbf{a}}_{k,n}^k},\mathfrak{C}_{k,n}=\pc{{C}_{1,n}^k,\ldots, {C}_{k,n}^k}}$ form a $k$-clustering of the extremal subsample $W_n$ as in \eqref{eq:W_n}. By Corollary \ref{Cor:disc consist}, there exists permutation $\pi$, such that $\mbf{a}_{\pi_n(i),n}^k$ and $p_{\pi_n(i),n}^k $ in \eqref{eq:est p} are consistent estimators for $\ba_i$ and $p_i $, respectively.
Note that an accurate estimation can be interpreted as that for small $x,y>0$, $D(\ba_{\pi(i),n}^k,\mbf{a}_i)<x$ and $|p_{\pi(i),n}^k -p_i|<y$ for all $i\in \{1,\ldots,k\}$. Now consider the complement ``large deviation'' event
 \begin{equation}\label{event_e}
     E(x,y)= \bigcap_{\pi}\bigcup_{i=1}^k \pc{ |\ba_{\pi(i),n}^k -\mbf{a}_i|>x  }\cup \pc{|p_{\pi(i),n}^k -p_i|>y}.
 \end{equation}
where the intersection $\cap_\pi$ is over all permutations $\pi:\{1,\ldots,k\}\mapsto \{1,\ldots,k\}$.
We have the following result.
\begin{Pro}\label{Pro:large dev}
Suppose $\mbf{X}$ satisfying \eqref{eq:equiv tail} and \eqref{eq:Lambda} has  a   spectral measure of the following form $ H= \sum_{i=1}^k p_i \delta_{\mbf{a}_i}$,
where
$\mbf{a}_i$'s are distinct points on $\dsph$,  and $p_i>0$, $p_1+\cdots+p_k=1$. Let $W_n$ denote the extremal subsample as in \eqref{eq:W_n},  and  $\pp{{A}_{k,n}=\pp{{\mbf{a}}_{1,n}^k,\ldots,{\mbf{a}}_{k,n}^k} ,\mathfrak{C}_{k,n}=\pc{{C}_{1,n}^k,\ldots, {C}_{k,n}^k}}$ form an $k$-clustering of $W_n$ as defined in Definition \ref{Def:k-clust} with respect to a dissimilarity measure $D$ defined in Definition \ref{Def:dis}. Let $E(x,y)$ be the event defined in \ref{event_e}. Then for any $x,y>0$,
 \[\limsup_n  \frac{1}{c_{(r)} \ell_n} \ln \Pr\pp{E(x,y)} \le \exp\pp{-2\Delta(x,y)^2}-1
 \]
 where 
 \[\Delta(x,y)=\begin{cases}
     \max\{y/c_k,p_{min}x/(k+x)\},&x<\epsilon_0 \text{, }  y<c_k p_{\min} \epsilon_0/(k+ \epsilon_0),\\
     p_{\min} \epsilon_0/(k+ \epsilon_0),&otherwise,
 \end{cases}\]
 where $ \epsilon_0:=\sup\{\epsilon>0:r_A> \epsilon + r_A^{\dagger}(\epsilon)\}$ and $c_k:=(k\vee 2-1)$.
\end{Pro}
\begin{proof}[\textbf{\upshape Proof:}]
If $H_n\pp{B_D(\mbf{a}_i,\epsilon)}=|W_n\cap B_D(\mbf{a}_i,\epsilon)|/|W_n|\ge p_i-\delta$ for all $i\in \{1,\ldots,k\}$, by Lemmas \ref{Lem:m>=k center} and \ref{Lem:m=k and}, as long as \eqref{eq:r_A lb} holds,  there exists a permutation $\pi:\{1,\ldots,k\}\mapsto\{1,\ldots,k\}$,  such that $D(\ba_{\pi(i),n}^k,\mbf{a}_i)<\epsilon'$ and $| p_{\pi(i),n}^k-p_i|\le c_k\delta$ for all $i\in \{1,\ldots,k\}$. Hence under \eqref{eq:r_A lb}, whenever $\epsilon'\le x$ or $ c_k\delta\le y$,
\[
\Pr\pp{E(x,y)}\le\Pr\left(\bigcap_{\pi}\bigcup_{i=1}^k \pc{ D(\ba_{\pi(i),n}^k,\mbf{a}_i)>\epsilon' }\cup \pc{|p_{\pi(i),n}^k -p_i|>c_k\delta}\right)\le \Pr\left(\bigcup_{i=1}^k\{H_n\pp{B_D(\mbf{a}_i,\epsilon)}< p_i-\delta\}\right),
 \]
 where $H_n$ is the empirical spectral measure in \eqref{eq:emp spec}.
 Observe that  for any $i\in \{1,\ldots,k\}$, 
 \[\pp{|W_n|, \pp{ \ind{\mbf{X}_j/\|\mbf{X}_j\|_{(s)}\in B_D(\mbf{a}_i,\epsilon),\  \|\mbf{X}_j\|_{(r)}\ge 
 \pp{n/ \ell_n}^{1/\alpha}}}_{j=1,\ldots,n}}\EqD (N, (B_j)_{j=1,\ldots,n}), \]
 where $N$ and $B_j$'s are as in Lemma \ref{Lem:hoef} with respective parameters $q_1$ and $q_2$ given as follows:
 \begin{equation}\label{eq:q_1}
q_1=q_1(i,\epsilon,n):=\Pr\pp{{\mbf{X}_1/\|\mbf{X}_1\|_{(s)}\in B_D(\mbf{a}_i,\epsilon),\  \|\mbf{X}_1\|_{(r)}\ge 
 \pp{n/ \ell_n}^{1/\alpha}}}/ \Pr\pp{\|\mbf{X}_1\|_{(r)}\ge 
 \pp{n/ \ell_n}^{1/\alpha}} \rightarrow p_i
 \end{equation} 
 as $n\rightarrow\infty$, 
 where the last convergence holds due to \eqref{eq:conv spec} and the fact that  $B_D(\mbf{a}_i,\epsilon)$'s are disjoint under $\epsilon<r_A$, and 
 \begin{equation}\label{eq:q_2}
 q_2=q_2(n)=\Pr\pp{\|\mbf{X}_1\|_{(r)}\ge 
 \pp{n/ \ell_n}^{1/\alpha}}\sim c_{(r)} \pp{\ell_n/n}  
 \end{equation}
 as $n\rightarrow\infty$.    Now applying Lemma \ref{Lem:hoef}, we have
 \begin{align}\label{eq:Pr H_n bound}
 \Pr\left(\bigcup_{i=1}^k\{H_n\pp{B_D(\mbf{a}_i,\epsilon)}< p_i-\delta\}\right)&\le \sum_{i=1}^k \Pr\pp{H_n(B_D(\mbf{a}_i,\epsilon) )< p_i-\delta }\le k \exp \pp{ n q_2(n) \pb{\exp\pc{ - 2\delta^2 }-1}}.
 \end{align}
 Therefore in view of also \eqref{eq:q_1} and \eqref{eq:q_2}, we have  
 \begin{align*}
  \limsup_n  \frac{1}{\ell_n} \ln\pc{\Pr\left(E(x,y)\right)}
  \le   c_{(r)}\pc{\exp(-2\delta^2)-1}.
 \end{align*}
 The next step is to determine the largest value of $\delta$ as possible. Recall $ \epsilon_0 =\sup\{\epsilon'>0:r_A> \epsilon' + r_A^{\dagger}(\epsilon')\}$. Then when $\epsilon'\in(0, \epsilon_0)$,  for all $\epsilon$ small enough we have $r_A> \epsilon'+ 2 r_A^\dagger(\epsilon) + r_A^{\dagger}(\epsilon')$, namely, \eqref{eq:r_A lb} holds. 
Hence by taking $\epsilon\downarrow 0$ in \eqref{eq:epsilon prime}, we get from $\epsilon'<\epsilon_0$ the restriction $\delta<p_{\min} \epsilon_0/(k+ \epsilon_0)$. Similarly, from $\epsilon'\le x$ we get the restriction $\delta<p_{min}x/(k+x)$. In addition, from $c_k\delta\le y$ we get the restriction $\delta\le y/c_k$. At least one of the last two conditions should be satisfied. Therefore, 
 \[\begin{cases}
     \delta<p_{\min} \epsilon_0/(k+ \epsilon_0),&\text{if }x\ge\epsilon_0,\\
     \delta<p_{\min} \epsilon_0/(k+ \epsilon_0),&\text{if }x<\epsilon_0, y\ge c_kp_{\min} \epsilon_0/(k+ \epsilon_0),\\
     \delta<\max\{y/c_k,p_{min}x/(k+x)\},&\text{if }x<\epsilon_0, y<c_kp_{\min} \epsilon_0/(k+ \epsilon_0).
 \end{cases}\]
 The result then follows.
\end{proof}

\begin{Rem}
  The large-deviation-type estimates in Proposition \ref{Pro:large dev} say that 
   the probability $\Pr\pp{E(x,y)}$ 
  decays exponentially in the  expected extremal subsample size $c_{(r)}\ell_n$. It is worth observing that the expression of  $\Delta(x,y)$ 
  reflects the following: The difficulty of clustering-based estimation  measured by the aforementioned large error probabilities depends negatively on $p_{min}$ and $r_A$ (note that $\epsilon_0$ is an increasing function of $r_A$) and positively on $k$. In other words, the estimation is more accurate when the true discrete spectral measure has fewer atoms, the dissimilarity between the atoms is large, or the probability mass concentrated on each atom is significant. 
\end{Rem}

We also have the following result which states that in the context of Theorem \ref{Thm:sil cons}, the probability of false order election tends to $0$ exponentially fast.
\begin{Pro}\label{Pro:sil con rate}
Suppose $\mbf{X}$ satisfying \eqref{eq:equiv tail} and \eqref{eq:Lambda} has a discrete spectral measure of the form $ H= \sum_{i=1}^k p_i \delta_{\mbf{a}_i}$, $k\in\bb{Z}_+$
where
$\mbf{a}_i$'s are distinct points on $\dsph$,  and $p_i>0$, $p_1+\cdots+p_k=1$. Let $W_n$ denote the extremal subsample as in \eqref{eq:W_n},  and  $\pp{A_{m,n},\mathfrak{C}_{m,n}}$, $m\in \bb{Z}_+$, form an $m$-clustering of $W_n$ as defined in Definition \ref{Def:k-clust} with respect to a dissimilarity measure $D$ defined in Definition \ref{Def:dis}.
Let $r_A$ be defined as in \eqref{eq:r_A}, $p_{\min}$ be defined as in \eqref{eq:pmin} and $t_0$ be defined as in \eqref{eq:t0}. Then fix $t\in (0,t_0)$,
\[
\limsup_n \frac{1}{c_{(r)}\ell_n}\ln\pc{\Pr \pp{S_t(W_n; A_{k,n}, \mathfrak{C}_{k,n})\le  S_t(W_n; A_{m,n}, \mathfrak{C}_{m,n})   \text{ for all } m\neq k }}\le \exp\pp{-2\delta_t(k,p_{\min},r_A)^2}-1
\]
where $\delta_t(k,p_{\min},r_A)>0$ is the solution $\delta$ of the equation $[k(p_{\min}-\delta)r_A]^t -k\delta= \pp{k^2\delta}^t \vee\pp{ \pp{1-\pp{p_{\min}-\delta}r_A} \ind{k\ge 2}}$.
\end{Pro}
\begin{proof}[\textbf{\upshape Proof:}]
Writing $S_t(m)=S_t(W_n; A_{m,n}, \mathfrak{C}_{m,n})$, we have
\begin{align*}
\Pr \pp{S_t(k)\le  S_t(m),\ m\neq k }\le \Pr \pp{\{S_t(k)\le  S_t(m),\ m\neq k \}\cap E_n(\epsilon,\delta)} + \Pr \pp{E_n(\epsilon,\delta)^c},
\end{align*}
where $E_n(\epsilon,\delta)$ is in \eqref{eq:E_eps del}. Combining the inequalities regarding $\bar{S}$ in the proof of Proposition \ref{Pro:ASW m<=k}, and the inequalities regarding $P_t$ in the proof of Proposition \ref{Pro:p_t m>=k},  the event in the first  probability on the right-hand side above is empty as long as   $\delta>0$ satisfies
\[
[k(p_{\min}-\delta)r_A]^t -k\delta> \pp{k^2\delta}^t\vee \pp{ \pp{1-\pp{p_{\min}-\delta}r_A} \ind{k\ge 2}}
\]
and  $\epsilon$ is sufficiently small (depending on $\delta$).
Note that the inequality above holds when $\delta$ is sufficiently small due to  $0<t<t_0=\ln\pp{1-r_A p_{\min}}/\ln\pp{r_Akp_{\min}}$, and its left-hand side is decreasing (to negative values) and its right-hand side is increasing with as $\delta$ increases to $p_{\min}$. Then for any $\delta\in(0,\delta_t(k,p_{\min},r_A))$,  we have in view of  \eqref{eq:q_2} and \eqref{eq:Pr H_n bound}   that
\[
 \lim_n  \frac{1}{c_{(r)}\ell_n} \ln\pc{\Pr \pp{S_t(k)\le  S_t(m),\ m\neq k }}\le \exp\pp{-2 \delta^2}-1.
\]
The proof is concluded by letting $\delta \uparrow \delta_t(k,p_{\min},r_A)$. 
\end{proof}

\section{Clustering and heavy-tailed factor models}\label{sec:clust factor models}

\subsection{The models}\label{sec:factor models}
As  observed by \citet{einmahl2012m} and \citet{janssen2020k}, one may relate a $k$-clustering algorithm to the estimation of certain factor-like models that are often considered in the analysis of multivariate extremes.    Suppose $B=\left(b_{ij}\right)_{i\in \{1,\ldots, d\}, j\in \{1,\ldots,k\} }=\pp{\mbf{b}_1,\ldots,\mbf{b}_k} $, where $\mbf{b}_j=(b_{1j},\ldots,b_{dj})^\top$, $j\in\{1,\ldots,k\}$, are $k$ distinct $d$-dimensional vectors,  $b_{ij}\ge 0$, and that each column and row vector of $B$ is nonzero (otherwise, the dimension $d$ or the factor order $k$ can  be reduced).
 Assume that $\mbf{Z}= (Z_1,\ldots,Z_k)^\top$ is a vector of i.i.d.\  positive  random variables satisfying $\Pr(Z_1> z)\sim z^{-\alpha}$ as  $z\rightarrow\infty$, $\alpha\in (0,\infty)$.  
 Then the  sum-linear   model is given as 
 \begin{equation} \label{eq:sum linear}
 \mbf{X} =\pp{X_1,\ldots,X_d}^\top  =  \pp{\sum_{j=1}^k b_{1j} Z_j,\ldots, \sum_{j=1}^k b_{ dj} Z_j}^\top = B \mbf{Z}. 
 \end{equation}
 On the other hand, we also have the max-linear  model as
 \begin{equation}\label{eq:max linear}
 \mbf{X}=\pp{X_1,\ldots,X_d}^\top   =\pp{\bigvee_{j=1}^k b_{ 1j} Z_j,\ldots, \bigvee_{j=1}^k b_{ dj} Z_j}^\top =    B  \odot \mbf{Z},
 \end{equation}
 where $\odot$ is interpreted as the matrix product with the sum operation replaced by the maximum operation. Note that due to the exchangeability of $\pp{Z_1,\ldots,Z_k}$, either model is identifiable only up to a permutation of the vectors $\mbf{b}_j$, $j\in\{1,\ldots,k\}$, i.e., the distribution of $\mbf{X}$ is unchanged if $B$ is replaced by $B_\pi:=\pp{\mbf{b}_{\pi(1)},\ldots,\mbf{b}_{\pi(k)}}$ for any permutation $\pi:\{1,\ldots,k\}\mapsto \{1,\ldots,k\}$.      The models of types \eqref{eq:sum linear} and \eqref{eq:max linear} have recently attracted considerable interest in connection with causal structural equations for extremes; see, e.g.,  \cite{gissibl2018max, gnecco2021causal}.

 It is known that both models above satisfy MRV \eqref{eq:Lambda}, and have a discrete spectral measure   as in \eqref{eq:disc spec} with
 \begin{equation}\label{eq:spec coef rel}
p_j=\frac{\|\mbf{b}_j\|_{(r)}^\alpha}{\sum_{\ell=1}^k \|\mbf{b}_\ell\|_{(r)}^\alpha},\quad \mbf{a}_j=\frac{\mbf{b}_j}{\|\mbf{b}_j\|_{(s)}}, \quad  j\in \pc{1,\ldots,k}.
\end{equation}
This can be derived based on the well-known  ``single large jump'' heuristic:   when $\|\mbf{X}\|_{(r)}$ is large, it is only due to a single large $Z_j$  with    overwhelming probability.  See, e.g., \cite{medina2024spectral, einmahl2012m};  we mention that these works usually assume the same norm $\|\cdot \|_{(r)}=\|\cdot\|_{(s)}$ and $\alpha=1$,  although an extension  is  straightforward.   
In addition, the marginal standardization condition \eqref{eq:equiv tail} or equivalently \eqref{eq:stand cond}, imposes the following restriction on $B$:
\begin{equation}\label{eq:B restr}
 \sum_{j=1}^k b_{ij}^{\alpha}=1,\quad i\in \pc{1,\ldots,d}.
\end{equation}
We also mention that one may relax the models \eqref{eq:sum linear} and  \eqref{eq:max linear} by adding  a noise term, e.g.,  $\mbf{X}=B\mbf{Z}+\boldsymbol{\varepsilon}$ or $\mbf{X}=(B \odot \mbf{Z}) \vee \boldsymbol{\varepsilon}$, where $\boldsymbol{\varepsilon}=(\varepsilon_1,\ldots,\varepsilon_d)^\top$ is a vector of  i.i.d.\ positive noise terms, and the maximum $\vee$ is performed coordinate-wise. As long as each $\varepsilon_i$ has a  tail lighter than that of $Z_j$,  the conclusions made above still hold (see, e.g.,  \cite{einmahl2012m}).    The discussion also applies to the transformed-linear model of \cite{cooley2019decompositions}.  Finally, we mention that in the context of multivariate extremes,  one typically only considers fitting these models to an extremal subsample (see, e.g., \eqref{eq:W_n}) instead of the whole sample.

\subsection{Order selection and coefficient estimation}\label{sec:fact order coef est}

Due to the discrete nature of the spectral measure,  the likelihood functions of these models are inaccessible (see, e.g., \cite{einmahl2012m,yuen2014crps,einmahl2018continuous}).  Even without taking a perspective of extremes, the max-linear model \eqref{eq:max linear} does not admit a smooth density. Therefore, the usual model selection techniques based on information criteria are not available.  On the other hand,   the spectral measure of these factor models, including (39) and (40), is of the form \eqref{eq:disc spec}. Therefore, the penalized ASW method proposed in Section \ref{sec:order}  could be used to select the order of factors $k$, whose consistency is supported by Theorem \ref{Thm:sil cons}.

Suppose from now on the order $k$ is assumed to be known.
Another noteworthy issue deserving discussion is whether we can translate the estimation of the spectral measure through a $k$-clustering algorithm (refer to Section \ref{sec:sph clust MEV}) into an estimation of the coefficient matrix $B=\pp{\mbf{b}_1,\ldots,\mbf{b}_k}$ in \eqref{eq:sum linear} or \eqref{eq:max linear}. Note that the constraint \eqref{eq:B restr} also needs to be taken into account.  Combining   \eqref{eq:spec coef rel} and \eqref{eq:B restr}, to solve the $kd$ coefficients in $B$ from $p_j$'s and $\mbf{a}_j$'s, we have totally $kd+d-1$ free equations ($k-1$  from the equations for $p_j$'s,  $(d-1)k$ from the equations for $\mbf{a}_j$'s and $d$ from \eqref{eq:B restr}).   When $p_j$'s and $\mbf{a}_j$'s are estimated via $k$-clustering, the over-determined system may not admit a solution, although this over-determined relation holds asymptotically in view of Corollary \ref{Cor:disc consist}.  

 In the following, we describe a simple method to convert spectral estimation to an estimation of $B$ that satisfies the constraint \eqref{eq:B restr}.  
 Observe that the exponent measure $\Lambda$    for the models \eqref{eq:sum linear} and \eqref{eq:max linear}
 concentrates on the rays $\{t\mbf{b}_j: t>0\}$, $j\in \pc{1,\ldots,k}$.    Hence a spectral mass point
 $\mbf{a}_j=\mbf{b}_j/\|\mbf{b}_j\|_{(s)}$ on the $\|\cdot \|_{(s)}$-norm sphere corresponds to a spectral mass point $\mbf{b}_j/\|\mbf{b}_j\|_\alpha=\mbf{a}_j/\|\mbf{a}_j\|_\alpha$ on the $\alpha$-norm sphere,  $j\in \pc{1,\ldots,k}$. 
 The advantage of considering the  $\alpha$-norm sphere  is that $$\sum_{j=1}^k \|\mbf{b}_j\|_\alpha^\alpha=\sum_{i=1}^d \sum_{j=1}^k b_{ij}^\alpha=d
 $$  
 due to   relation \eqref{eq:B restr}. Therefore, under the choice $\|\cdot \|_{(r)}=\|\cdot \|_\alpha$ in \eqref{eq:spec coef rel}, we have $p_jd=\|\mbf{b}_j\|_\alpha^\alpha$, and hence
 \begin{equation}\label{eq:b_j solve}
 \mbf{b}_j= \pp{p_j d}^{1/\alpha} \frac{\mbf{a}_j}{\|\mbf{a}_j\|_\alpha},\quad j\in \pc{1,\ldots,k}.
 \end{equation}
 {Note that if  $\|\cdot\|_{(s)}=\|\cdot\|_\alpha$ already, then $\|\mbf{a}_j\|_\alpha=1$.
 So  one can plug  in estimated $\mbf{a}_j$ and $p_j$   via  $k$-clustering on the  $\alpha$-norm sphere into \eqref{eq:b_j solve},  obtaining,  say, $\wh{\mbf{b}}_j$, $j\in\{1,\ldots,k\}$. However,  the condition  \eqref{eq:B restr} may not be satisfied.  We propose the following simple correction: first, form the preliminary estimated coefficient matrix $\wh{B}:= \pp{\wh{\mbf{b}}_1,\ldots,\wh{\mbf{b}}_k}=: \pp{ {\mbf{r}}_1 ,\ldots, {\mbf{r}}_d}^\top$, where $\mbf{r}_i^\top$, $i\in \pc{1,\ldots,d}$, are row vectors of $\wh{B}$. Then we obtain the final estimate $\wt{B}=\pp{\wt{\mbf{b}}_1,\ldots,\wt{\mbf{b}}_k}$ of $B$  through replacing each row $\mbf{r}_i$ by $\mbf{r}_i/\|\mbf{r}_i\|_\alpha$, which ensures \eqref{eq:B restr}.   It follows from Corollary \ref{Cor:disc consist} and a continuous mapping argument that the thus obtained estimate of $B$ is consistent (up to a permutation of $\mbf{b}_i$'s).

\section{Simulation and real data studies} \label{sec:simdata}

\subsection{Simulation studies}\label{sec:sim}

In this section, we present some simulation studies to illustrate the performance of the penalized ASW method introduced in Section \ref{sec:order}.  We follow the setup in \cite[Section 4]{janssen2020k} to simulate the max-linear factor model \eqref{eq:max linear} with randomly generated coefficient matrix $B$.  In particular, we let the factors $Z_j$'s each follow a standard Fr\'echet ($\alpha=1$) distribution.  We consider 4 different combinations of dimensionality $d$ and true order $k$. Under each $(d,k)$ combination, we describe in the list below the way the coefficient vector $\mbf{b}_j$'s are generated. Note that due to the standardization \eqref{eq:B restr}, only $\mbf{b}_1,\ldots,\mbf{b}_{k-1}$ need to be specified. Let $U_i$'s stand for i.i.d.\ uniform random variables on $[0,1]$;
\begin{itemize}
    \item $d=4, k=2$:  $\mbf{b}_1=(U_1,U_2,U_3,U_4)^\top/2$.
    \item $d=4, k=6$:   $\mbf{b}_1=(U_1,U_2,U_3,U_4)^\top/3$, $\mbf{b}_2=(U_5,0,U_6,0)^\top/3$, $\mbf{b}_3=(0,U_7,0,U_8)^\top/3$, $\mbf{b}_4=(U_9,U_{10},0,0)^\top/3$, $\mbf{b}_5=(0,0,U_{11},U_{12})^\top/3$.
    \item $d=6, k=6$:   $\mbf{b}_1=(U_1,\cdots,U_6)^\top/3$, $\mbf{b}_2=(U_7,0,U_8,0,U_9,0)^\top/3$, $\mbf{b}_3=(0,U_{10},0,U_{11},0,U_{12})^\top/3$, $\mbf{b}_4=(U_{13},U_{14},U_{15},0,0,0)^\top/3$, $\mbf{b}_5=(0,0,0,U_{13},U_{14},U_{15})^\top/3$.
    \item $d=10, k=6$: First 5 factors are $\mbf{b}_1=(U_1,\cdots,U_{10})^\top/2$, $\mbf{b}_2=(U_{11},U_{12},0,\cdots,0)^\top/2$, \\ $\mbf{b}_3=(0,0,U_{13},U_{14},0,\cdots,0)^\top/2$, $\mbf{b}_4=(0,0,0,0,U_{15},U_{16},0,0,0,0)^\top/2$, $\mbf{b}_5=(0,\cdots,0,U_{17},U_{18},U_{19},U_{20})^\top/2$. 
\end{itemize}
For each of the 4 simulation setups described above, we randomly generate 100 models (i.e, 100  coefficient $B$ matrices).  From each of these generated models, we simulate a dataset of size $10000$, extract a subsample of size $1000$  with the largest $2$-norms, and project the subsample on the $2$-norm sphere, namely, we work with $\|\cdot \|_{(r)}=\|\cdot \|_{(s)}=\|\cdot \|_2$. Subsequently, a spherical clustering algorithm (spherical $k$-means or $k$-pc) and the computation of the penalized ASW score is carried out on this projected subsample.  Throughout the paper, for  the spherical $k$-means algorithm, we use the implementation in the $\textsf{R}$ package  \texttt{skmeans}  \cite{skmeans}, and for the $k$-pc algorithm, we use the \textsf{R} implementation provided in the supplementary material of \cite{fomichov2023spherical}. 

In Fig.\ \ref{d4k2} $\sim$ \ref{d10k6}, we demonstrate the simulation results through some graphical representations. Specifically, each colored matrix plot is associated with a $(d,k)$ setup as described above.  In each plot, a column corresponds to a simulated dataset, and there are 100 columns. The upper half of the plot corresponds to spherical $k$-means and the lower half corresponds to $k$-pc. Within each of these halves, a row corresponds to a $t$ penalty parameter specification. The color of a cell in the matrix signifies the order  $m$ chosen by maximizing the penalized ASW. We use a white color to indicate a coincidence of $m$ with the true order $k$, with a deeper shade of red indicating that the greater $m$ falls below the true $k$, and a deeper shade of blue indicating the greater it exceeds the true $k$. The bar graph to the right of the matrix indicates the success rate of order identification (that is, $m=k$) in all 100 instances.

In all these simulation setups, we can observe a tendency for the non-penalized ($t=0$) ASW to overestimate (sometimes greatly) the order. As the penalty parameter $t$ is tuned up from $0$,  we observe a significant bias correction effect, and the order identification success rate is noticeably improved over a range of $t>0$. Note that this success rate is calculated with respect to the same $t$ for different simulated data sets. We expect the success rate to improve if $t$ is adaptively tuned for each dataset following the visual method described in Section \ref{sec:order}.
It is also worth mentioning that the order identification based on $k$-pc tends to be more accurate than that based on $k$-means in most of these simulations.  Moreover, the lower success rate for $d=4,k=6$ seems to arise because this configuration tends to produce centers that are in proximity to one another.

 {In the supplementary material, we include more simulation studies as above with different choices of extremal subsample sizes. The results are similar in terms of the bias correction effect of the penalty. We also observe that when the extremal subsample size is chosen relatively large (e.g., $1500$), the accuracy of order selection deteriorates, which is expected since the effect of the convergence of Corollary \ref{Cor:disc consist} is not yet strong with a low threshold.   
}

\begin{figure}[h]
    \centering
    \includegraphics[width=0.65\textwidth]{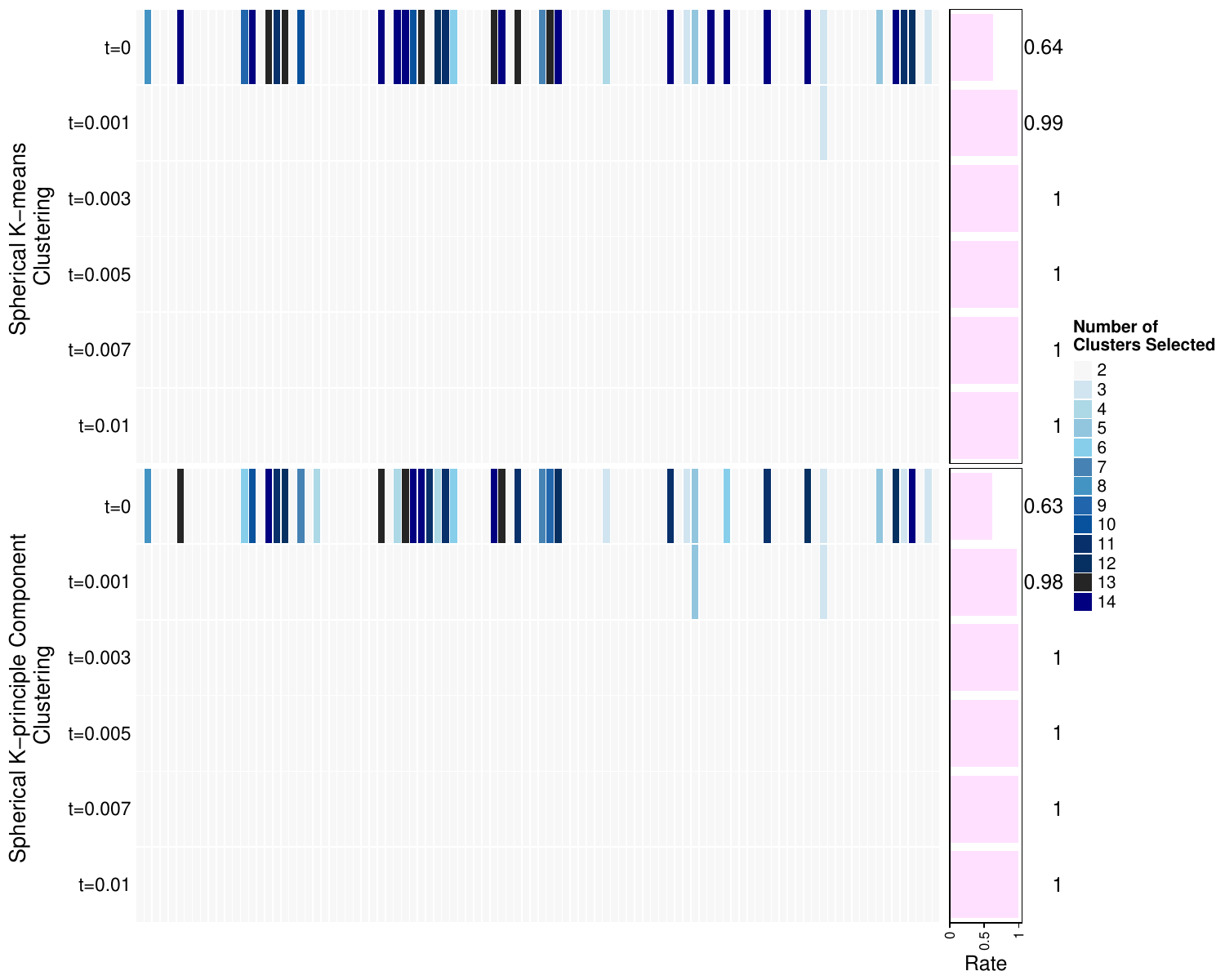}
    \caption{Simulation result visualization for the setup $d=4,k=2$ in Section \ref{sec:sim}.  A column corresponds to a simulated dataset, and a row corresponds to a $t$ penalty parameter specification. See Section \ref{sec:sim} for more details.}
    \label{d4k2}
\end{figure}

\begin{figure}[h]
    \centering
    \includegraphics[width=0.65\textwidth]{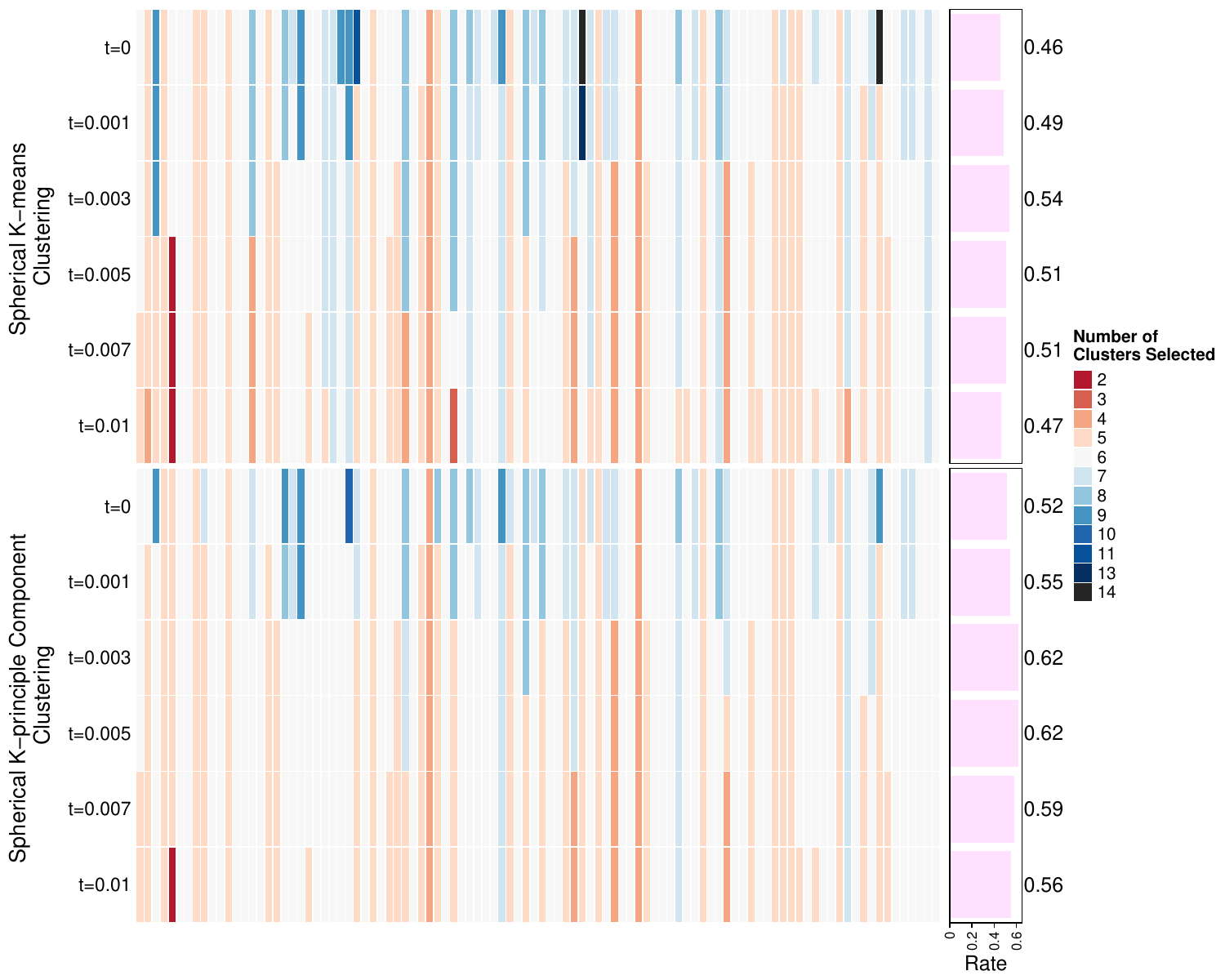}
    \caption{Simulation result visualization for the setup $d=4,k=6$ in Section \ref{sec:sim}. A column corresponds to a simulated dataset, and a row corresponds to a $t$ penalty parameter specification. See Section \ref{sec:sim} for more details.}
    \label{d4k6}
\end{figure}

\begin{figure}[h]
    \centering
    \includegraphics[width=0.65\textwidth]{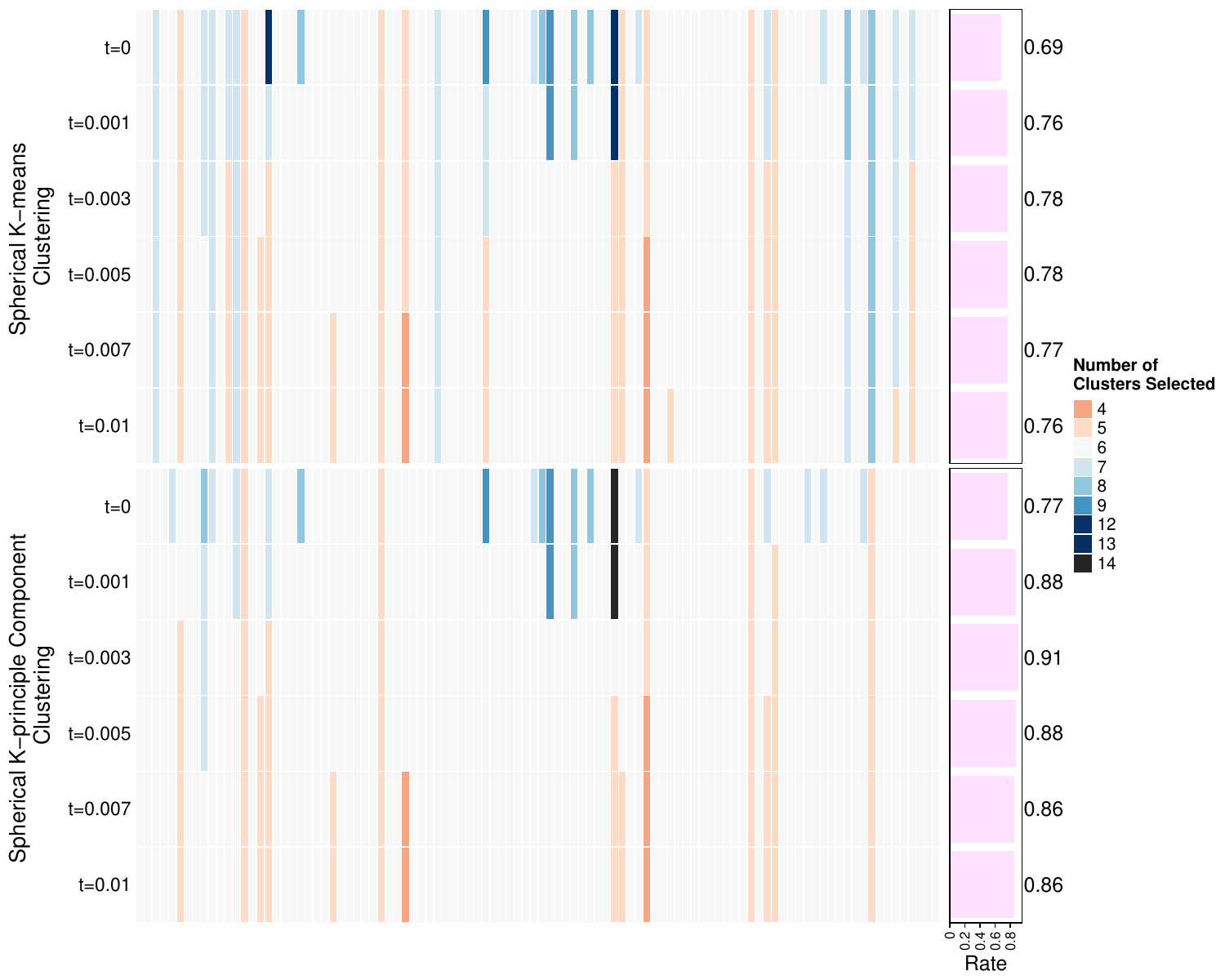}
    \caption{Simulation result visualization for the setup $d=6,k=6$ in Section \ref{sec:sim}. A column corresponds to a simulated dataset, and a row corresponds to a $t$ penalty parameter specification. See Section \ref{sec:sim} for more details.}
    \label{d6k6}
\end{figure}

\begin{figure}[h]
    \centering    \includegraphics[width=0.65\textwidth]{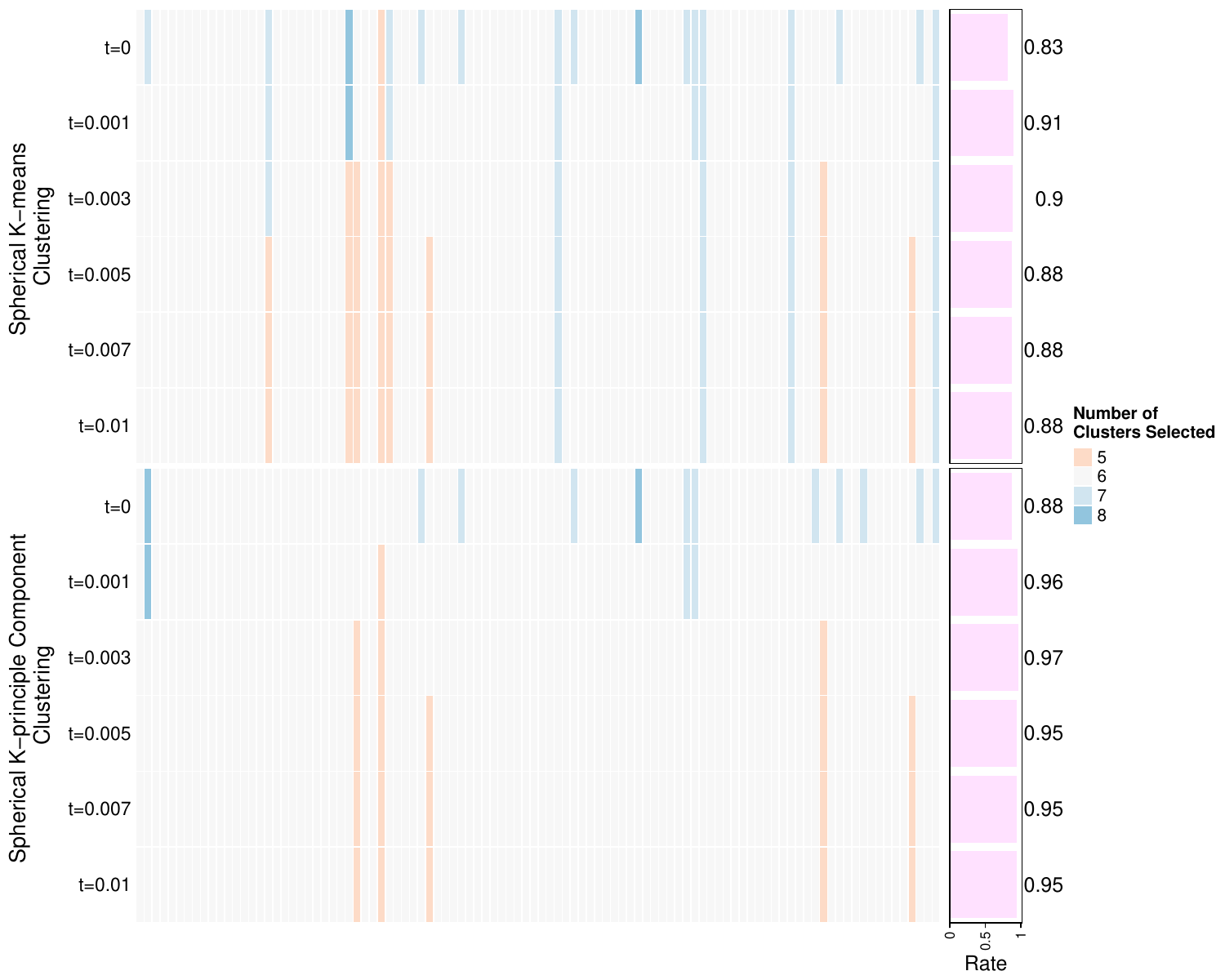}
    \caption{Simulation result visualization for the setup $d=10,k=6$ in Section \ref{sec:sim}. A column corresponds to a simulated dataset, and a row corresponds to a $t$ penalty parameter specification. See Section \ref{sec:sim} for more details.}
    \label{d10k6}
\end{figure}
\clearpage

\subsection{Real data demonstrations}\label{sec:data}

In this section, we use real data examples to demonstrate order selection through penalized ASW as introduced in Section \ref{sec:order}, as well as conversion of clustering-based spectral estimation to a factor coefficient matrix as mentioned in Section \ref{sec:fact order coef est}. We present only the analysis based on the spherical $k$-pc algorithm, that is, the dissimilarity measure $D$ is as in \eqref{eq:pc dis}. The reason for doing so is two-fold. Firstly, the simulation study in Section \ref{sec:sim} seems to suggest a better empirical performance for order selection based on the $k$-pc algorithm. Secondly, as pointed out in \cite{fomichov2023spherical}, the $k$-pc algorithm is more suitable for the detection of groups of concomitant extremes, namely, subsets of variables that
tend to be simultaneously large.  The second property facilitates the comparison of the order $k$ selected with some ``ground truth'' from the background information of the datasets.

 In each of these studies, suppose that the observed data is $(\mbf{x}_i)=(\mbf{x}_i= \pp{x_{i1},\ldots,x_{id}}^\top \in [0,\infty)^d,\ i\in \{1,\ldots,n\})$.
  We follow a conventional approach to marginally standardize a dataset,   so that the assumption \eqref{eq:stand cond} with $\alpha=2$ is roughly met.   In particular, setting $\hat{F}_j(x)=n^{-1}\sum_{i=1}^n \ind{x_{ij}<  x}$  (under this choice of empirical CDF we ensure $\hat{F}_j(x_{ij}))<1$), $j\in \{1,\ldots,d\}$,  the transformed data is given by $(\wt{\mbf{x}}_i)=(\wt{\mbf{x}}_i= \pp{\wt{x}_{i1},\ldots,\wt{x}_{id}}^\top \in [0,\infty)^d,\ i\in \{1,\ldots,n\})$, where $\wt{x}_{ij}:=\pb{-\log\pc{\hat{F}_j(x_{ij})}}^{-1/2}$; if $\hat{F}_j$ were the true CDF for the data in dimension $j$, then $\wt{x}_{ij}$ would follow a standard $2$-Fr\'echet distribution. Next, to prepare for the clustering of multivariate extremes, as in the simulation study in Section \ref{sec:sim}, we select the extremal subsample of $(\wt{\mbf{x}}_i)$  with  $10$\% largest $2$-norms and project the subsample onto the $2$-norm sphere, namely, we work with $\|\cdot \|_{(r)}=\|\cdot \|_{(s)}=\|\cdot \|_2$.

\subsubsection{Air Pollution Data}

 The air pollution dataset is found in the   $\textsf{R}$ package \texttt{texmex} \cite{texmex}, orginated from an online supplementary material of \cite{heffernan2004conditional}. It concerns air quality recordings in Leeds, U.K., specifically in the city center. The data span from 1994 to 1998, divided into summer and winter sets. The summer dataset comprises 578 observations, covering the months from April to July inclusively, while the winter dataset consists of 532 observations, encompassing the months from November to February inclusively. Each observation records the daily maximum values of five pollutants: Ozone, NO2, NO, SO2 and PM10. These datasets were also used in \cite{janssen2020k} to demonstrate the application of the spherical $k$-means clustering method to multivariate extremes.

 In Fig.\ \ref{summer1} and \ref{winter1}, following the same manner as in Fig.\ \ref{kmeanelbow}, the penalized ASW  is plotted against the number of clusters, where different curves correspond to different values of the tuning parameter $t$.  With the visual method described in Section \ref{sec:order}, we can identify orders as  $5$  (although $4$ seems to be a reasonable choice as well) for the summer data and  $3$  for the winter data respectively.
 These orders are similar to the choices $5$ for the summer data and $4$  for the winter data made in \cite{janssen2020k}  under the guidance of certain elbow plots (see \cite[Fig.\ 1]{janssen2020k}).   From the elbow plot in \cite[Fig.\ 1]{janssen2020k}, it seems that $k=3$ for the winter data is also plausible.  Recall also that here we use the spherical $k$-pc algorithm of \cite{fomichov2023spherical} while \cite{janssen2020k} used the spherical $k$-means.

    Following the method introduced in Section \ref{sec:fact order coef est}  with   $\|\cdot \|_{(s)}=\|\cdot \|_{(r)}=\|\cdot \|_{2}$ and $\alpha=2$, we compute the factor coefficient matrix $B$   for the two datasets; see  Fig.\ \ref{summer2} and \ref{winter2}.   
  For the summer data in Fig.\ \ref{summer2}, whose order has been chosen as $5$, the factor coordinates concentrate sharply near coordinate directions, which to an extent indicates an asymptotic (or say extremal) independence (see, e.g., \cite[Chapter 8]{beirlant2006statistics})  of the pollutants.  
   In contrast, for the winter data in Fig.\ \ref{winter2},  whose order has been chosen as $3$, a factor indicates a group of concomitant extremes consisting of NO, NO2 and PM10. The asymptotic dependence between these 3 variables has been observed in \cite{heffernan2004conditional}. This serves as a support for our order choice which has placed these 3 variables in the same concomitant group.



\begin{figure}[h]
    \centering
    \begin{minipage}{0.45\textwidth}
        \centering
        \includegraphics[width=\textwidth]{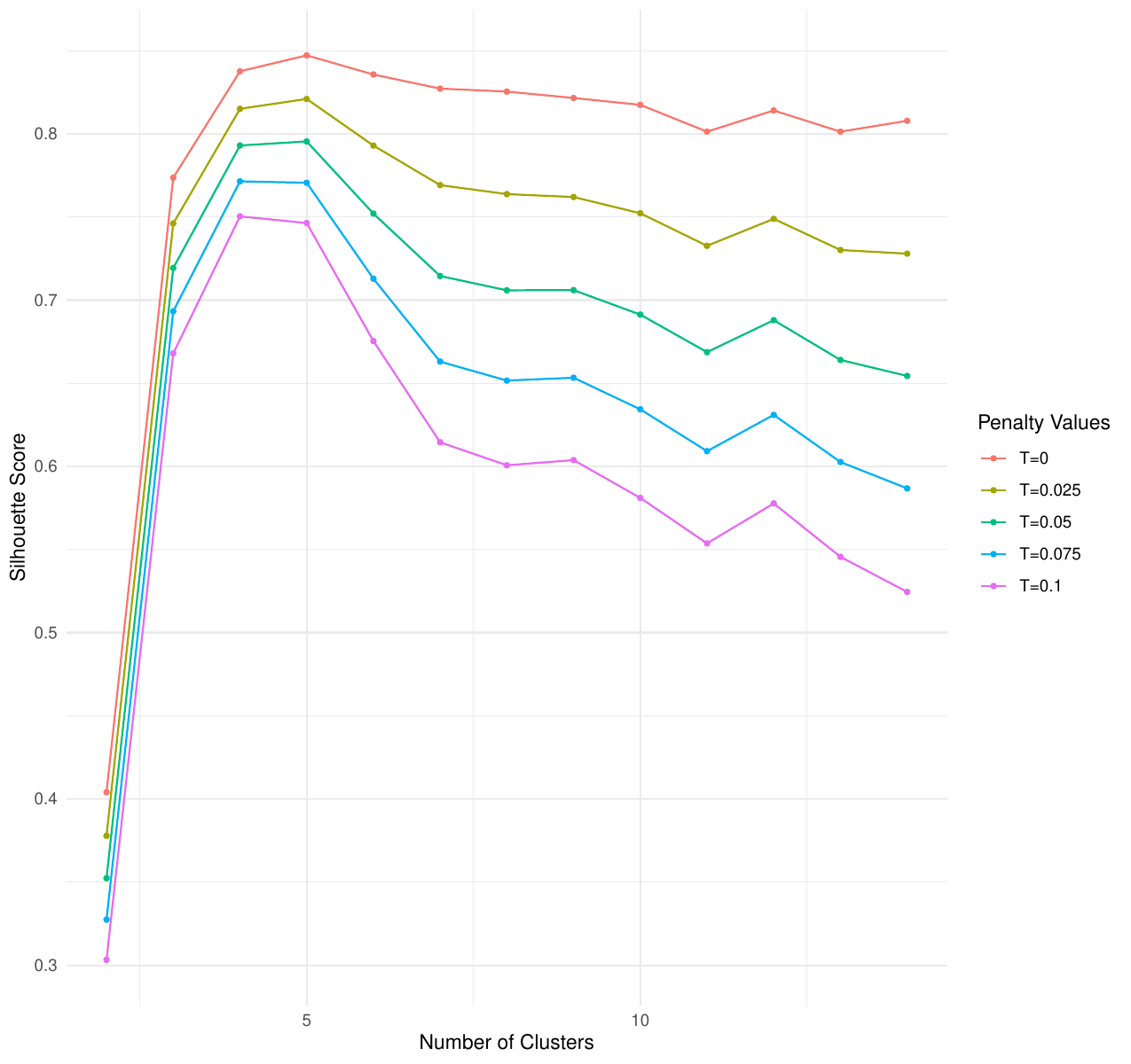} 
        \caption{Penalized ASW curves for summer air pollution data (top $10\%$ norms).}  \label{summer1}
    \end{minipage}
    \hfill
    \begin{minipage}{0.4\textwidth}
        \centering
        \includegraphics[width=\textwidth]{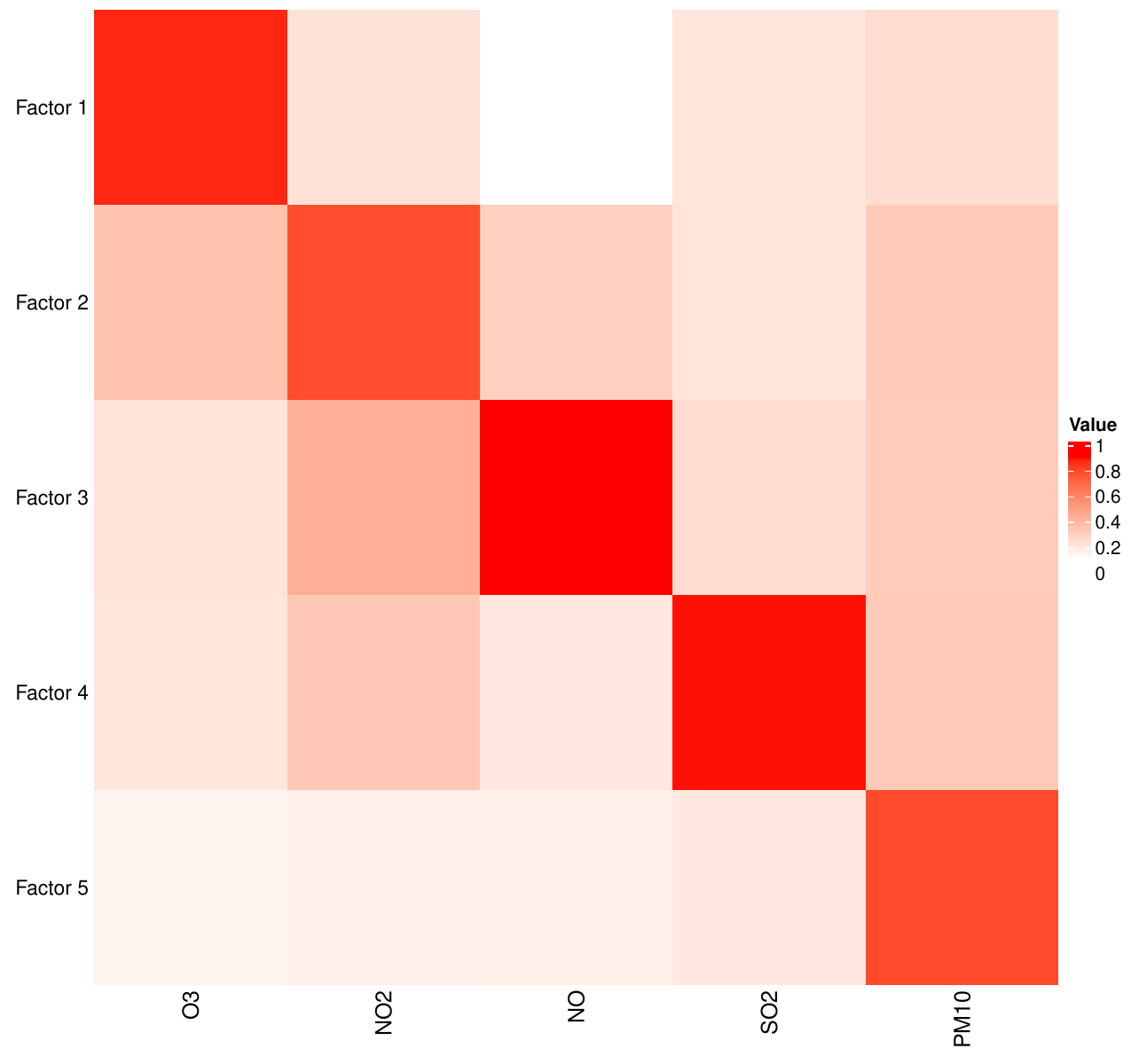} 
        \caption{Estimated $B^\top$ for summer 
       pollution data  (top $10\%$ norms).}\label{summer2}
    \end{minipage}
\end{figure}

\begin{figure}[h]
    \centering
    \begin{minipage}{0.45\textwidth}
        \centering
        \includegraphics[width=\textwidth]{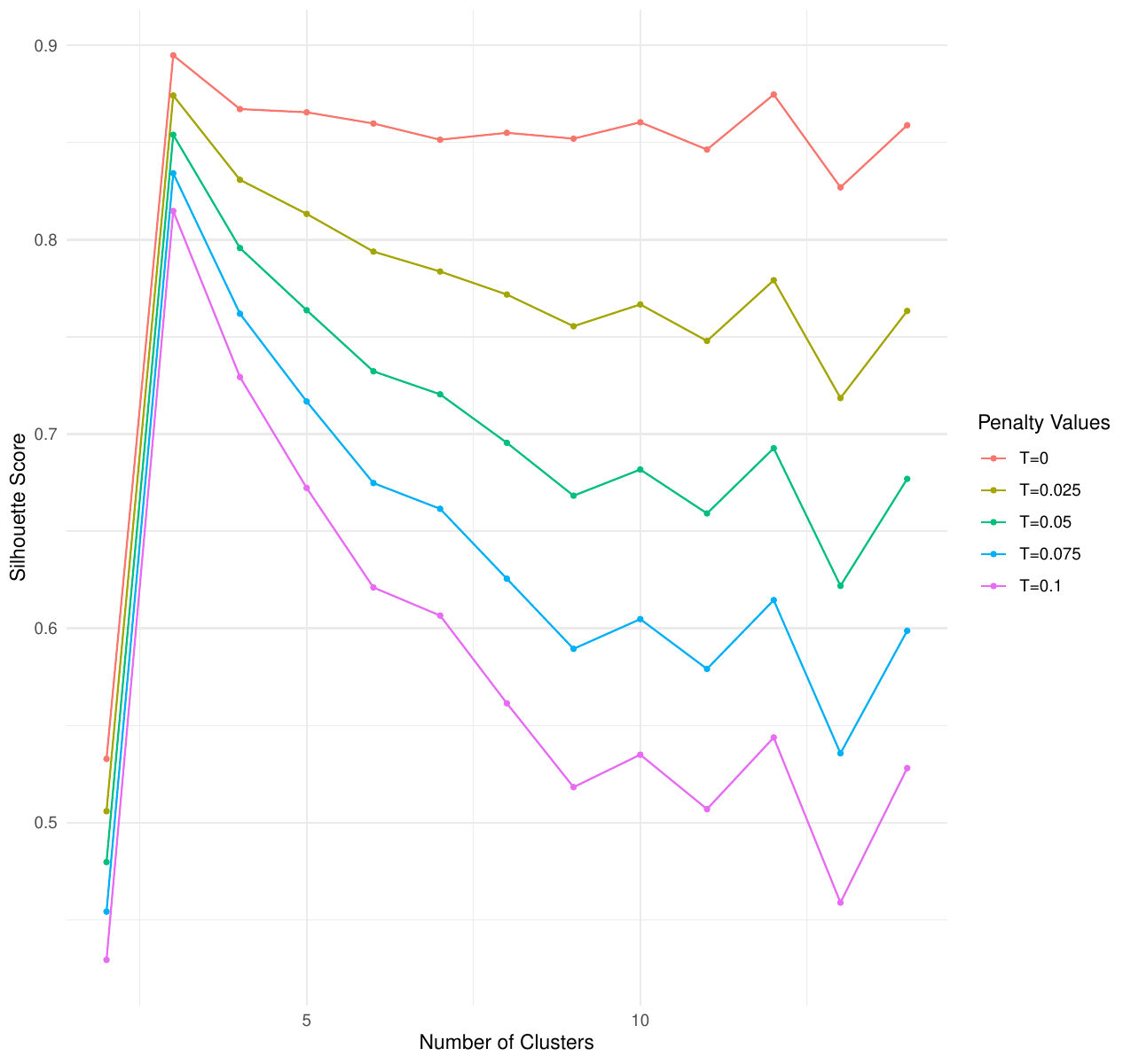} 
        \caption{Penalized ASW curves for winter air pollution data (top $10\%$ norms).} \label{winter1}
    \end{minipage}
    \hfill
    \begin{minipage}{0.4\textwidth}
        \centering
        \includegraphics[width=\textwidth]{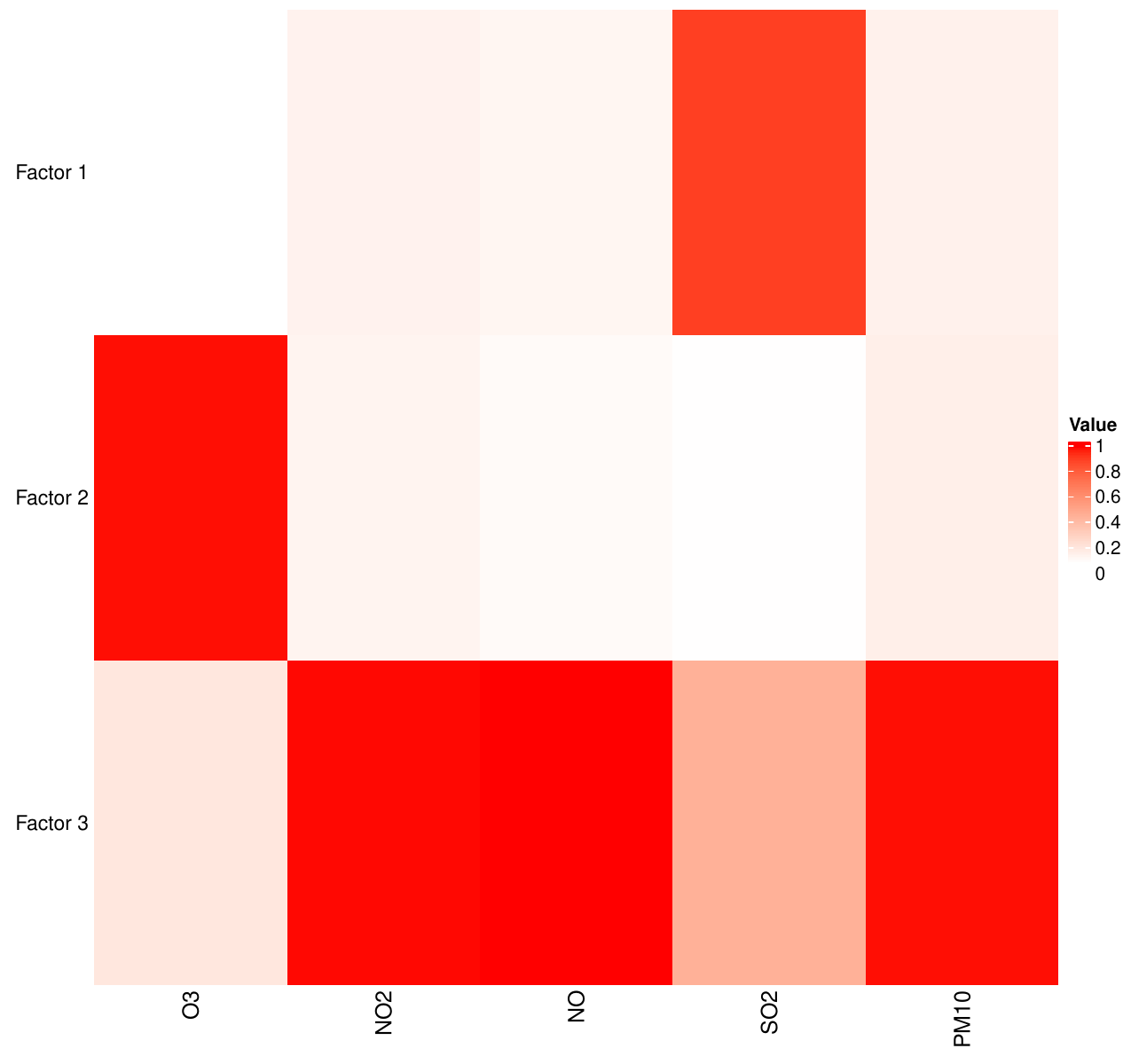} 
        \caption{Estimated $B^\top$ for winter 
       pollution data (top $10\%$ norms).}\label{winter2}
    \end{minipage}
\end{figure}


\clearpage

\subsubsection{River Discharge Data}\label{sec:river}

The river discharge data concerns the daily discharge rate of rivers in North America sourced from the Global Runoff Data Centre  \cite{grdc_portal}. The dataset comprises 16,386 daily records of discharge values from 13 stations spanning the period from December 1, 1976, to October 11, 2021. These 13 stations, shown in Table \ref{tab:my_label} and Fig.\ \ref{map}, are positioned along 5 rivers in America: Willamette River, Mississippi River, Williamson River, Hudson River, and Broad River.







 \begin{table}[h!]
    \caption{Clustering of 13 river discharge stations based   on concomitant Extremes.}
    \label{tab:my_label}
 
 \vskip-0.3cm\hrule

\smallskip
\centering\small
    \centering
    \begin{tabular}[b]{ccc}
      Station Name & River Name &   Factor (Cluster) Index  \\ 
      SALEM, OR	& WILLAMETTE RIVER  &  4   \\
    PORTLAND, OR	& WILLAMETTE RIVER & 4  \\
    HARRISBURG, OR	& WILLAMETTE RIVER &   4  \\
    ST.PAUL, MN	& MISSISSIPPI RIVER & 1  \\
    AITKIN, MN	& MISSISSIPPI RIVER & 1 \\
    THEBES, IL	& MISSISSIPPI RIVER &  6 \\
    CHESTER, IL	& MISSISSIPPI RIVER & 6 \\
    BELOW SPRAGUE RIVER NEAR CHILOQUIN, OR	& WILLIAMSON RIVER & 2 \\
    GREEN ISLAND, NY	& HUDSON RIVER & 5 \\
    FORT EDWARD, NY	& HUDSON RIVER & 5  \\
    NORTH CREEK, NY	& HUDSON RIVER & 5 \\
    NEAR CARLISLE, SC	& BROAD RIVER & 3 \\
    NEAR BELL, GA	& BROAD RIVER &  3 \\
    \end{tabular}
\hrule
\end{table}

\begin{figure}[h!]
    \centering
    \includegraphics[width=0.8\textwidth]{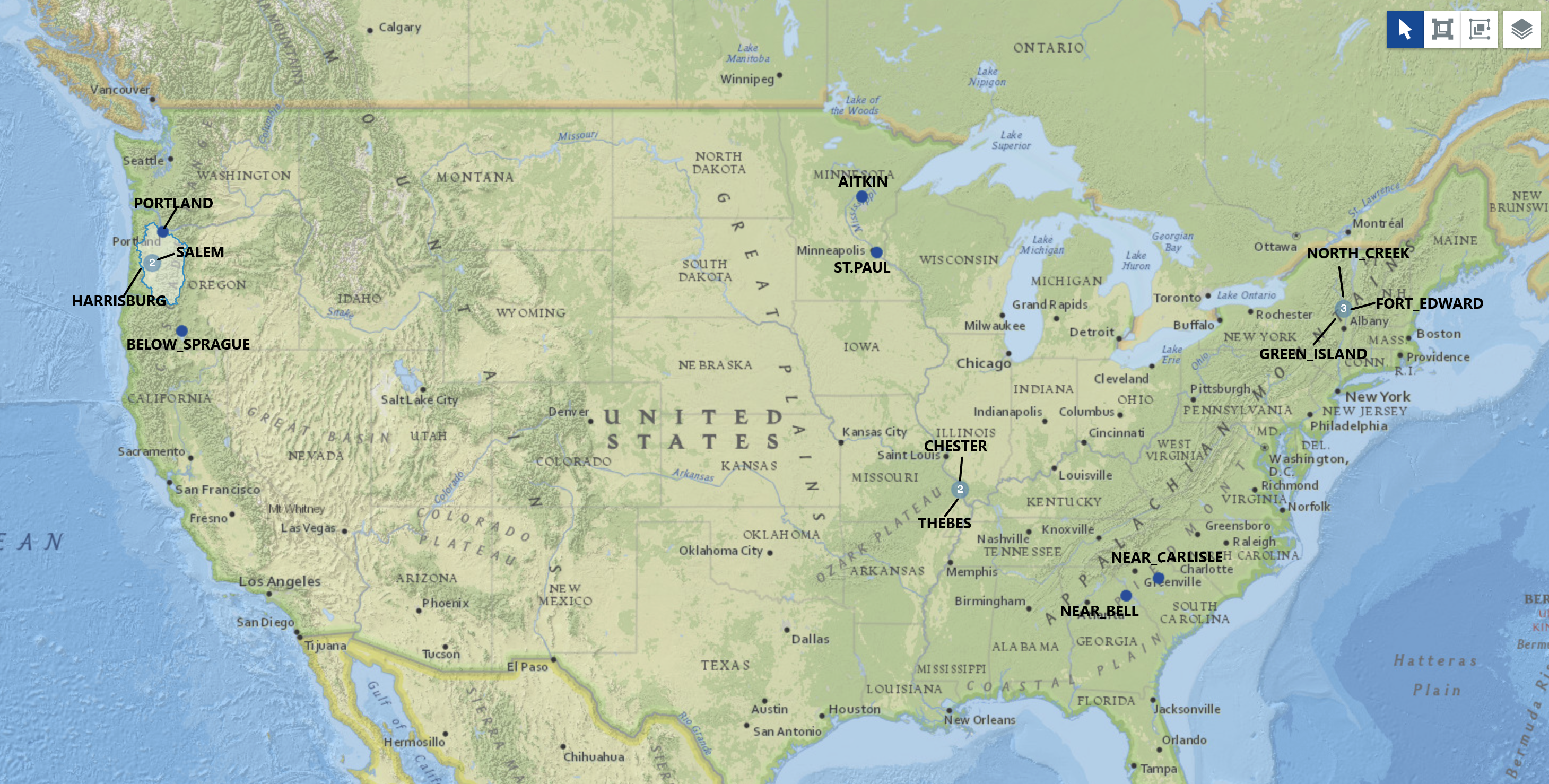}
    \caption{Geographical locations of the 13 river discharge stations.}
    \label{map}
\end{figure}

\begin{figure}[h]
    \centering
    \begin{minipage}{0.45\textwidth}
        \centering
        \includegraphics[width=\textwidth]{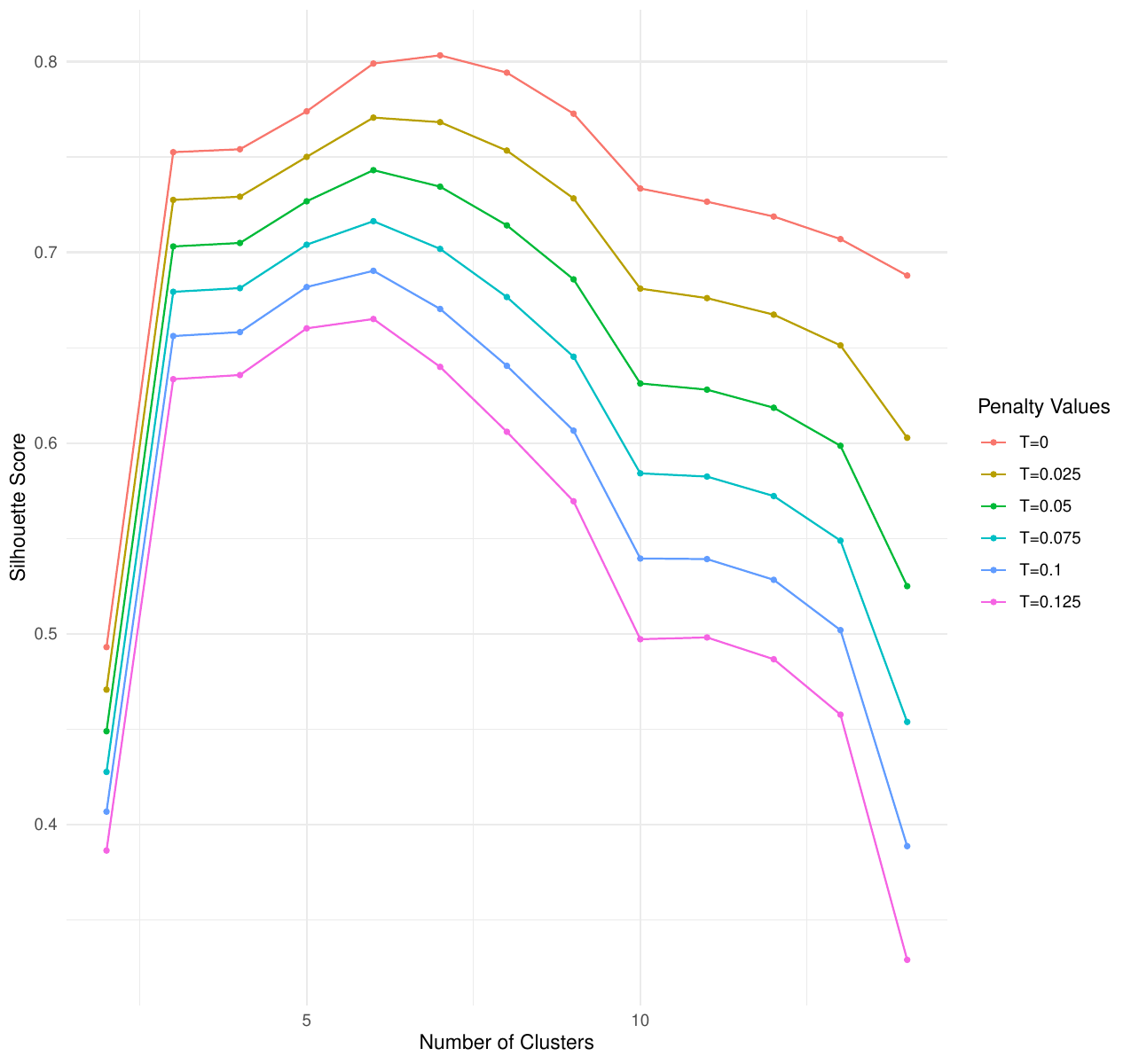} 
        \caption{Penalized ASW curves for river discharge data (top $10\%$ norms).}\label{silgrdc}
    \end{minipage}
    \hfill
    \begin{minipage}{0.45\textwidth}
        \centering
        \includegraphics[width=\textwidth]{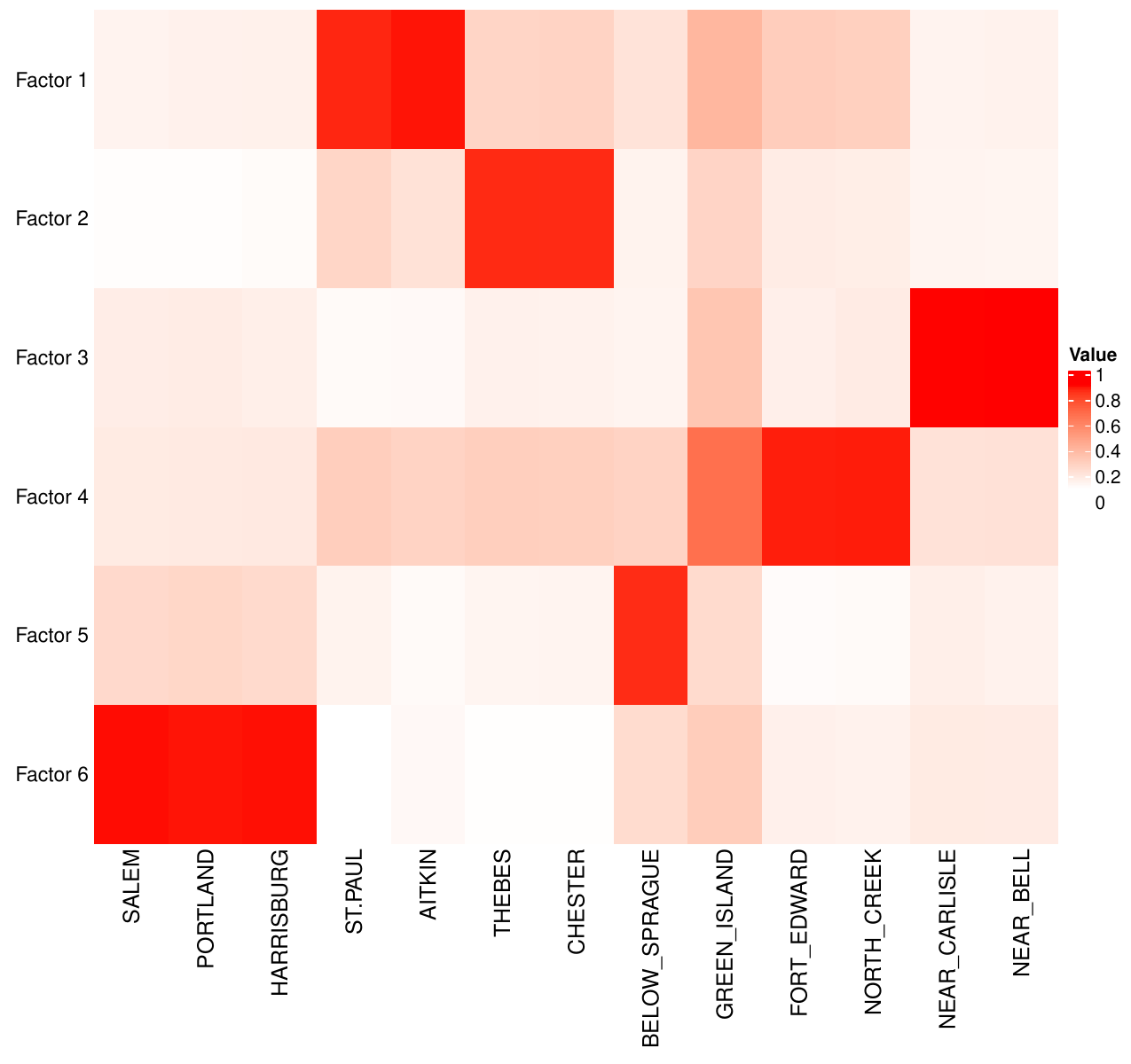} 
        \caption{Estimated $B^\top$ for river discharge data (top $10\%$ norms).}  \label{cgrdc}
    \end{minipage}
\end{figure}


As in the previous example,  Fig.\ \ref{silgrdc} presents the penalized ASW    curves,    from which we found that 6 seems to be an appropriate choice of order. Fig.\ \ref{cgrdc} illustrates the  factor matrix $B$ converted from the spectral estimation  following the method in Section \ref{sec:fact order coef est} with $\|\cdot \|_{(s)}=\|\cdot \|_{(r)}=\|\cdot \|_{2}$ and $\alpha=2$. In addition, for each row of the matrix $B$, we find to which factor index (the same as the cluster index in Fig.\ \ref{cgrdc}) the largest value  corresponds. We include these factor indices in the last column of Table \ref{tab:my_label}, which can be viewed roughly as markings of groups of concomitant extremes. {The results align well with the expected geographical context:} These 6 groups are in good accordance with the geographical context: most of the stations located along the same river are found in the same group, with the only exception of the 4 stations along the Mississippi River. The further division of these 4 stations into 2 groups may be easily justified by the large geographical distance  between the 2 groups: one group located in Minnesota (MN) and the other located in Illinois (IL).


 {In the supplementary material,  regarding the real data demonstrations presented above, we have included more results with different choices of extremal subsample sizes, or say, top norm percentages. For the air pollution data,  the results are relatively stable when the percentage varies. For the river discharge data, when the percentage is decreased to $5\%\sim 1\%$, the penalized ASW criterion starts to favor $k=5$. In this case, the station BELOW SPRAGUE RIVER NEAR CHILOQUIN labeled as group 2 in Table \ref{tab:my_label} is merged to the concomitant group labeled as 4 formed by SALEM, PORTLAND and HARRISBURG, all of which are geographically  close (see Fig.\ \ref{map}).    
}
 
\section{{Summary and Discussion}}\label{sec:Summary}
Following recent developments in literature, we explore the estimation of multivariate extreme models with a discrete spectral measure by employing spherical clustering techniques. Our main contribution is a method for selecting the order, that is, the number of clusters, which consistently identifies the true number of spectral atoms. Specifically, we introduce an additional penalty term to the well-known simplified average silhouette width, which penalizes small cluster sizes and small dissimilarities between cluster centers.  As a by-product, we offer an approach to determining the order of a max-linear factor model. This method is not only straightforward to implement, but also performs effectively in practical applications.
We also conduct a large-deviation-type analysis for estimating the discrete spectral measure through clustering methods. This analysis provides insights into the convergence quality of clustering-based estimation for multivariate extremes. In addition, we show how these estimates can be applied to the parameter estimation of heavy-tailed factor models.

At last, we point out several potential future work directions. First, the tuning parameter $t$ for the penalty \eqref{eq:pen} is  chosen with the help of visual inspection. It is of interest to explore the data-driven method to guide the choice of $t$, which may require a refined understanding of the consistency result in Theorem \ref{Thm:sil cons}. In addition, it may be worth exploring alternative clustering assessment criteria to ASW in \eqref{eq:ASW}. Some preliminary experiments show that the cross-validation method based on algorithmic instability of \cite{wang2010consistent} might also be a competent criterion in the context of clustering multivariate extremes.  It would also be desirable to develop a method for choosing the threshold $l_n$ in \eqref{eq:W_n} in the context of clustering multivariate extremes.  The recent work \cite{wan2019threshold} might be relevant in this regard.  
Furthermore, it may be of interest to enhance the large-deviation-type analysis in Section \ref{sec:rate} to a finite-sample probabilistic guarantee bound.   

\medskip

\noindent\textbf{Acknowledgment.} We thank two anonymous reviewers and the associate editor for their helpful comments that lead to substantial improvements of the article.

\appendix
\section*{Appendix: Supplementary Material}

In this supplementary material, we provide results in addition to \cite{deng2024estimation}. In particular, additional simulations studies and real data demonstrations with different choices of extremal subsample thresholds are included.

\section{Simulation studies}
The results here correspond to Section 6.1 of \cite{deng2024estimation}. We follow the same simulation setup there and the datasets are generated separately independently. The spherical $k$-means and spherical $k$-principal component methods are applied using varying proportions of data with top norms: 1\%, 5\%, 10\%, and 15\%, corresponding to effective sample sizes of 100, 500, 1000, and 1500, respectively. The results align with the \cite{deng2024estimation}, showing that the proposed penalization has a significant bias correction effect on the selection of the optimal \( k \). 
\subsection{Case: $d=4,k=2$}
\begin{figure}[h]
    \centering
    \includegraphics[width=0.55\textwidth]{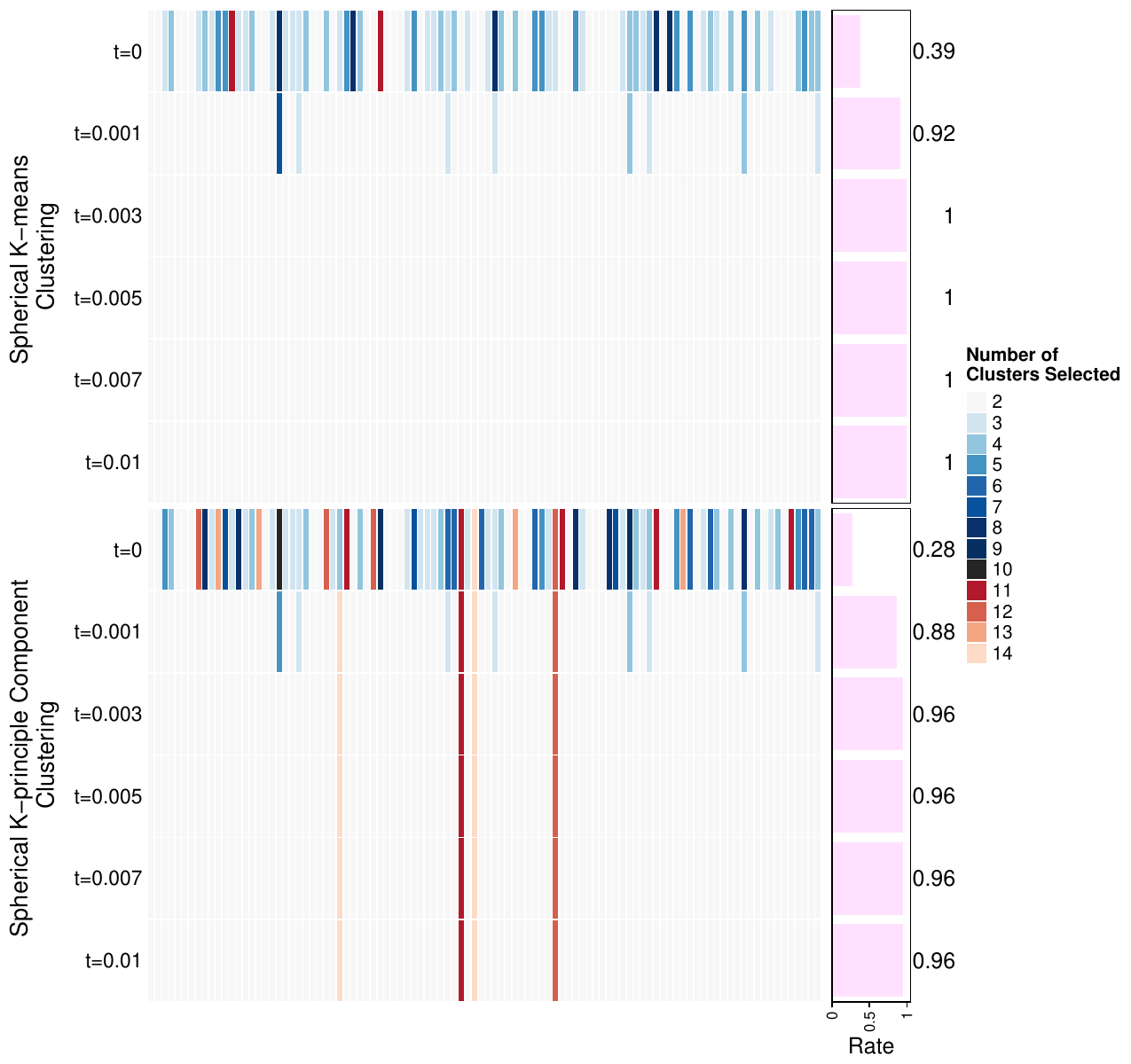}
    \caption{Simulation result visualization for the setup $d=4,k=2$ in selection of 1\% data}
    \label{Pd4k2_1}
\end{figure}

\begin{figure}[h]
    \centering
    \includegraphics[width=0.55\textwidth]{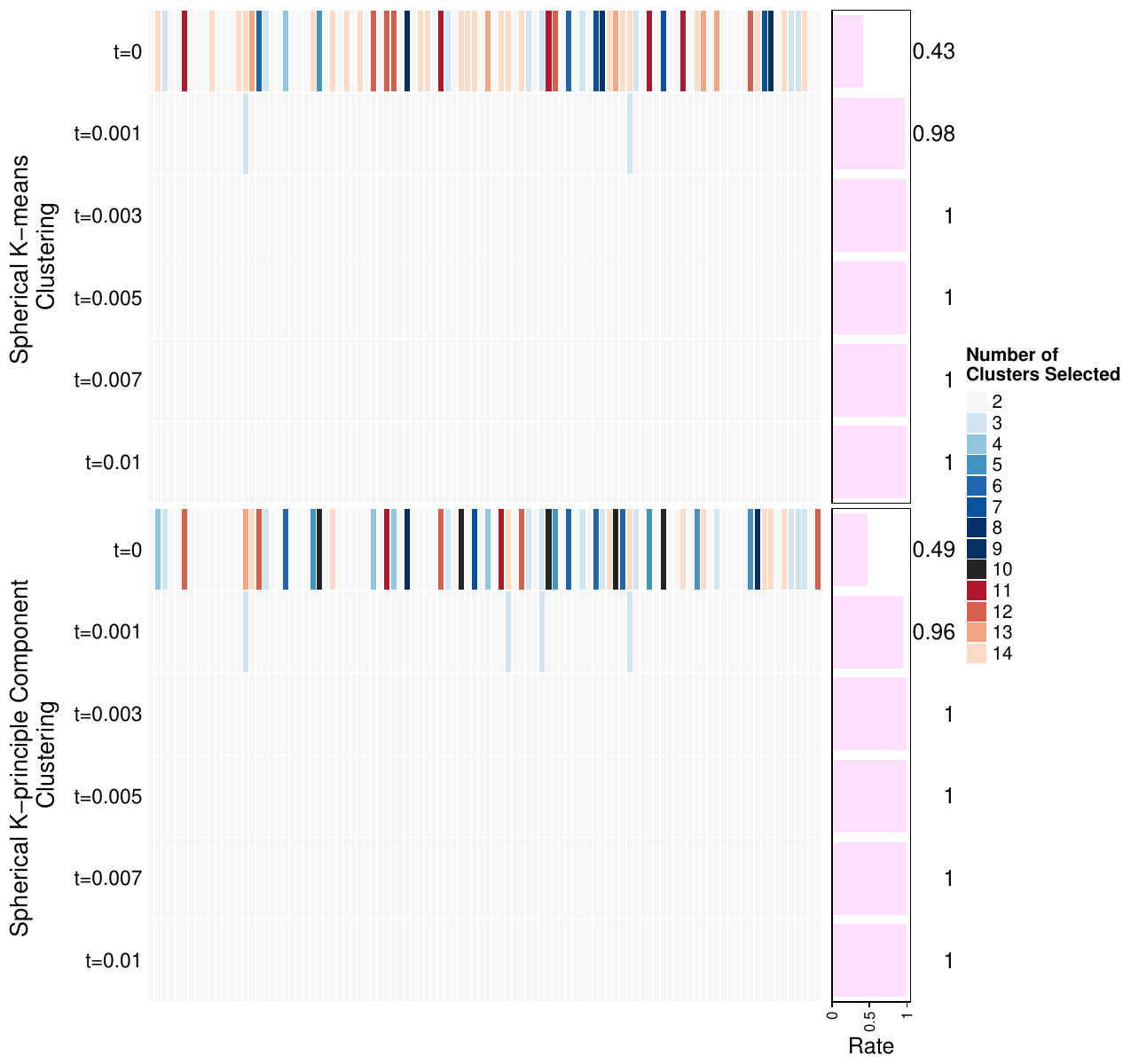}
    \caption{Simulation result visualization for the setup $d=4,k=2$ in selection of 5\% data}
    \label{Pd4k2_5}
\end{figure}

\begin{figure}[h]
    \centering
    \includegraphics[width=0.55\textwidth]{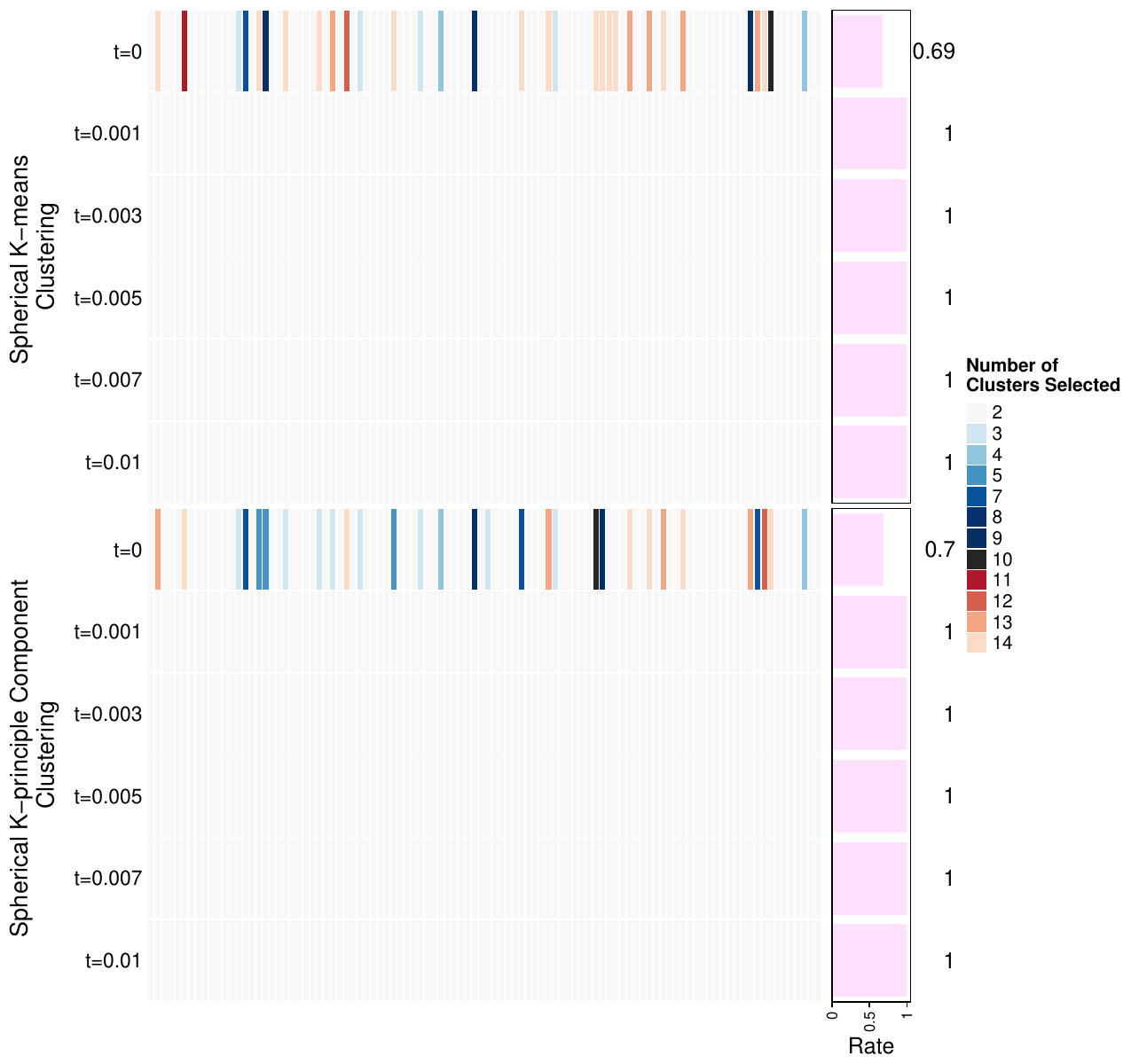}
    \caption{Simulation result visualization for the setup $d=4,k=2$ in selection of 10\% data}
    \label{Pd4k2_10}
\end{figure}

\begin{figure}[h]
    \centering
    \includegraphics[width=0.55\textwidth]{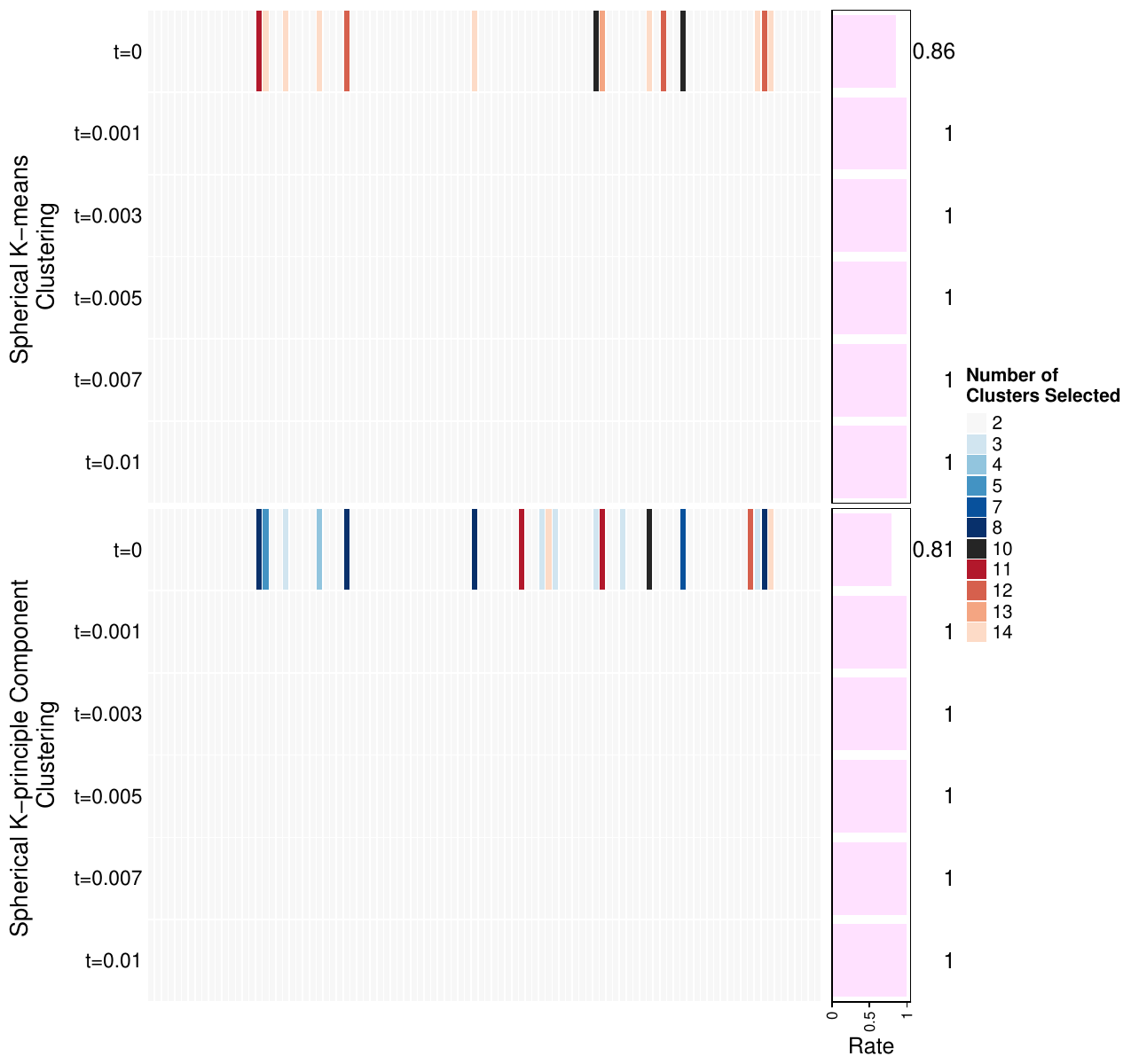}
    \caption{Simulation result visualization for the setup $d=4,k=2$ in selection of 15\% data}
    \label{Pd4k2_15}
\end{figure}
\clearpage

\subsection{Case: $d=4,k=6$}
\begin{figure}[h]
    \centering
    \includegraphics[width=0.55\textwidth]{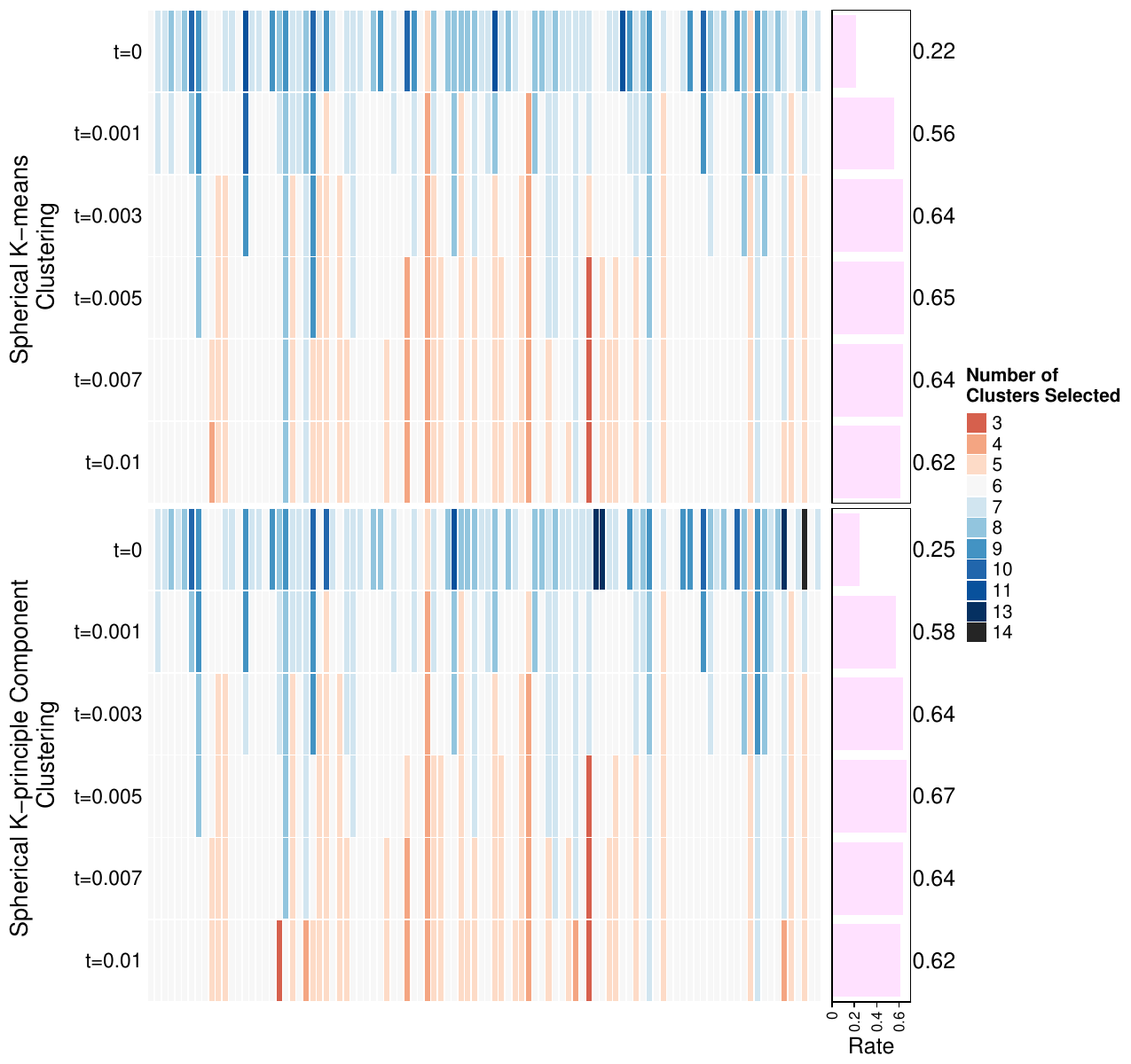}
    \caption{Simulation result visualization for the setup $d=4,k=6$ in selection of 1\% data}
    \label{Pd4k6_1}
\end{figure}

\begin{figure}[h]
    \centering
    \includegraphics[width=0.55\textwidth]{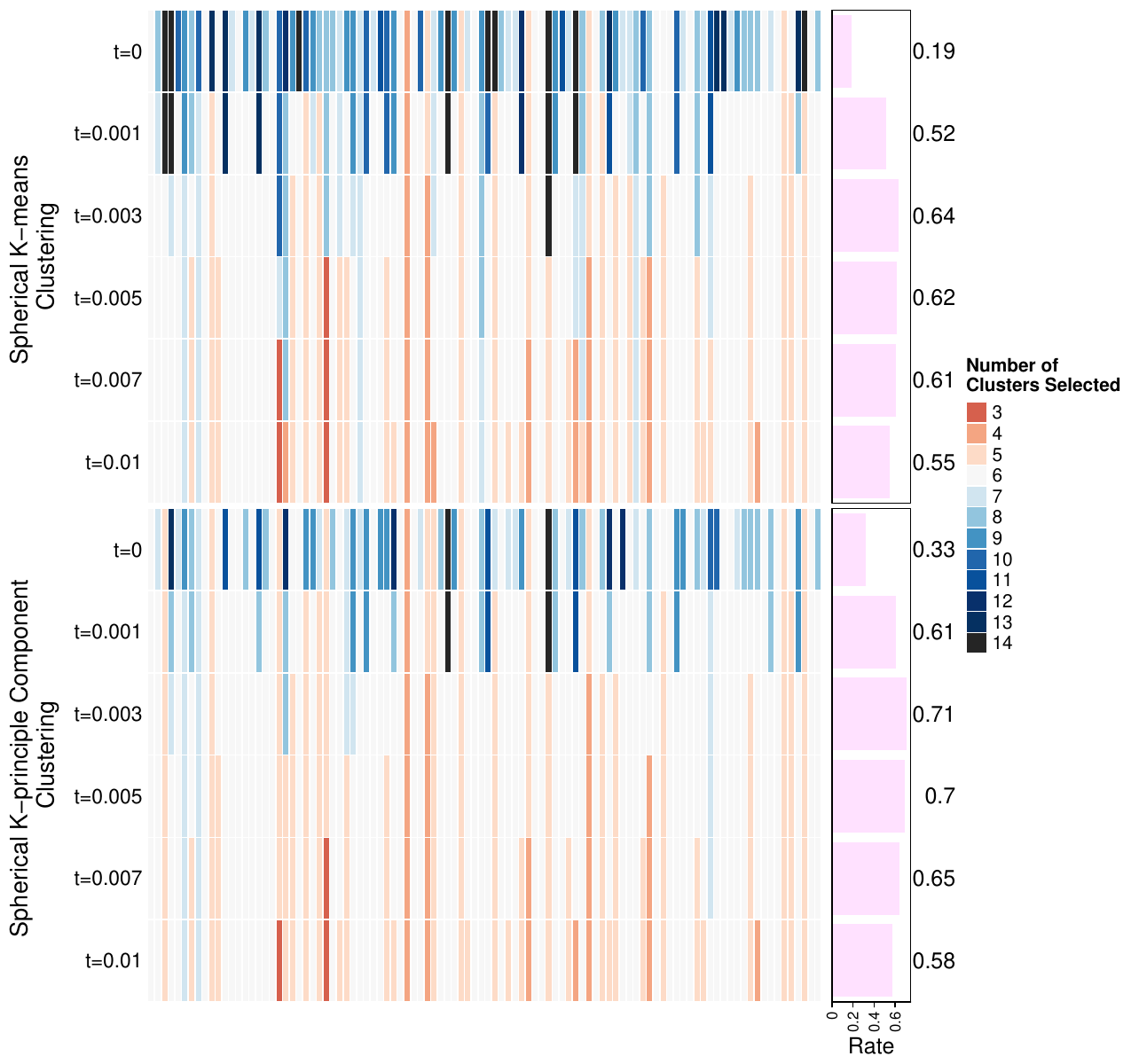}
    \caption{Simulation result visualization for the setup $d=4,k=6$ in selection of 5\% data}
    \label{Pd4k6_5}
\end{figure}

\begin{figure}[h]
    \centering
    \includegraphics[width=0.55\textwidth]{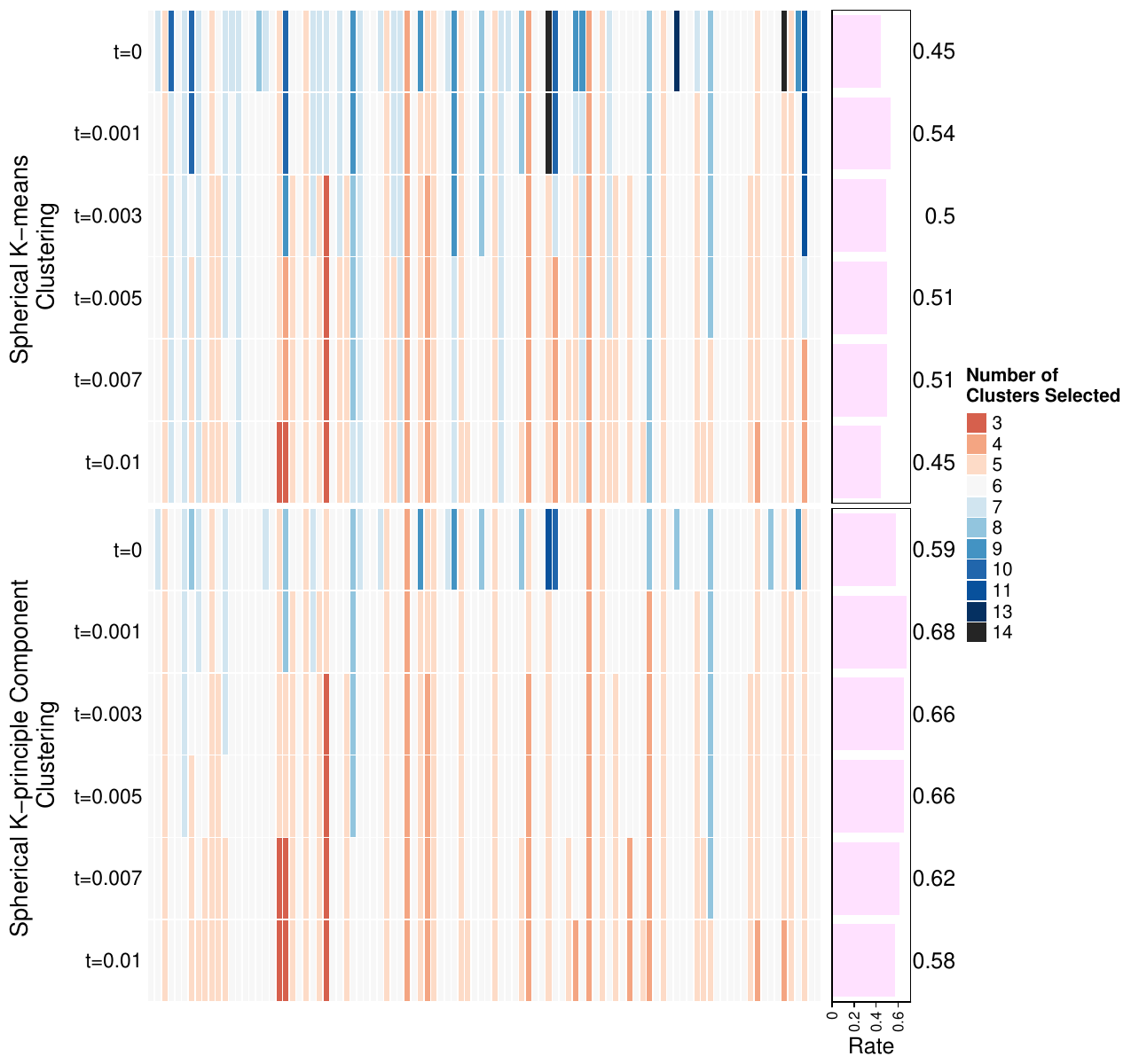}
    \caption{Simulation result visualization for the setup $d=4,k=6$ in selection of 10\% data}
    \label{Pd4k6_10}
\end{figure}

\begin{figure}[h]
    \centering
    \includegraphics[width=0.55\textwidth]{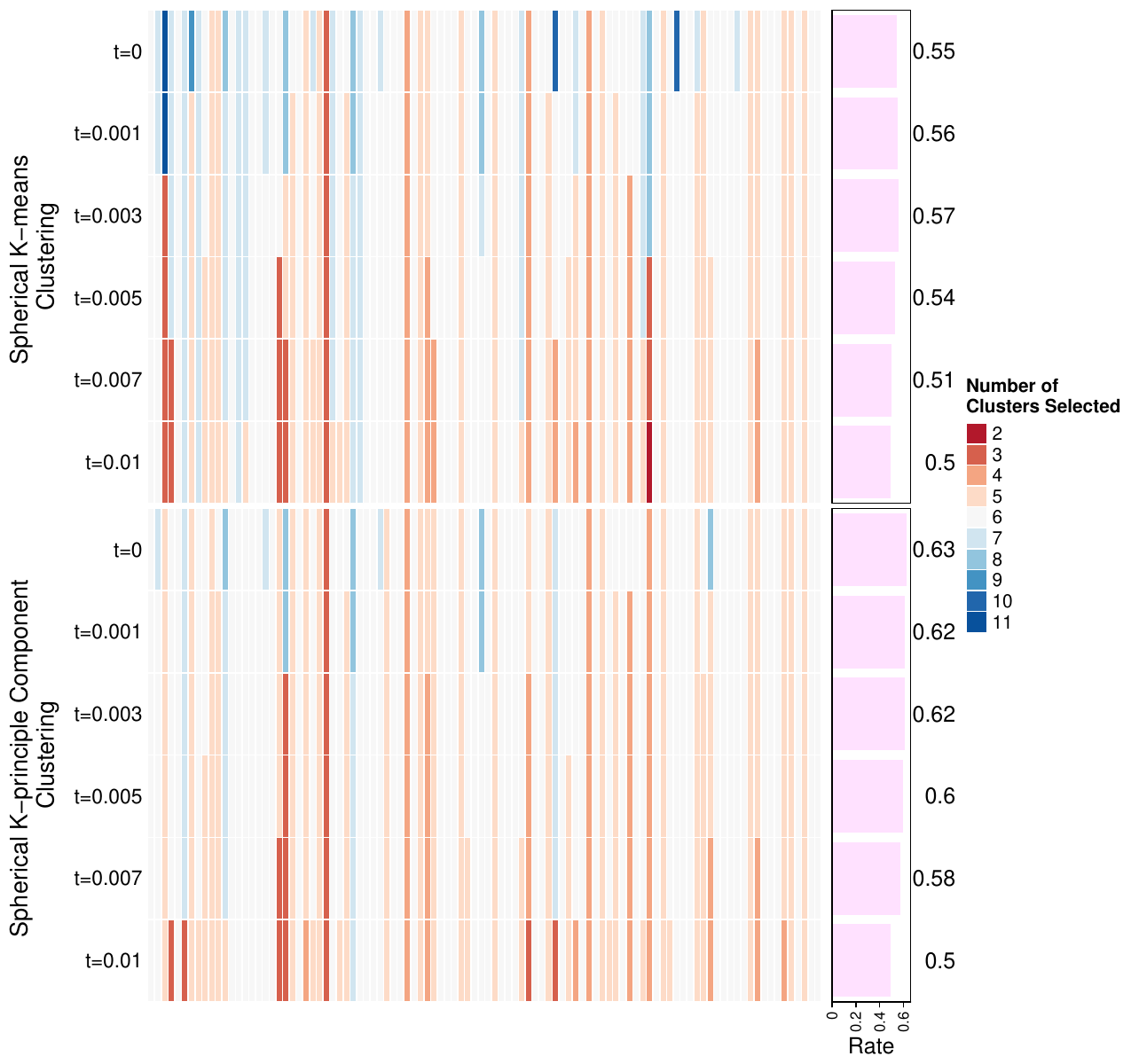}
    \caption{Simulation result visualization for the setup $d=4,k=6$ in selection of 15\% data}
    \label{Pd4k6_15}
\end{figure}
\clearpage

\subsection{Case: $d=6,k=6$:}
\begin{figure}[h]
    \centering
    \includegraphics[width=0.55\textwidth]{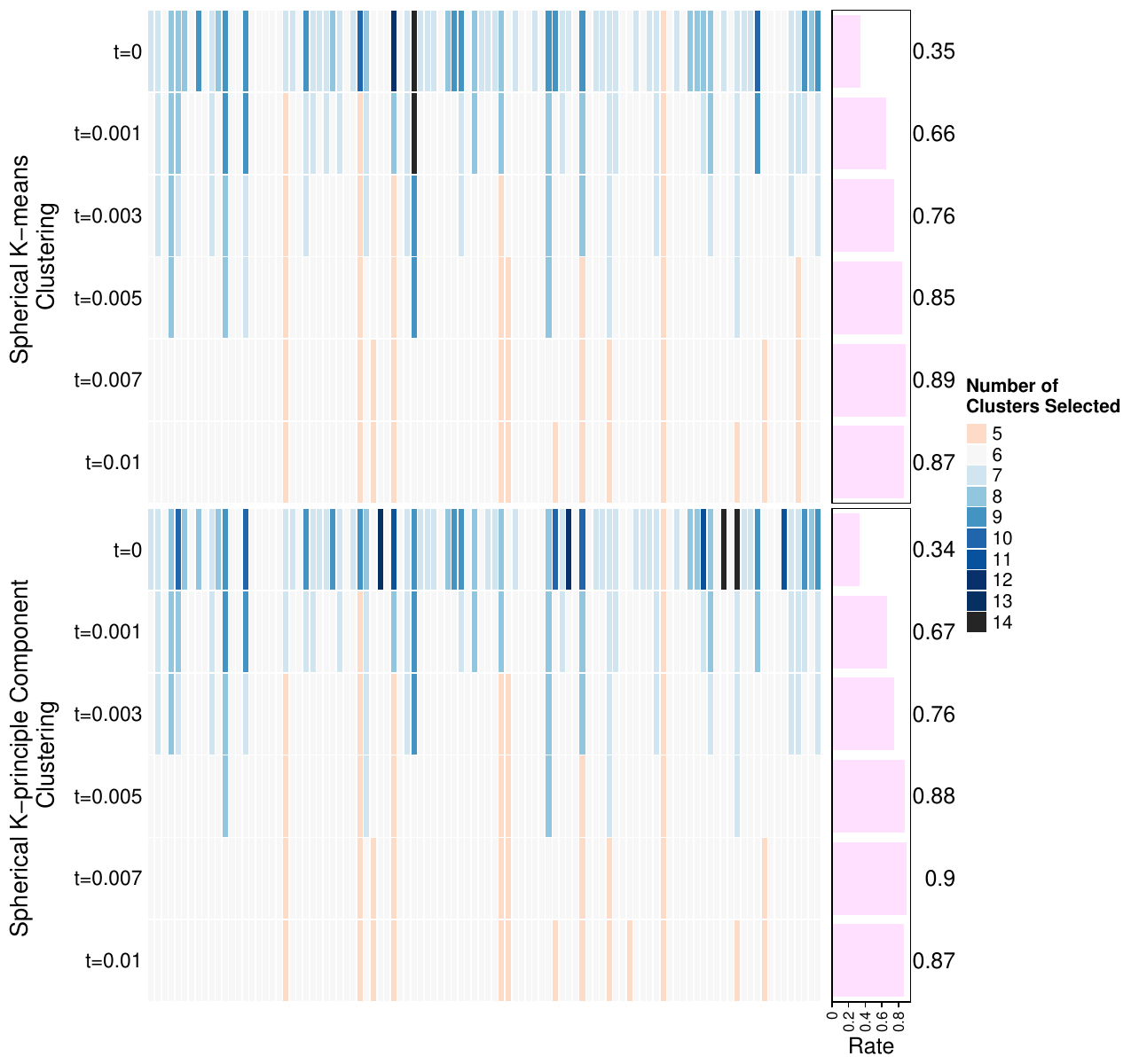}
    \caption{Simulation result visualization for the setup $d=6,k=6$ in selection of 1\% data}
    \label{Pd6k6_1}
\end{figure}

\begin{figure}[h]
    \centering
    \includegraphics[width=0.55\textwidth]{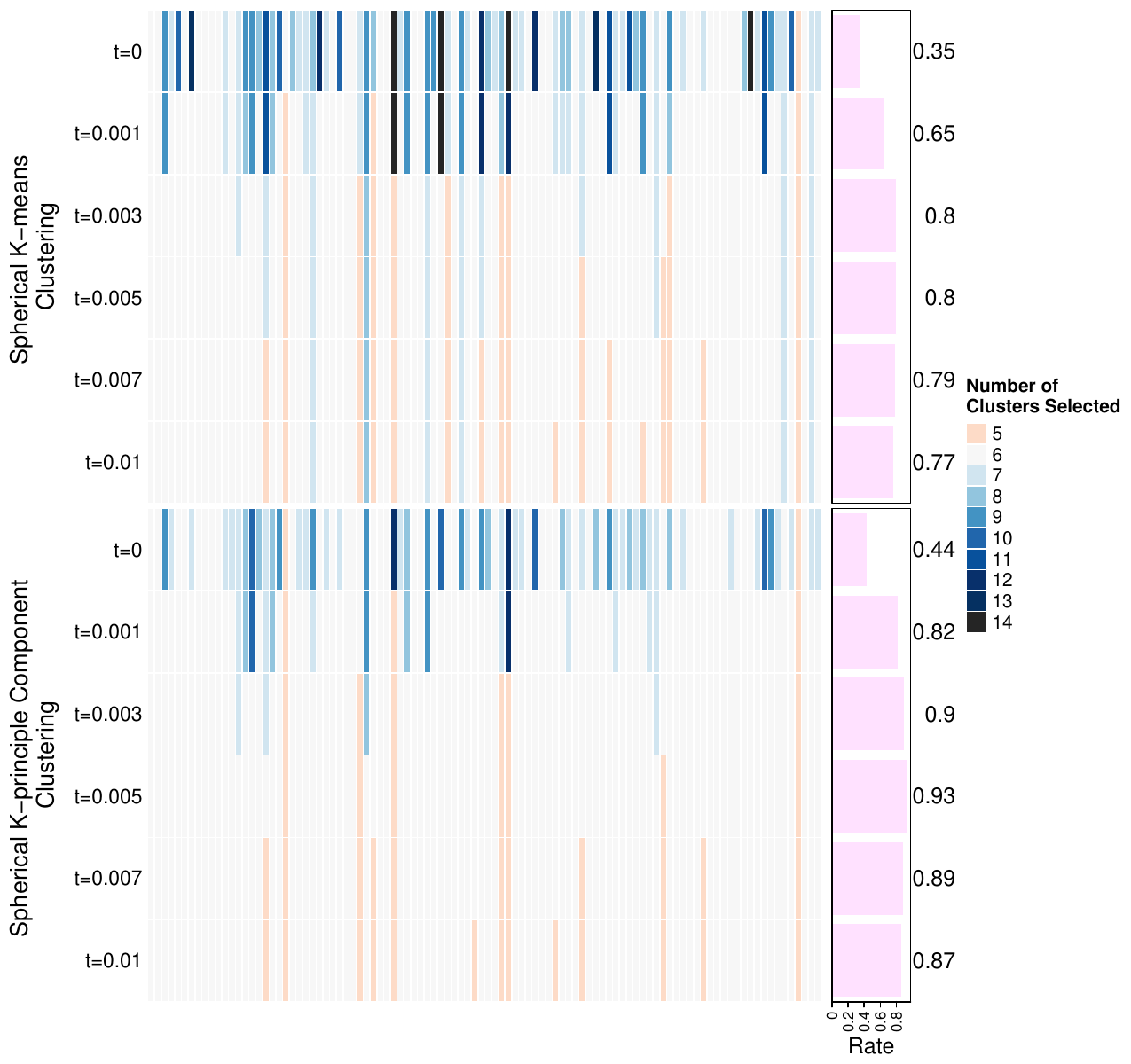}
    \caption{Simulation result visualization for the setup $d=6,k=6$ in selection of 5\% data}
    \label{Pd6k6_5}
\end{figure}

\begin{figure}[h]
    \centering
    \includegraphics[width=0.55\textwidth]{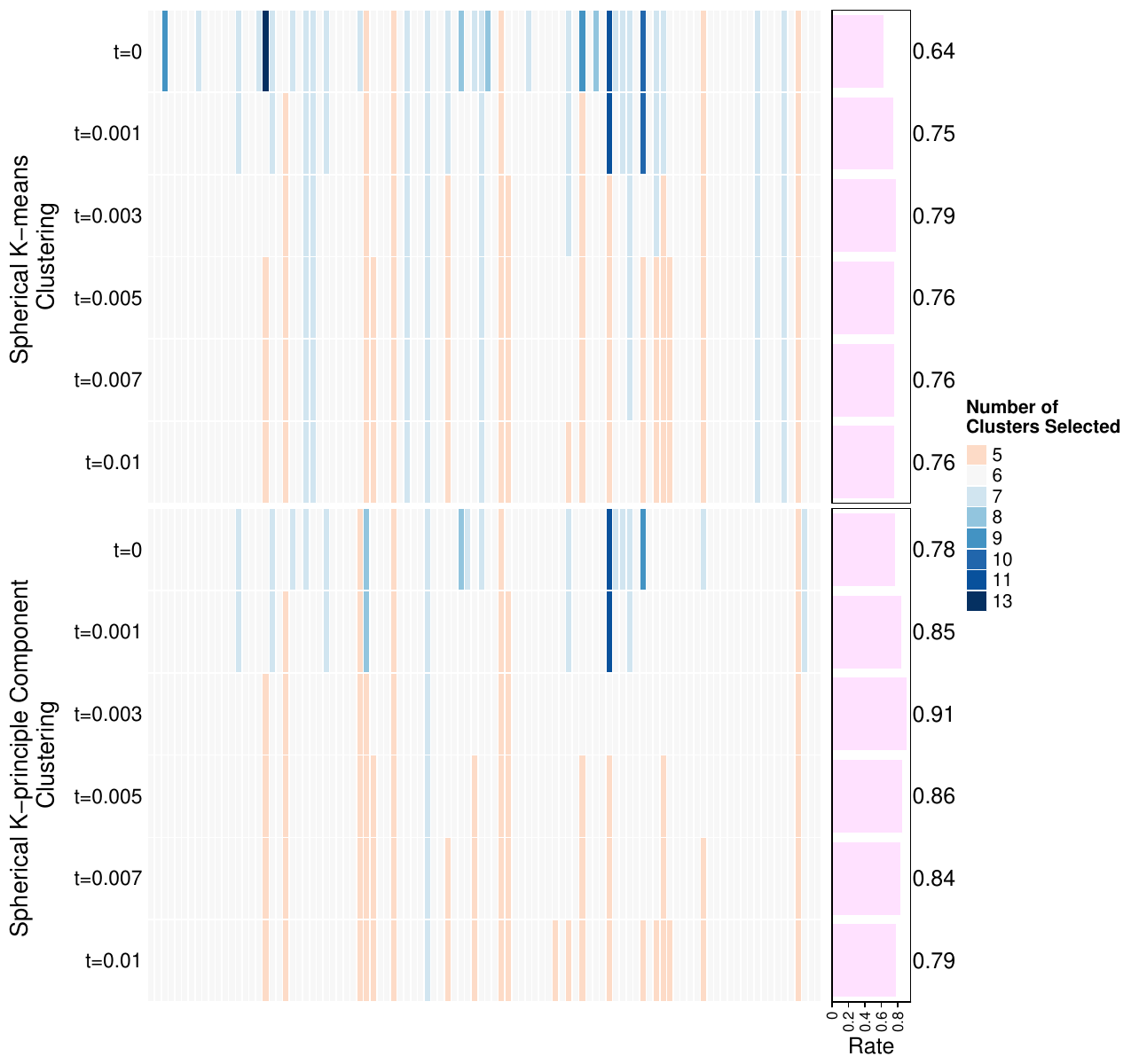}
    \caption{Simulation result visualization for the setup $d=6,k=6$ in selection of 10\% data}
    \label{Pd6k6_10}
\end{figure}

\begin{figure}[h]
    \centering
    \includegraphics[width=0.55\textwidth]{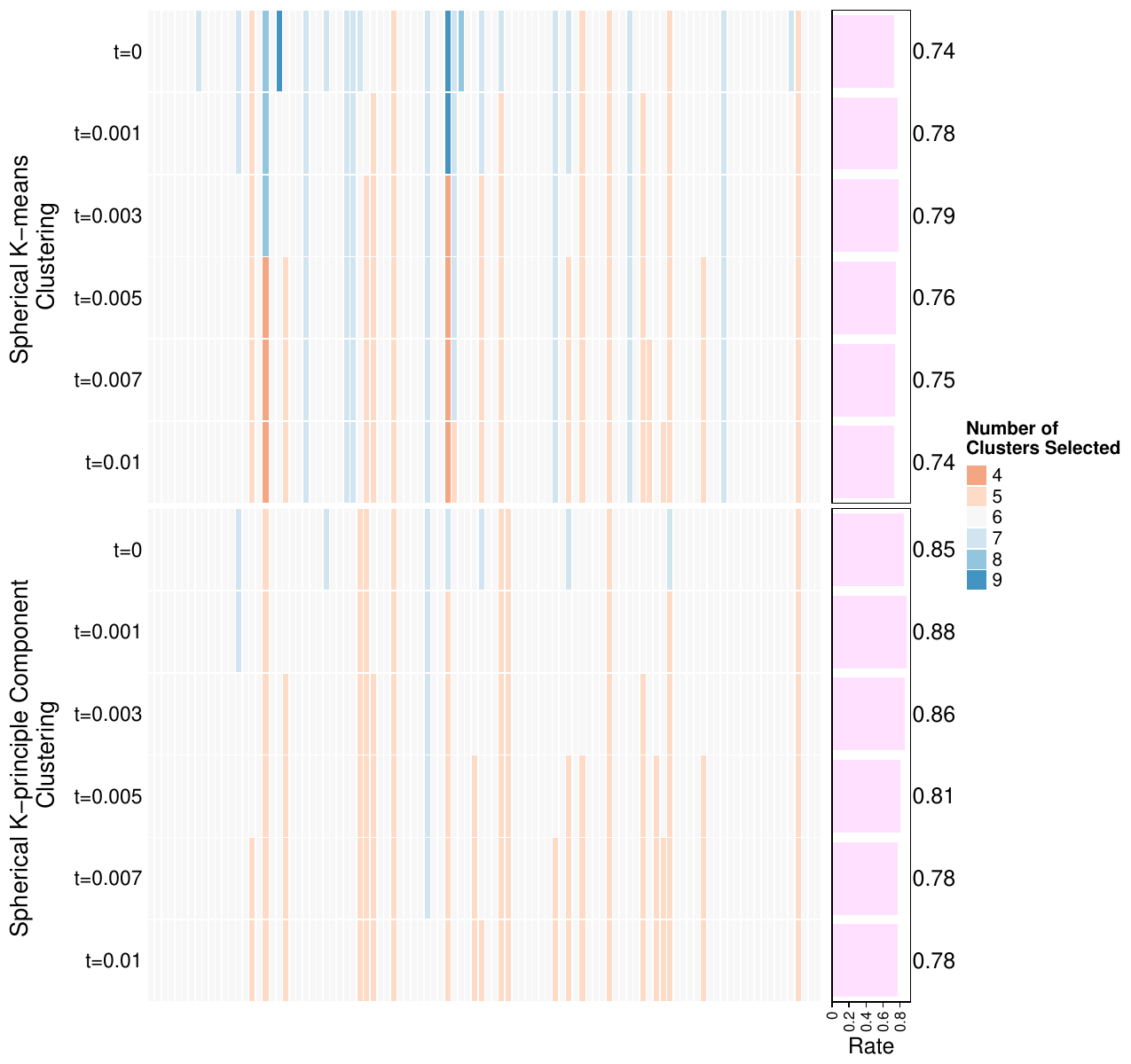}
    \caption{Simulation result visualization for the setup $d=6,k=6$ in selection of 15\% data}
    \label{Pd6k6_15}
\end{figure}

\clearpage

\subsection{Case: $d=10,k=6$:}
\begin{figure}[h]
    \centering
    \includegraphics[width=0.55\textwidth]{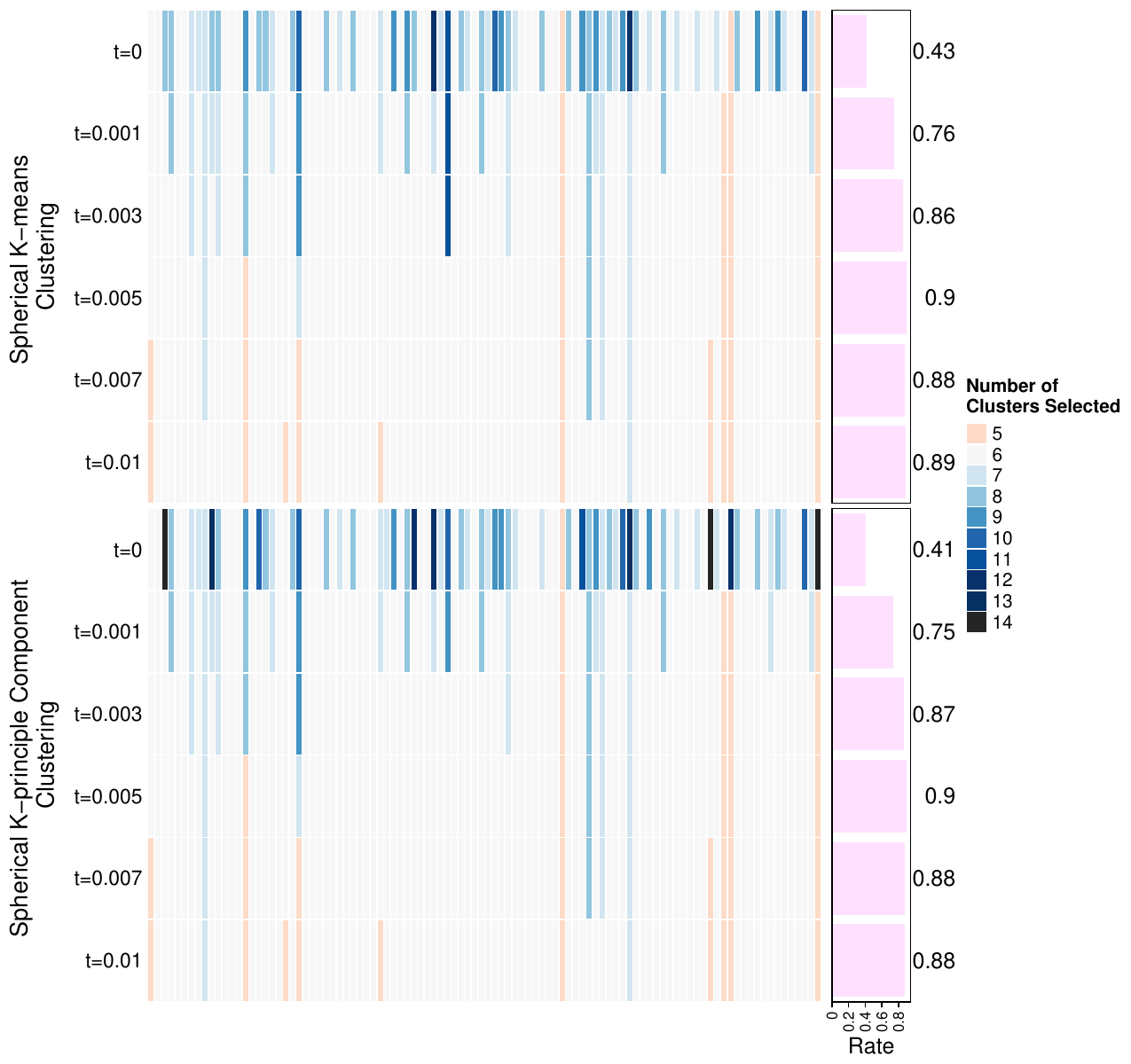}
    \caption{Simulation result visualization for the setup $d=10,k=6$ in selection of 1\% data}
    \label{Pd10k6_1}
\end{figure}

\begin{figure}[h]
    \centering
    \includegraphics[width=0.55\textwidth]{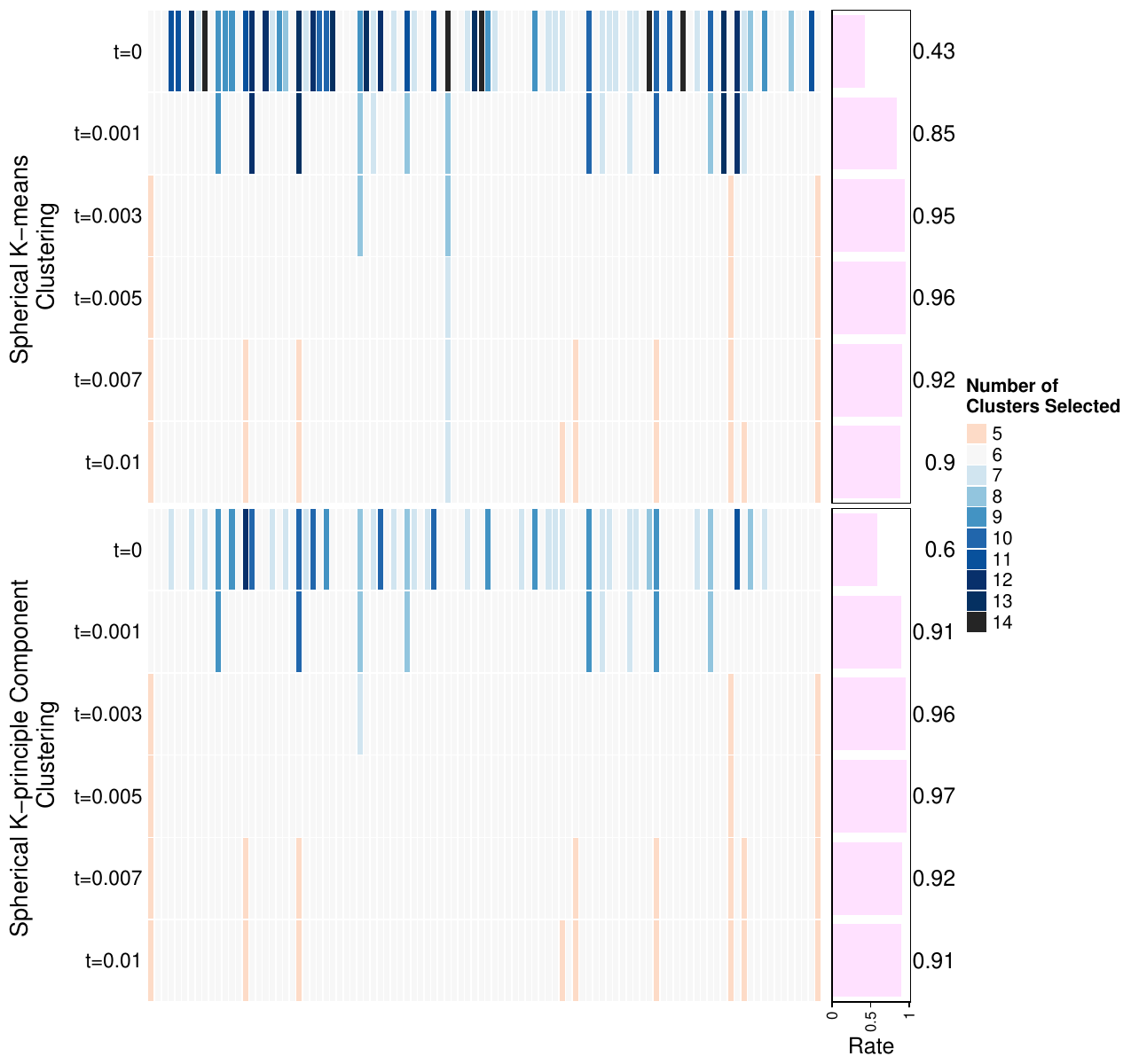}
    \caption{Simulation result visualization for the setup $d=10,k=6$ in selection of 5\% data}
    \label{Pd10k6_5}
\end{figure}

\begin{figure}[h]
    \centering
    \includegraphics[width=0.55\textwidth]{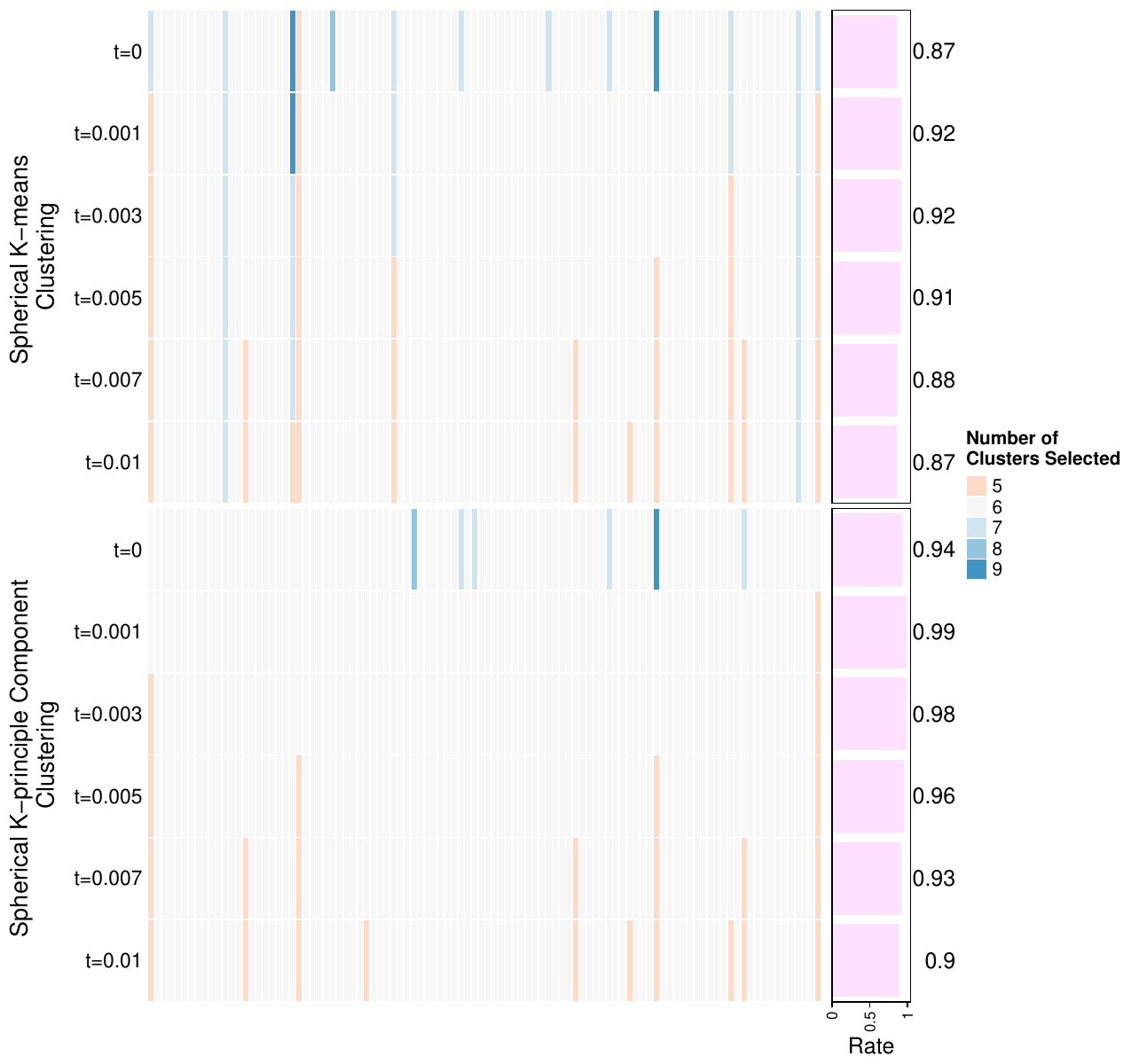}
    \caption{Simulation result visualization for the setup $d=10,k=6$ in selection of 10\% data}
    \label{Pd10k6_10}
\end{figure}

\begin{figure}[h]
    \centering
    \includegraphics[width=0.55\textwidth]{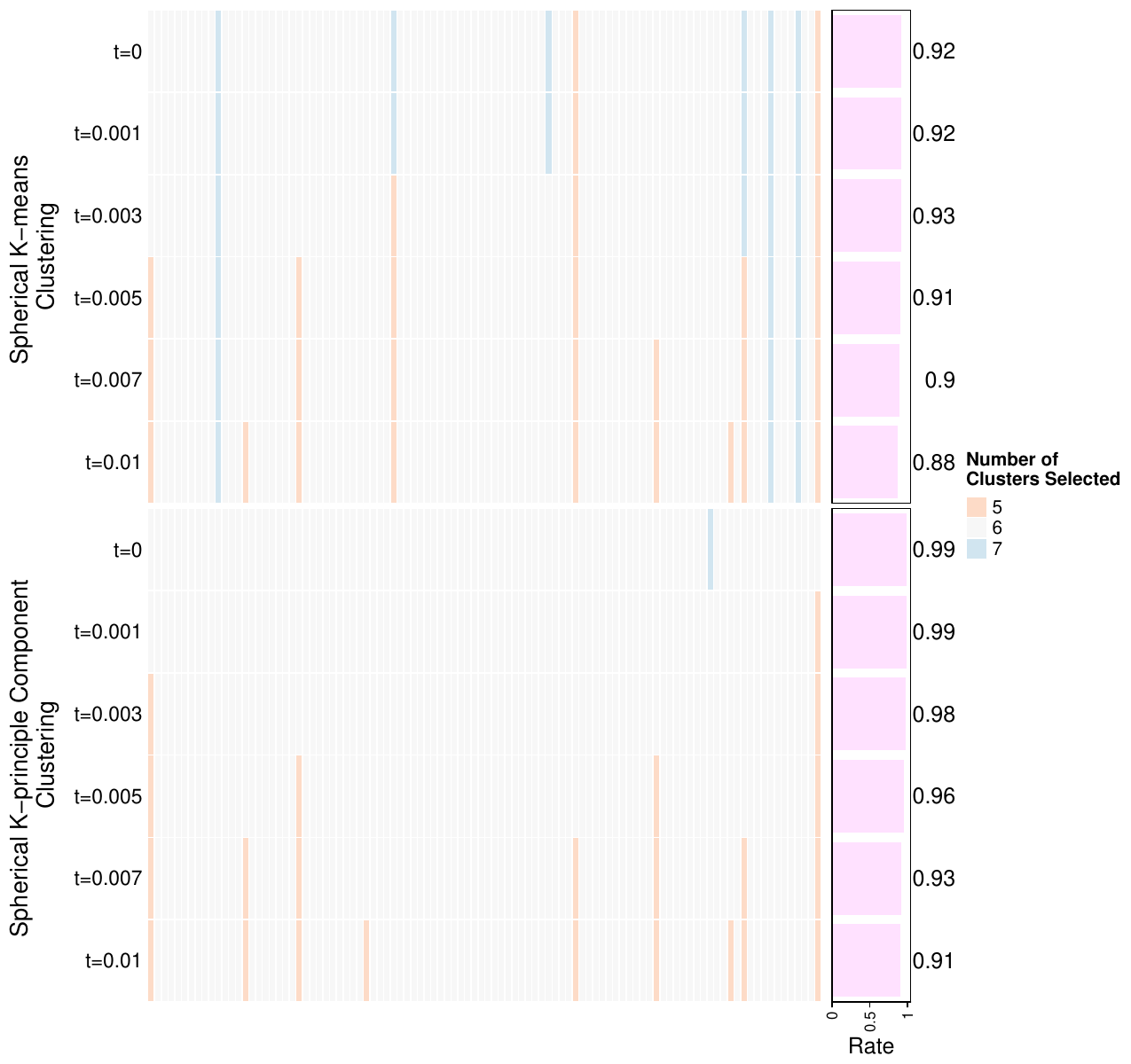}
    \caption{Simulation result visualization for the setup $d=10,k=6$ in selection of 15\% data}
    \label{Pd10k6_15}
\end{figure}

\clearpage
\section{Real data demonstrations}
The results here correspond to Section 6.2 of \cite{deng2024estimation}, and we follow exactly the same analysis except for different choices of extremal subsample percentages.  The estimated $B^\top$ is computed based on the optimal $k$ chosen by inspecting the bending behavior of the corresponding penalized ASW curves.

\subsection{Air Pollution Data}
The results here correspond to Section 6.2.1 of \cite{deng2024estimation}.

\begin{figure}[h]
    \centering
    \begin{minipage}{0.45\textwidth}
        \centering
        \includegraphics[width=\textwidth]{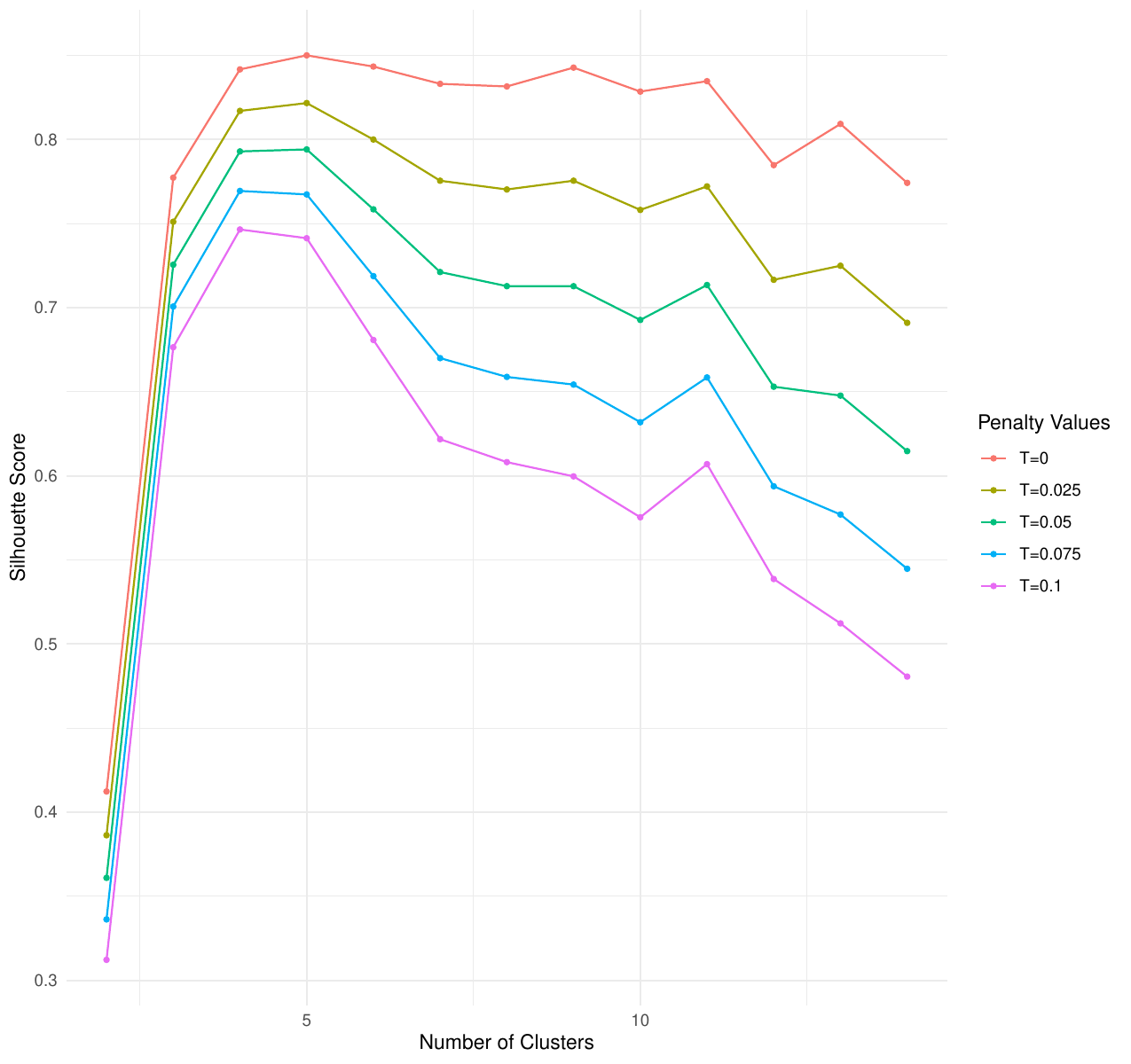} 
        \caption{Penalized ASW Curves for Summer Air Pollution Data (10\%)}
    \end{minipage}
    \hfill
    \begin{minipage}{0.45\textwidth}
        \centering
        \includegraphics[width=\textwidth]{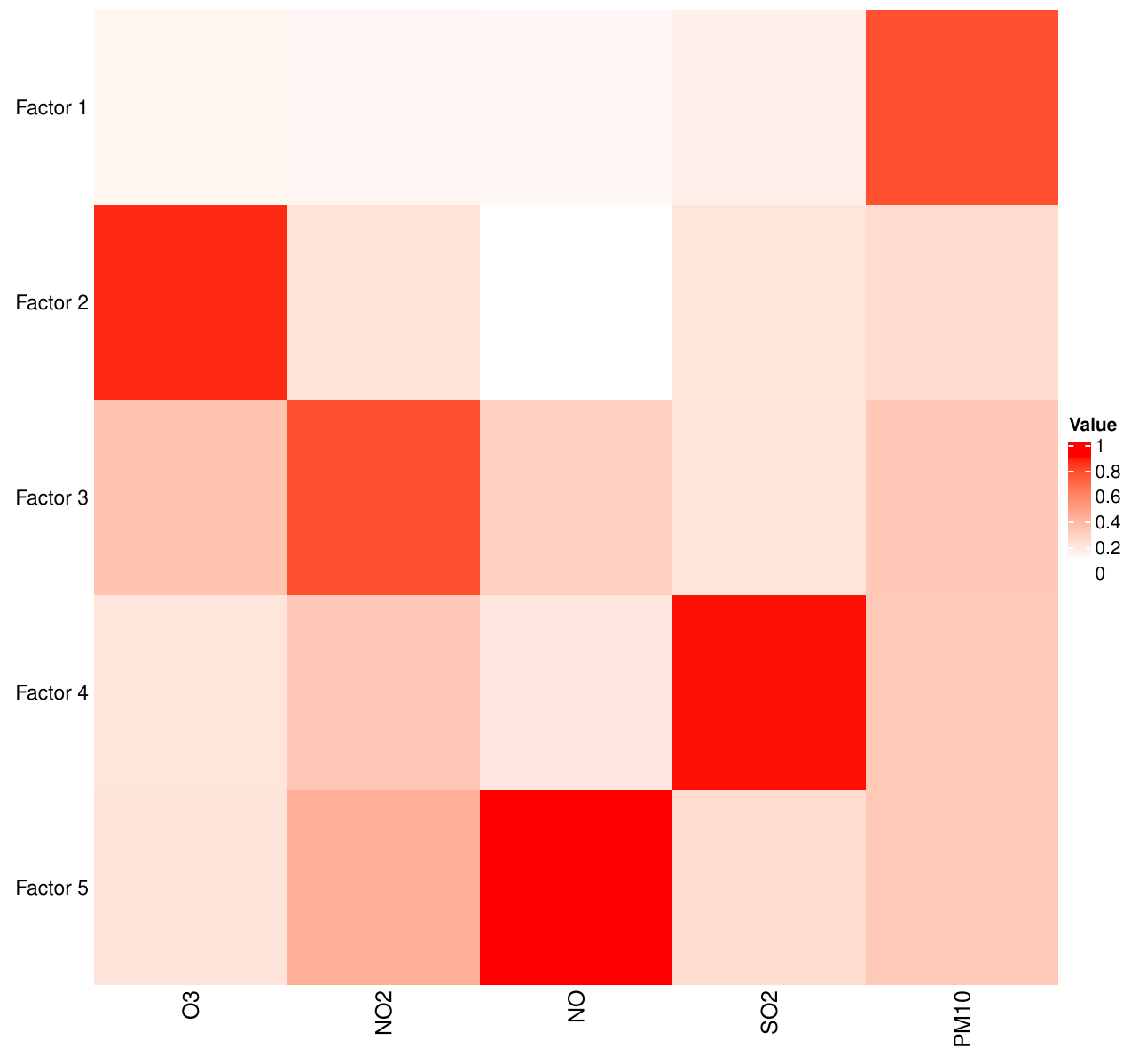} 
        \caption{Estimated $B^\top$ (10\%)}
    \end{minipage}
\end{figure}

\begin{figure}[h]
    \centering
    \begin{minipage}{0.45\textwidth}
        \centering
        \includegraphics[width=\textwidth]{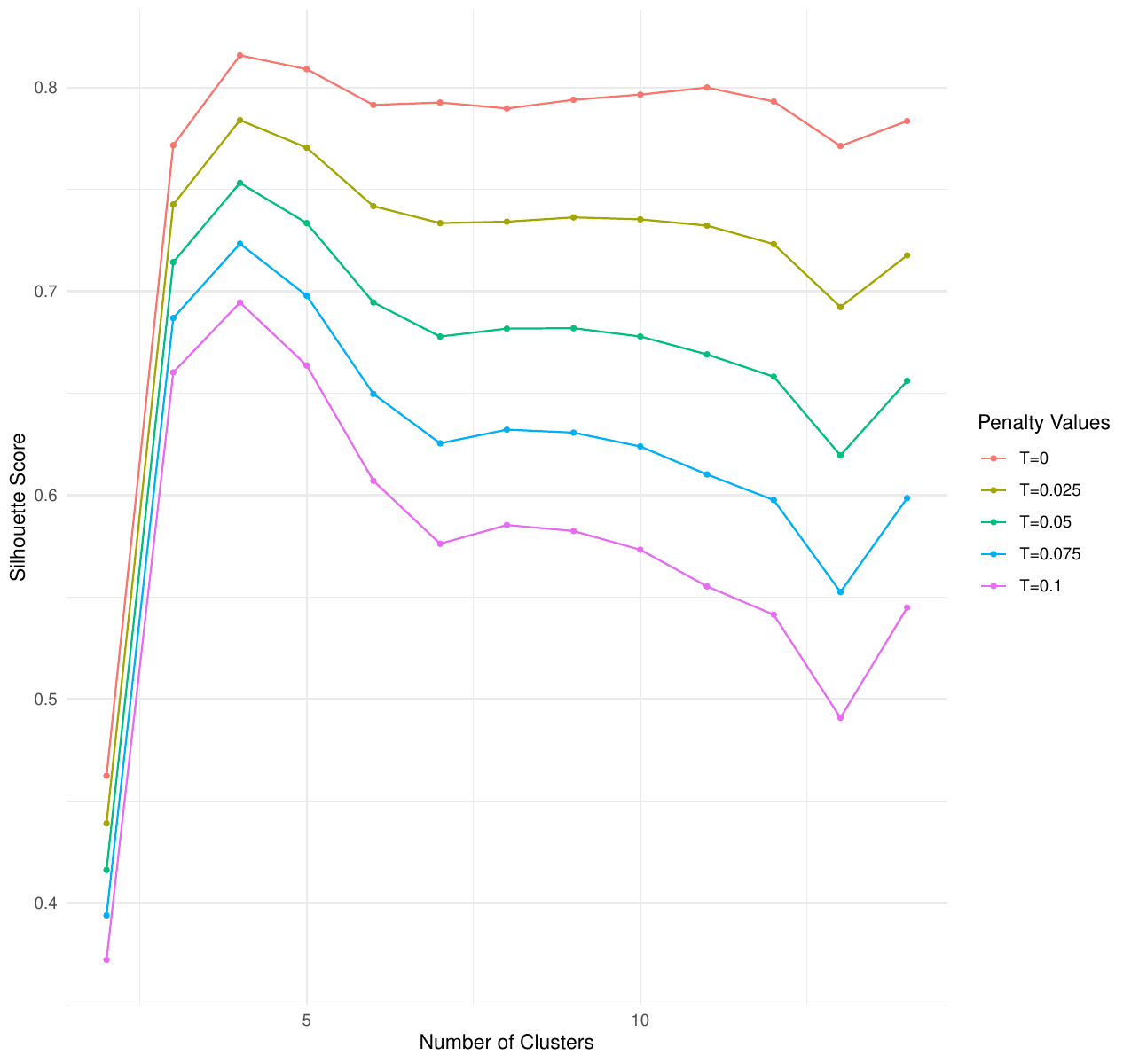} 
        \caption{Penalized ASW Curves for Summer Air Pollution Data (15\%)}
    \end{minipage}
    \hfill
    \begin{minipage}{0.45\textwidth}
        \centering
        \includegraphics[width=\textwidth]{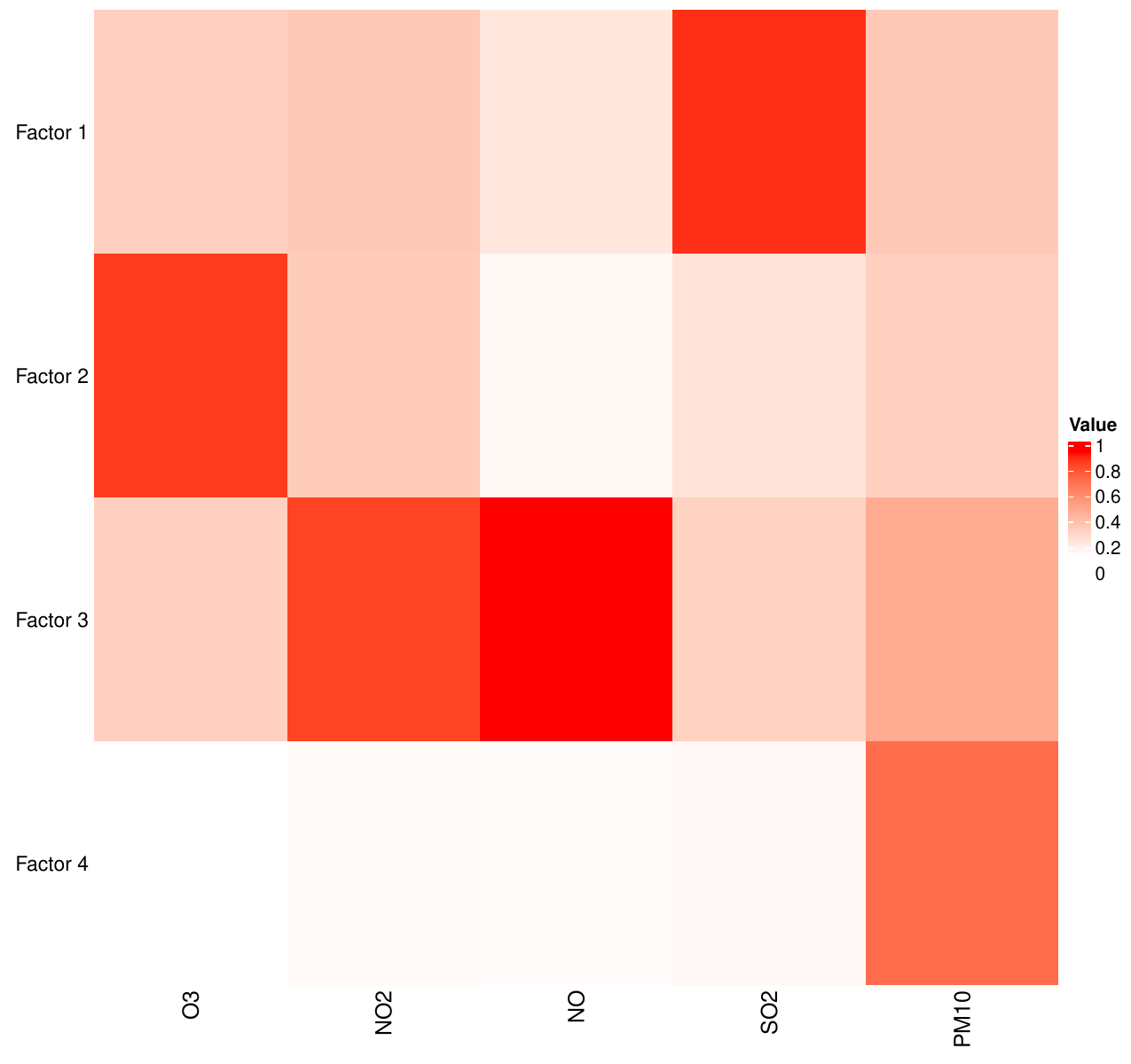} 
        \caption{Estimated $B^\top$ (15\%)}
    \end{minipage}
\end{figure}

\begin{figure}[h]
    \centering
    \begin{minipage}{0.45\textwidth}
        \centering
        \includegraphics[width=\textwidth]{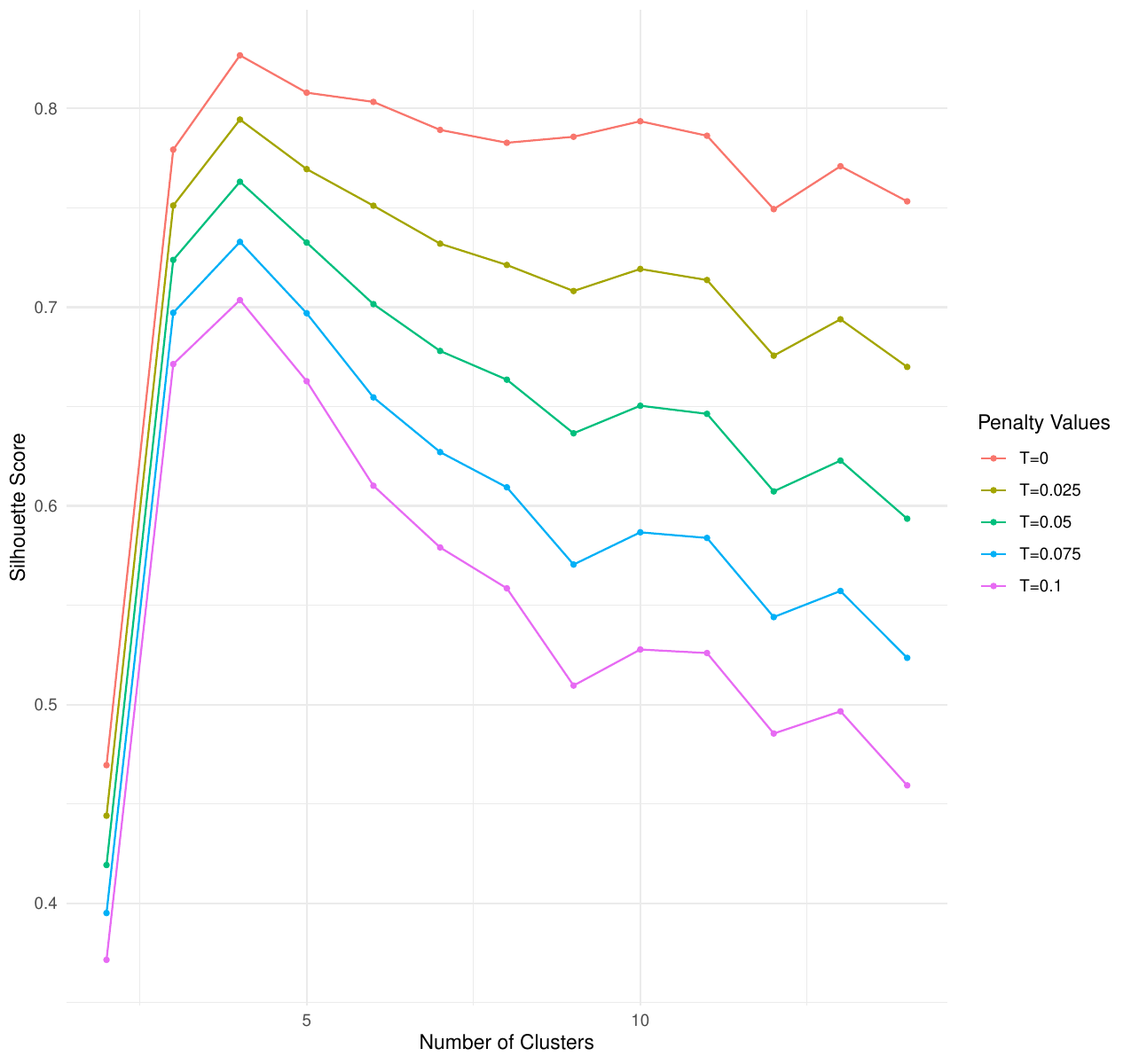} 
        \caption{Penalized ASW Curves for Summer Air Pollution Data (20\%)}
    \end{minipage}
    \hfill
    \begin{minipage}{0.45\textwidth}
        \centering
        \includegraphics[width=\textwidth]{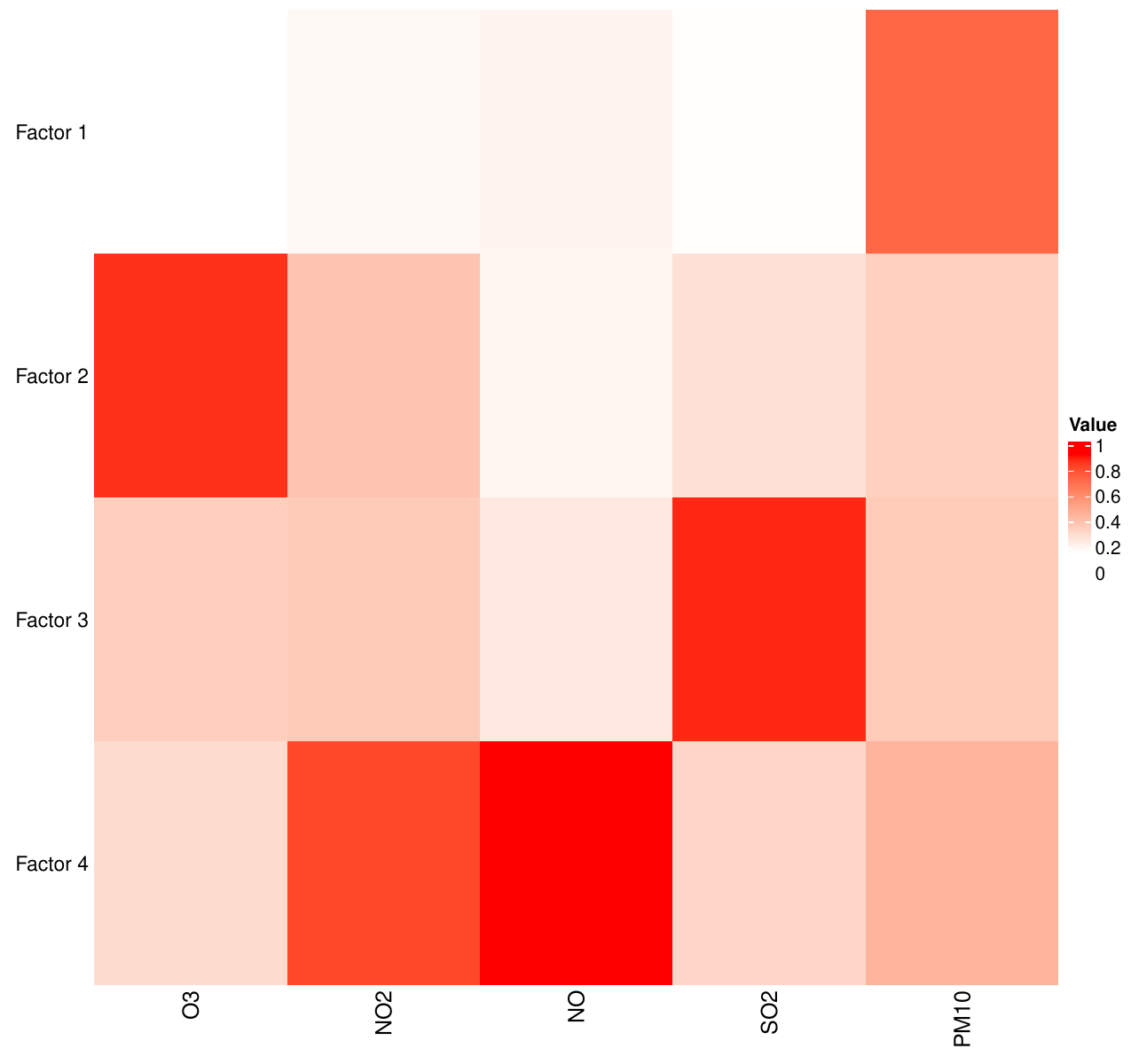} 
        \caption{Estimated $B^\top$ (20\%)}
    \end{minipage}
\end{figure}
\clearpage

\begin{figure}[h]
    \centering
    \begin{minipage}{0.45\textwidth}
        \centering
        \includegraphics[width=\textwidth]{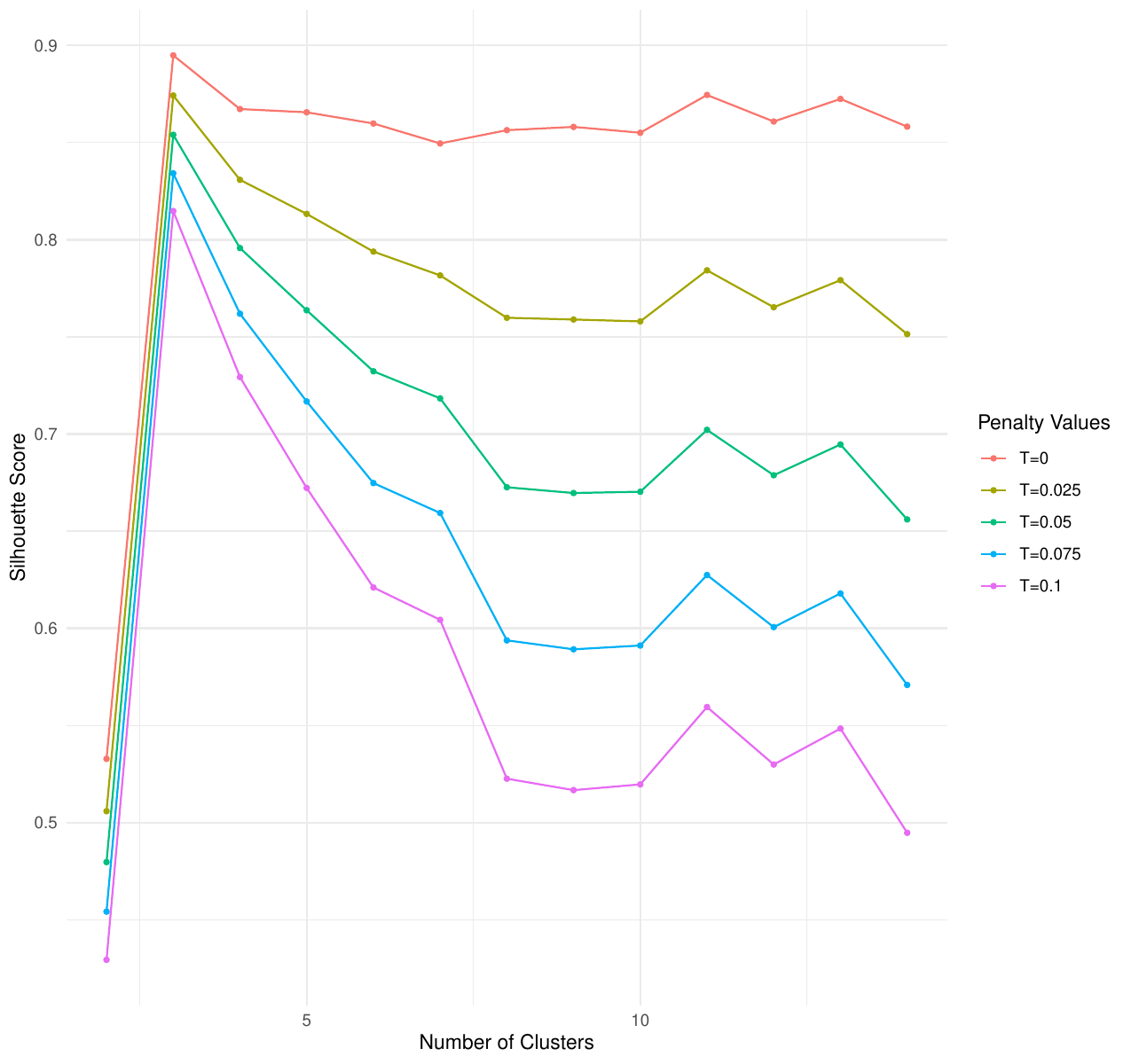} 
        \caption{Penalized ASW Curves for Winter Air Pollution Data (10\%)}
    \end{minipage}
    \hfill
    \begin{minipage}{0.45\textwidth}
        \centering
        \includegraphics[width=\textwidth]{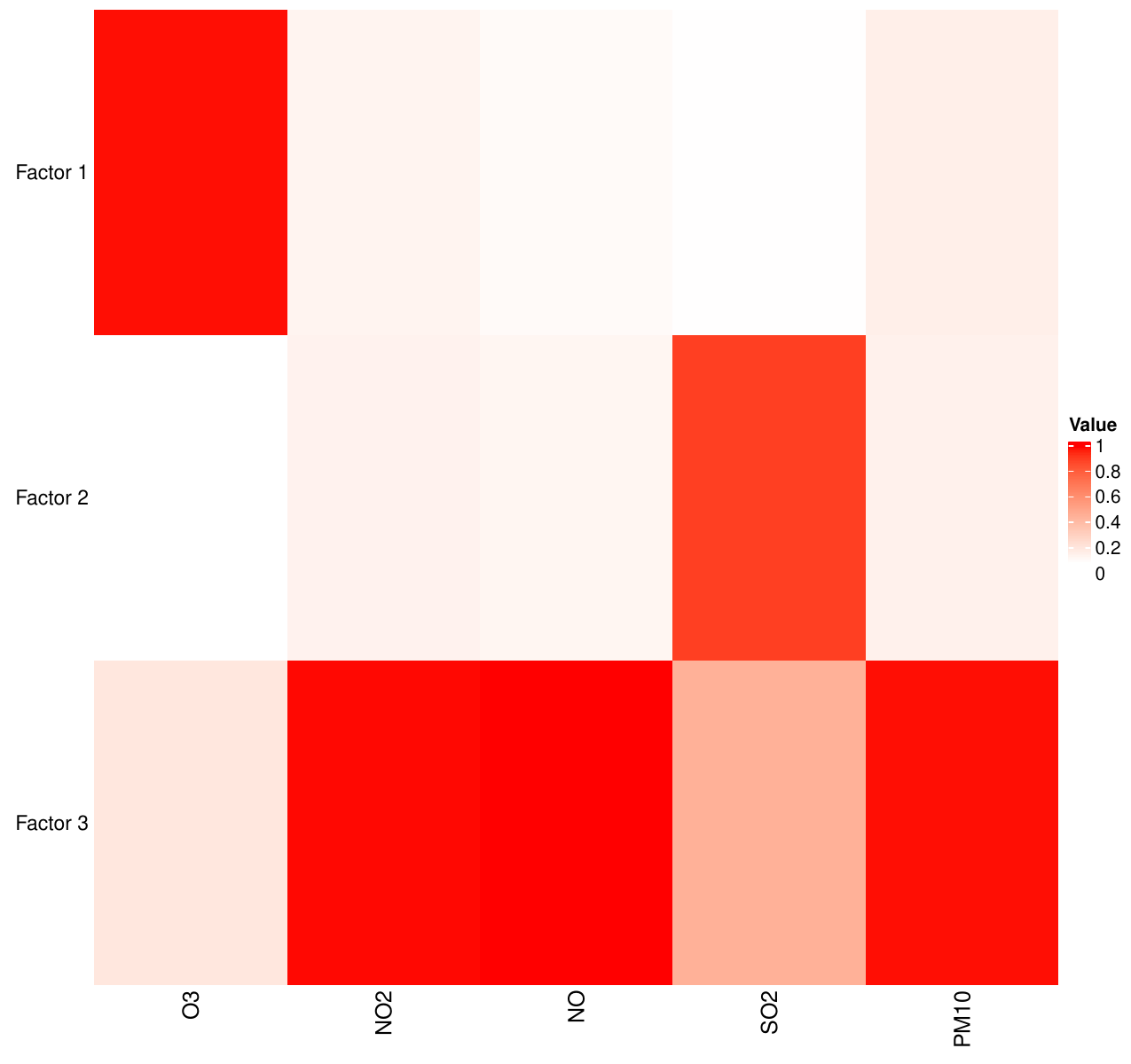} 
        \caption{Estimated $B^\top$ (10\%)}
    \end{minipage}
\end{figure}

\begin{figure}[h]
    \centering
    \begin{minipage}{0.45\textwidth}
        \centering
        \includegraphics[width=\textwidth]{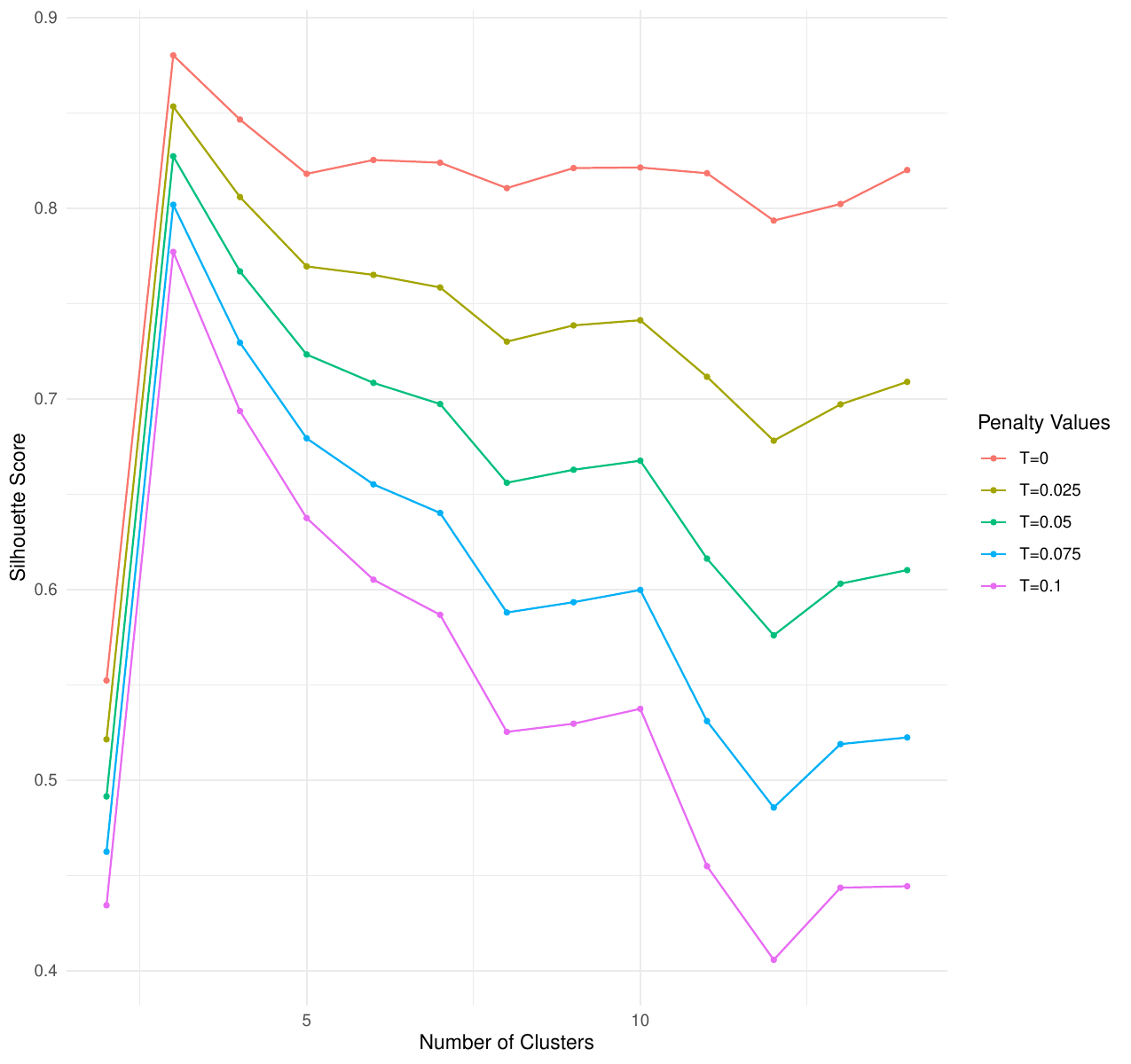} 
        \caption{Penalized ASW Curves for Winter Air Pollution Data (15\%)}
    \end{minipage}
    \hfill
    \begin{minipage}{0.45\textwidth}
        \centering
        \includegraphics[width=\textwidth]{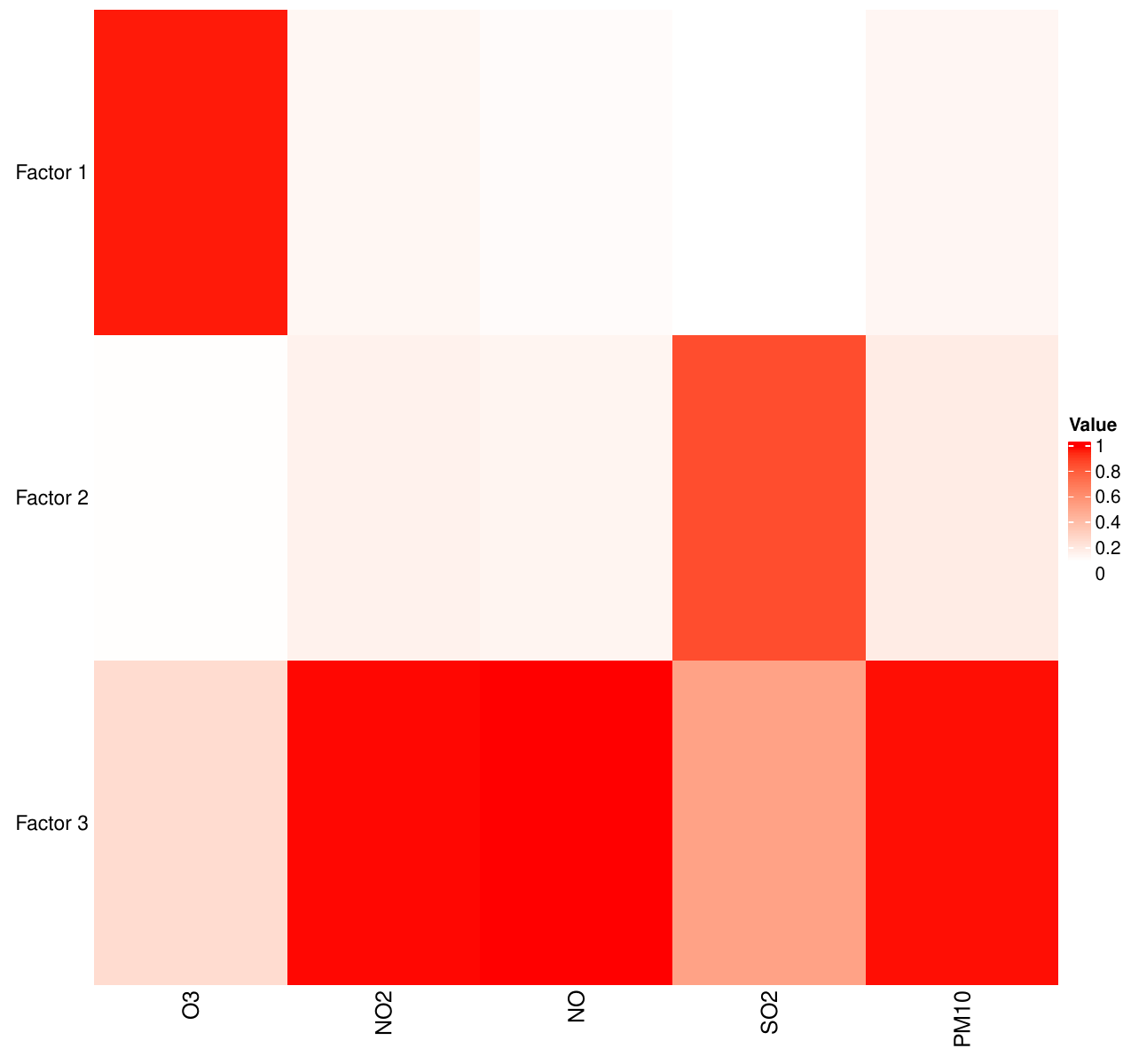} 
        \caption{Estimated $B^\top$ (15\%)}
    \end{minipage}
\end{figure}

\begin{figure}[h]
    \centering
    \begin{minipage}{0.45\textwidth}
        \centering
        \includegraphics[width=\textwidth]{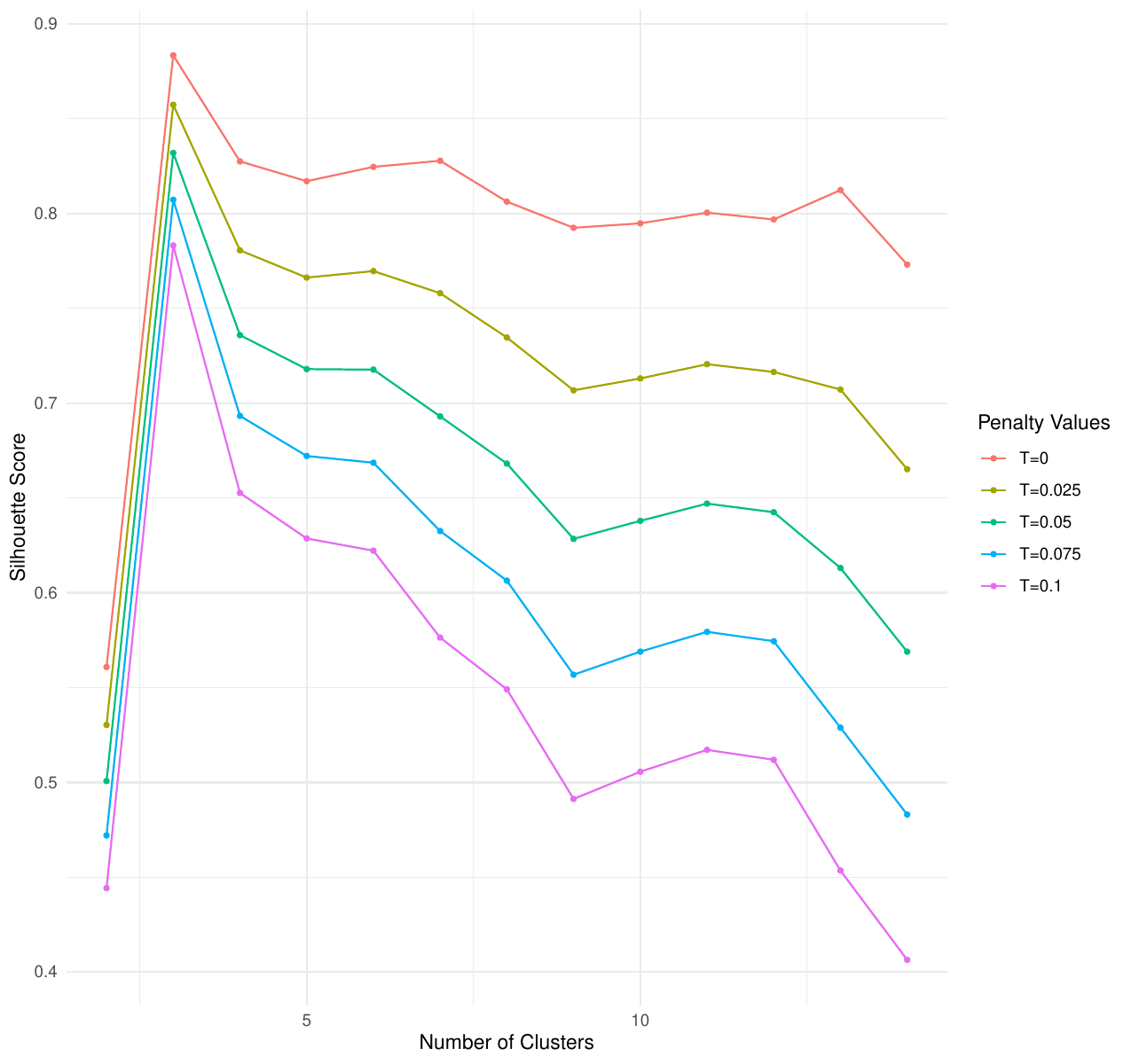} 
        \caption{Penalized ASW Curves for Winter Air Pollution Data (20\%)}
    \end{minipage}
    \hfill
    \begin{minipage}{0.45\textwidth}
        \centering
        \includegraphics[width=\textwidth]{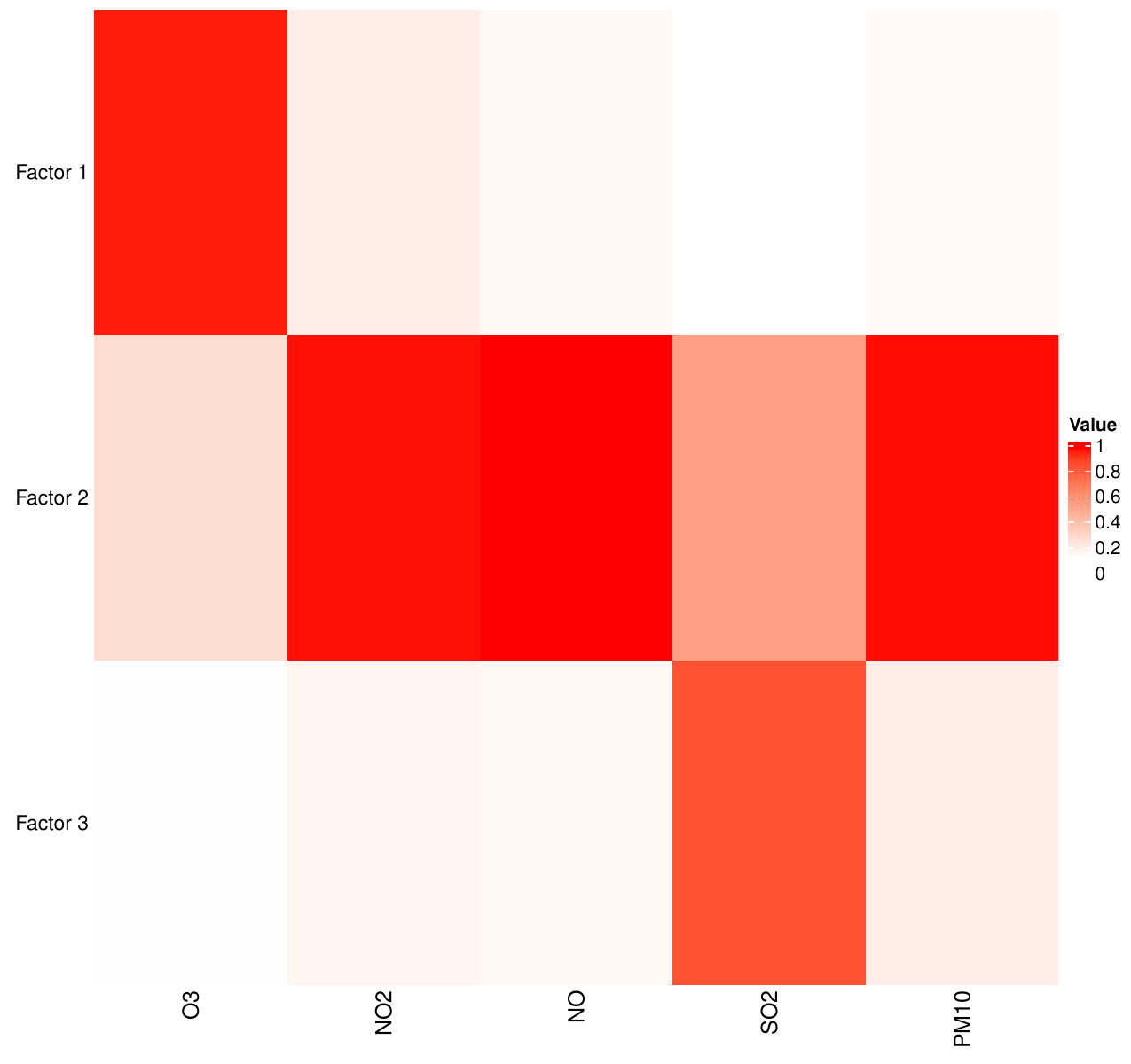} 
        \caption{Estimated $B^\top$ (20\%)}
    \end{minipage}
\end{figure}
\clearpage
\subsection{River Discharge Data}

The results here correspond to Section 6.2.2 of \cite{deng2024estimation}.


\begin{figure}[h]
    \centering
    \begin{minipage}{0.45\textwidth}
        \centering
        \includegraphics[width=\textwidth]{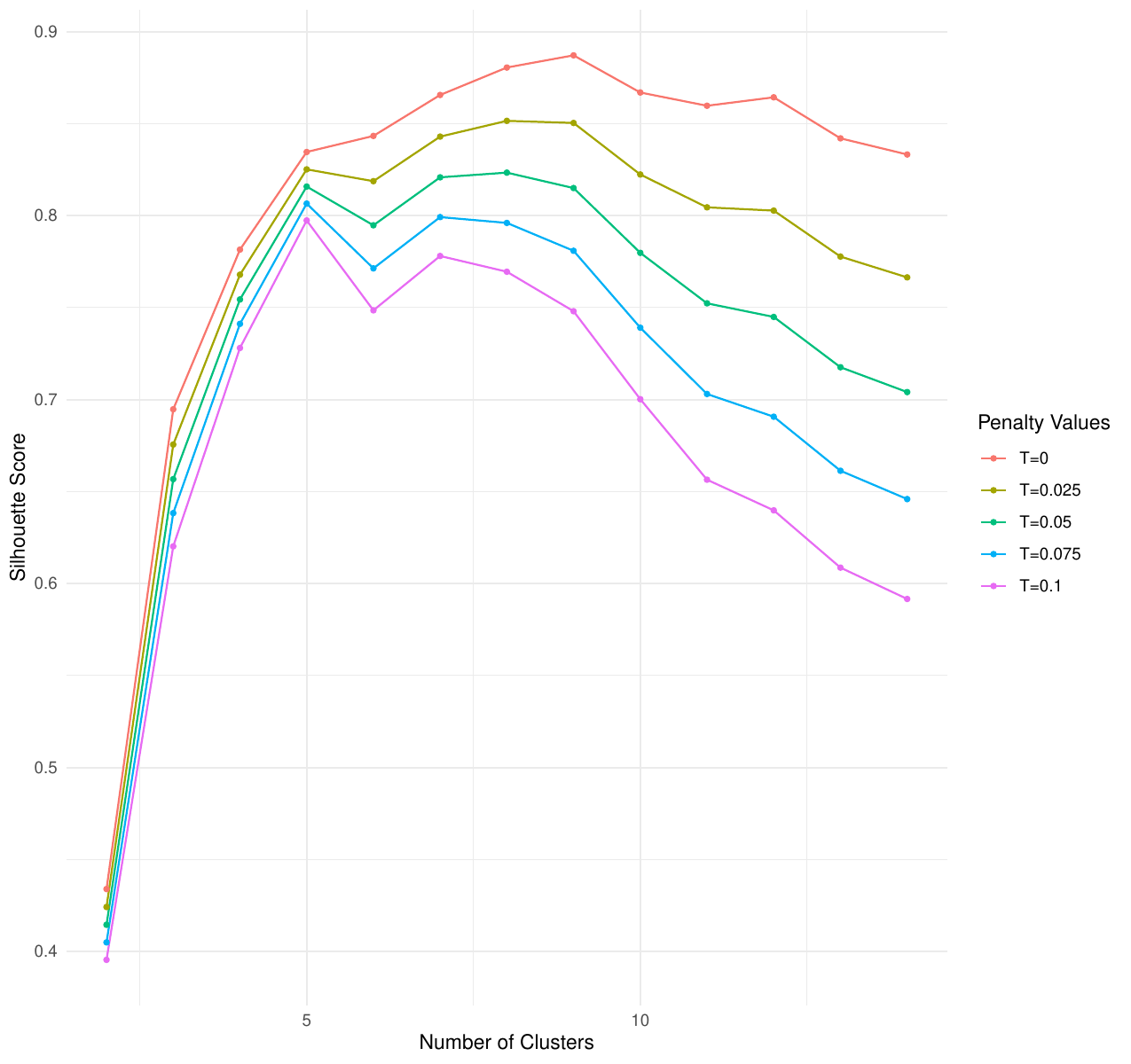} 
        \caption{Penalized ASW Curves for River Discharge Data (1\%)}
    \end{minipage}
    \hfill
    \begin{minipage}{0.45\textwidth}
        \centering
        \includegraphics[width=\textwidth]{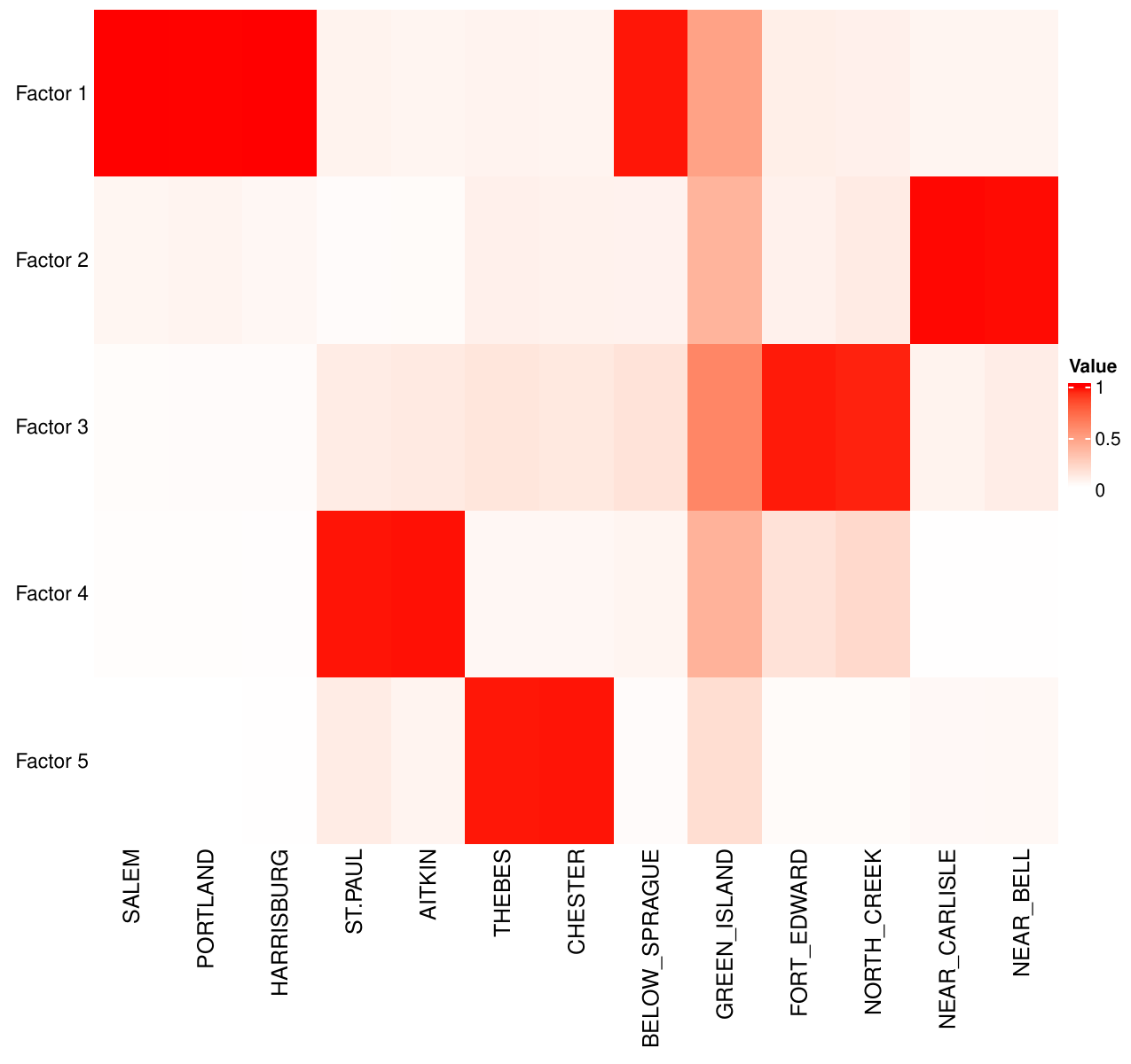} 
        \caption{Estimated $B^\top$ (1\%)}
    \end{minipage}
\end{figure}

\begin{figure}[h]
    \centering
    \begin{minipage}{0.45\textwidth}
        \centering
        \includegraphics[width=\textwidth]{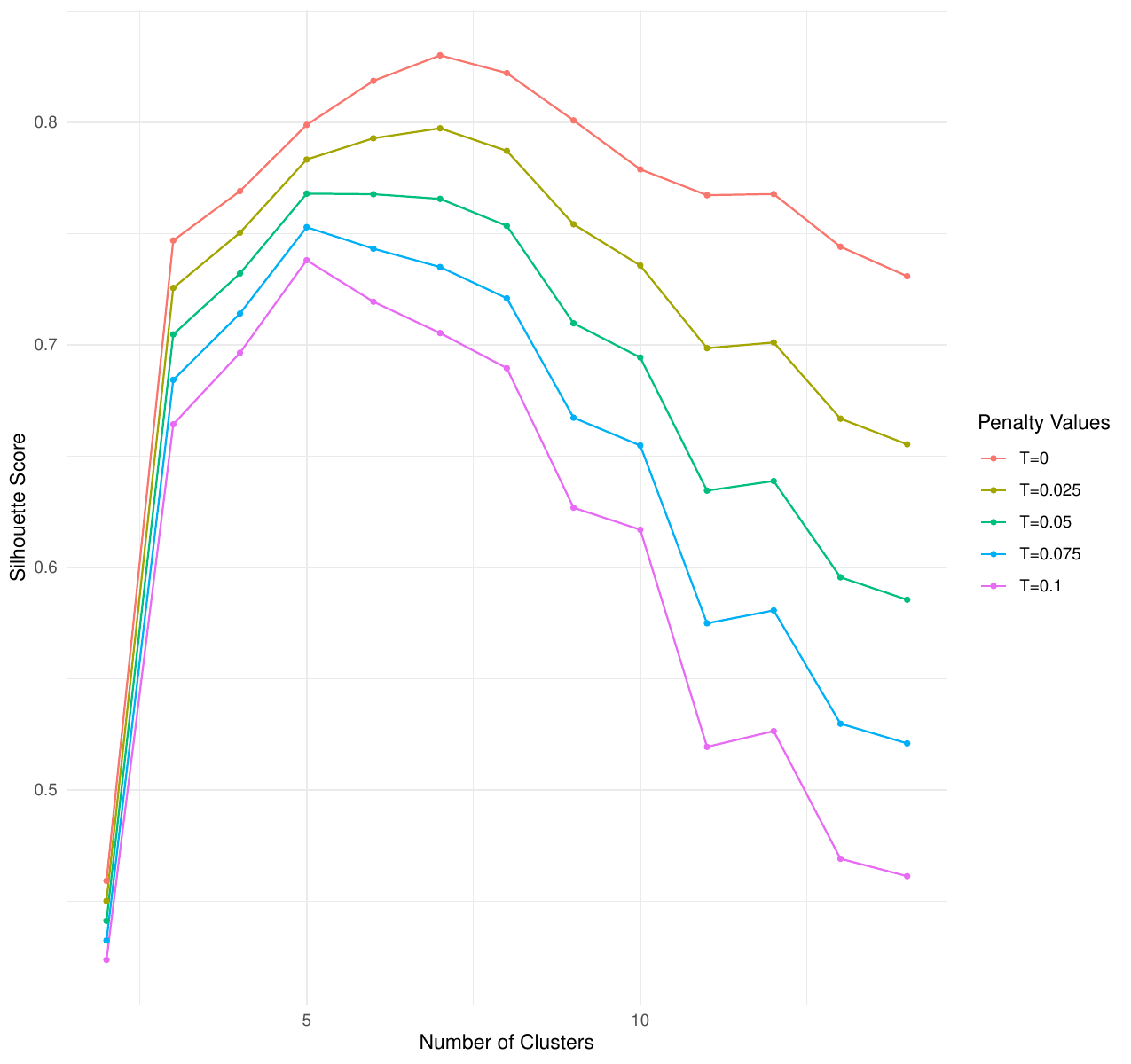} 
        \caption{Penalized ASW Curves for River Discharge Data (5\%)}
    \end{minipage}
    \hfill
    \begin{minipage}{0.45\textwidth}
        \centering
        \includegraphics[width=\textwidth]{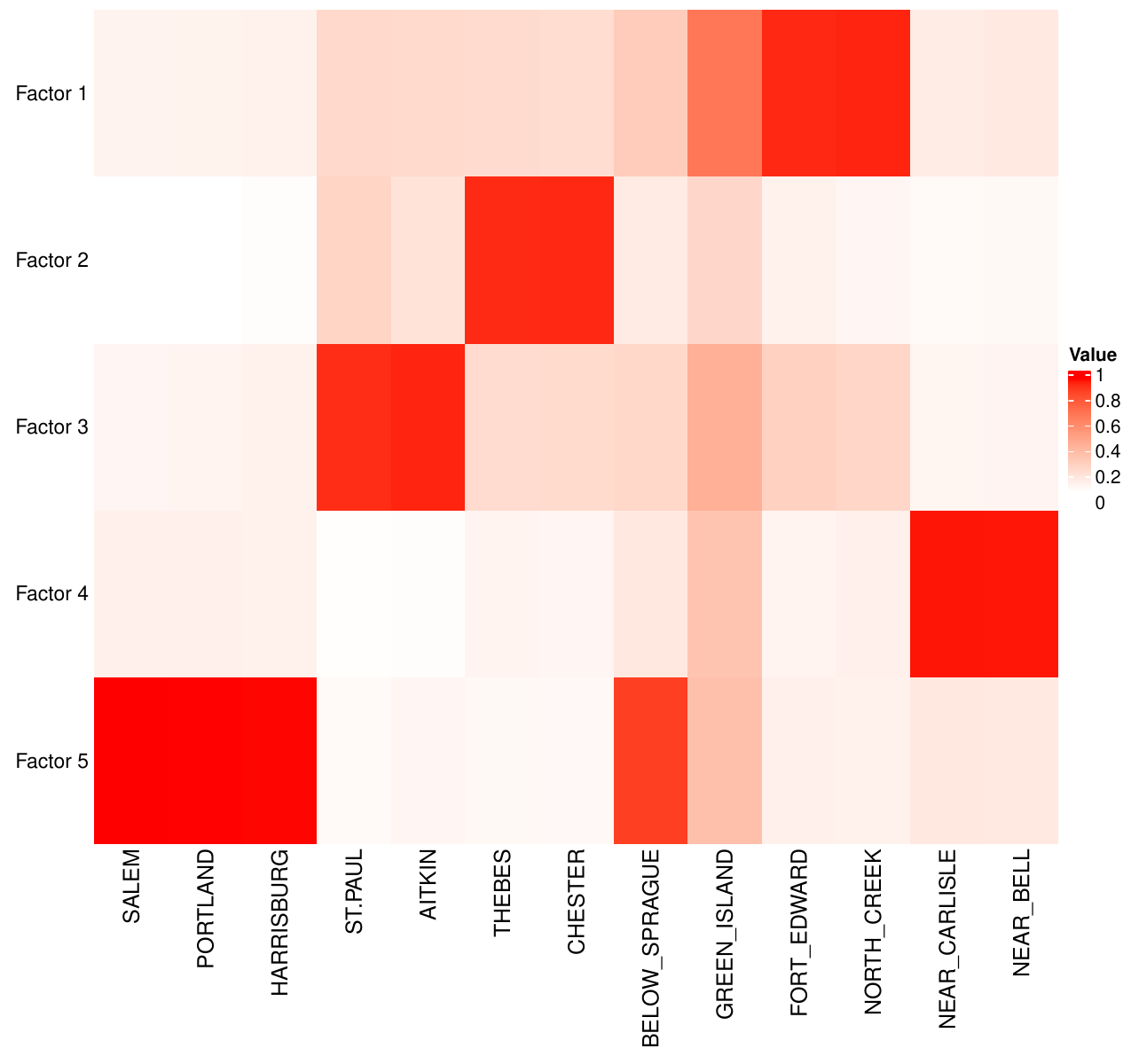} 
        \caption{Estimated $B^\top$ (5\%)}
    \end{minipage}
\end{figure}

\begin{figure}[h]
    \centering
    \begin{minipage}{0.45\textwidth}
        \centering
        \includegraphics[width=\textwidth]{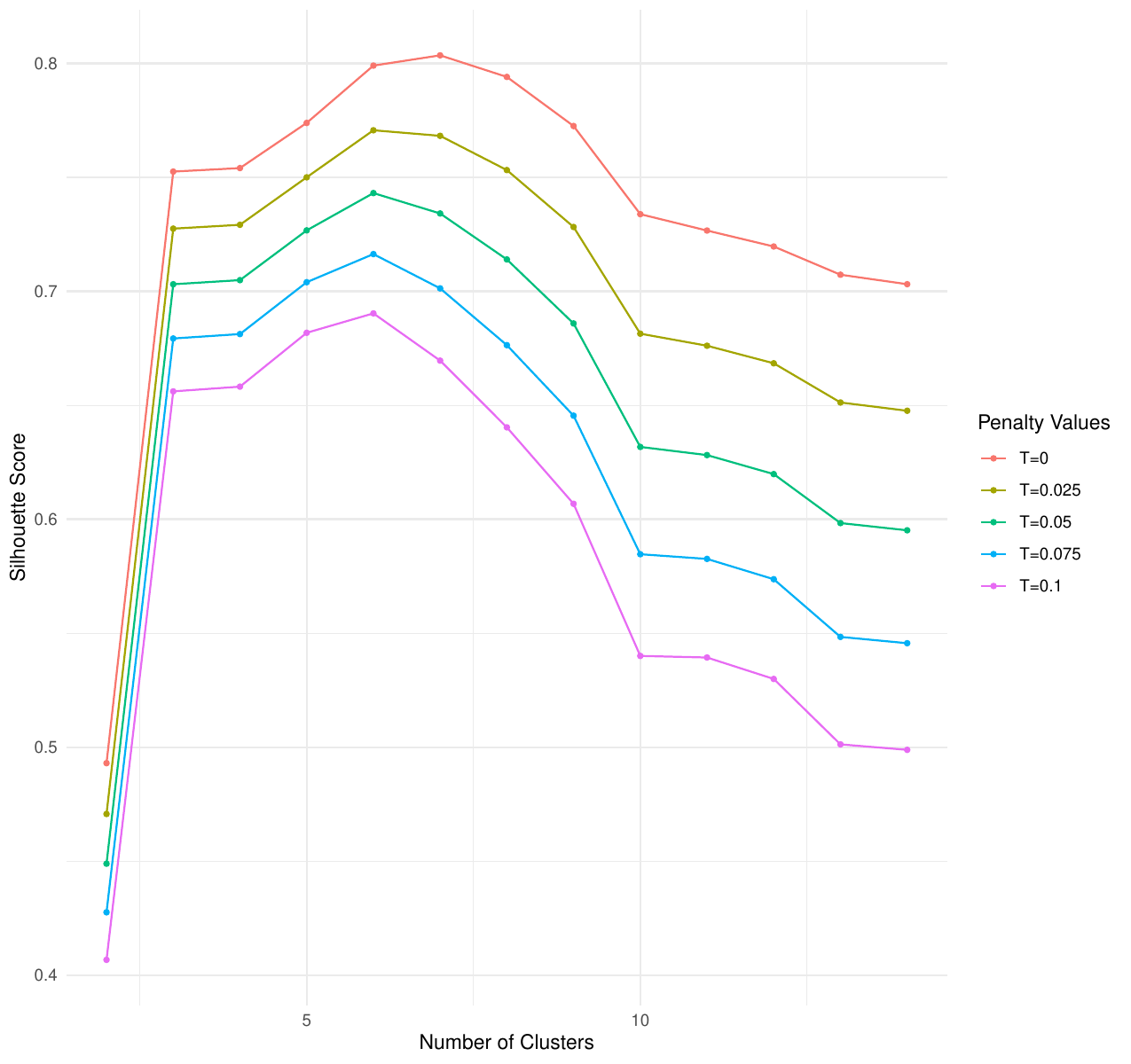} 
        \caption{Penalized ASW Curves for River Discharge Data (10\%)}
    \end{minipage}
    \hfill
    \begin{minipage}{0.45\textwidth}
        \centering
        \includegraphics[width=\textwidth]{figures-new/riverplot2-update10.pdf} 
        \caption{Estimated $B^\top$ (10\%)}
    \end{minipage}
\end{figure}

\begin{figure}[h]
    \centering
    \begin{minipage}{0.45\textwidth}
        \centering
        \includegraphics[width=\textwidth]{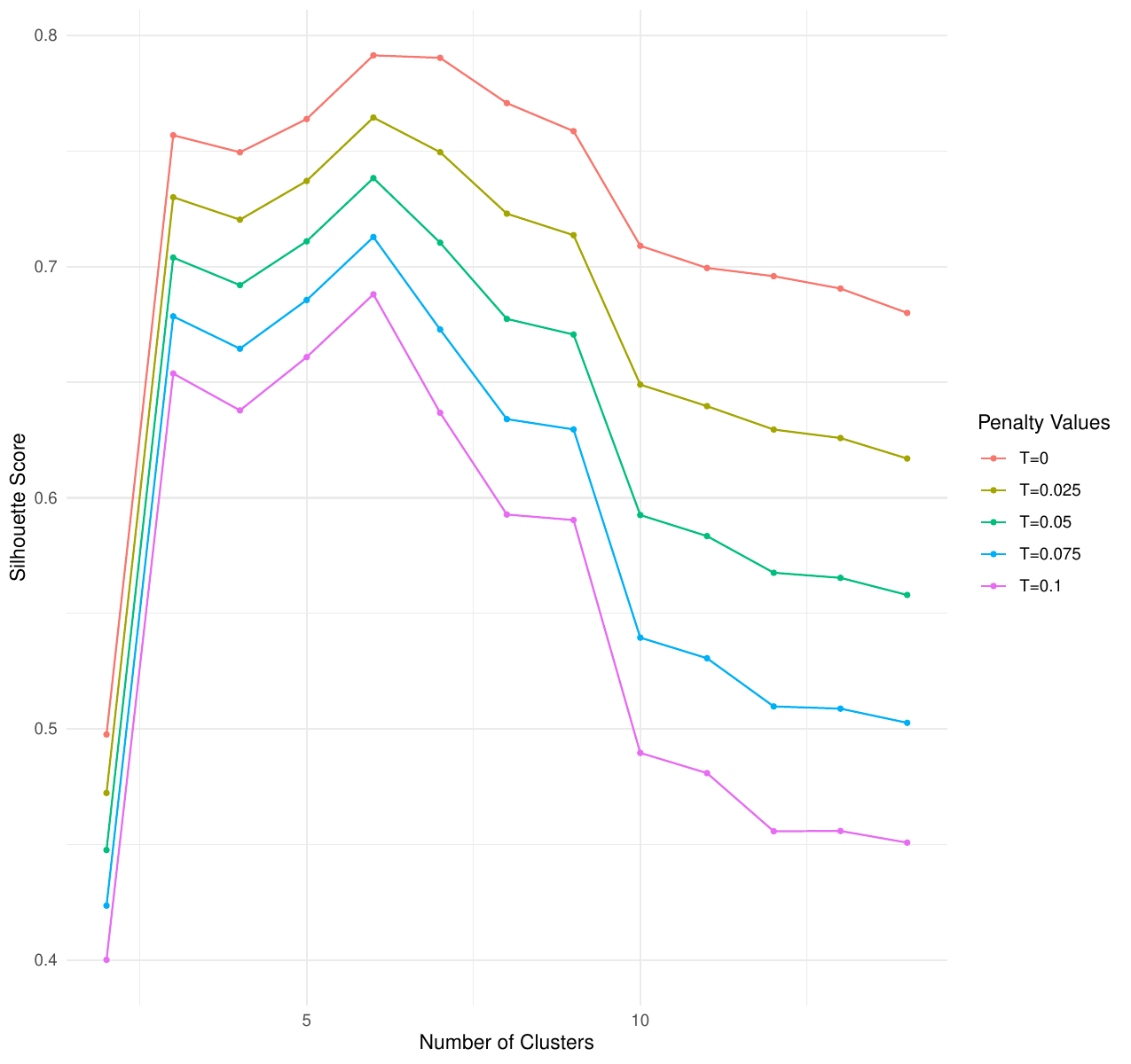} 
        \caption{Penalized ASW Curves for River Discharge Data (15\%)}
    \end{minipage}
    \hfill
    \begin{minipage}{0.45\textwidth}
        \centering
        \includegraphics[width=\textwidth]{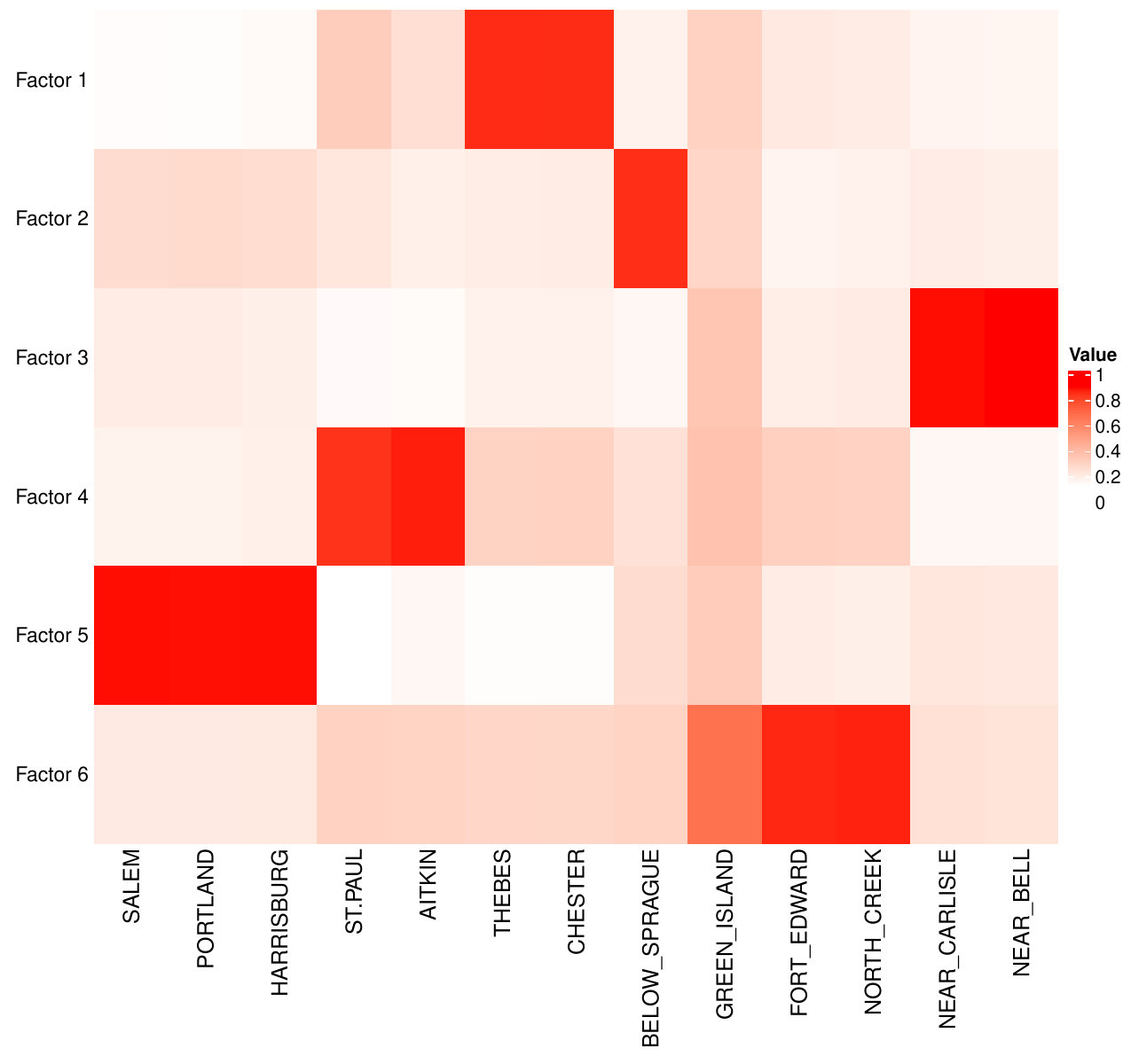} 
        \caption{Estimated $B^\top$ (15\%)}
    \end{minipage}
\end{figure}

\begin{figure}[h]
    \centering
    \begin{minipage}{0.45\textwidth}
        \centering
        \includegraphics[width=\textwidth]{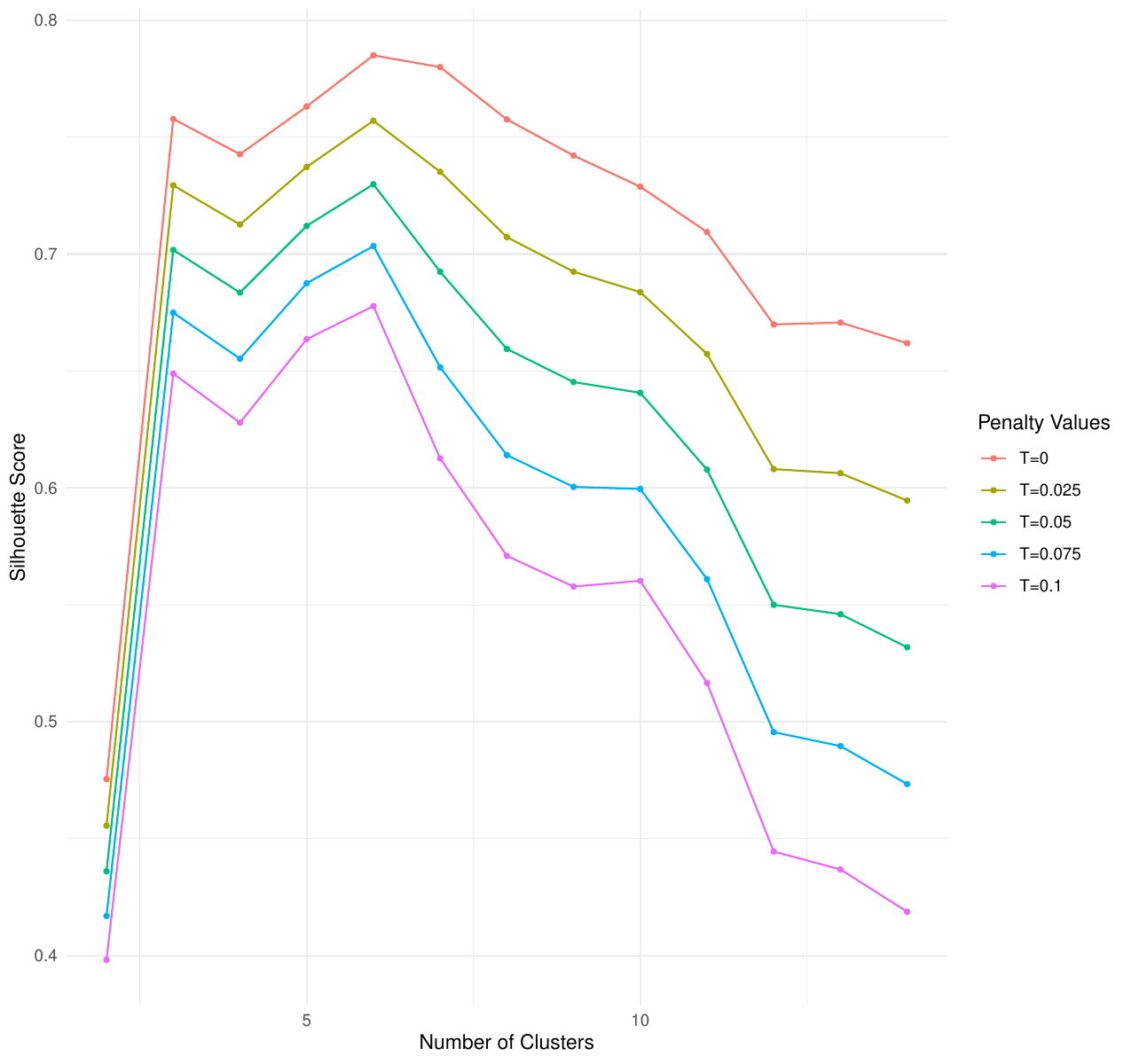} 
        \caption{Penalized ASW Curves for River Discharge Data (20\%)}
    \end{minipage}
    \hfill
    \begin{minipage}{0.45\textwidth}
        \centering
        \includegraphics[width=\textwidth]{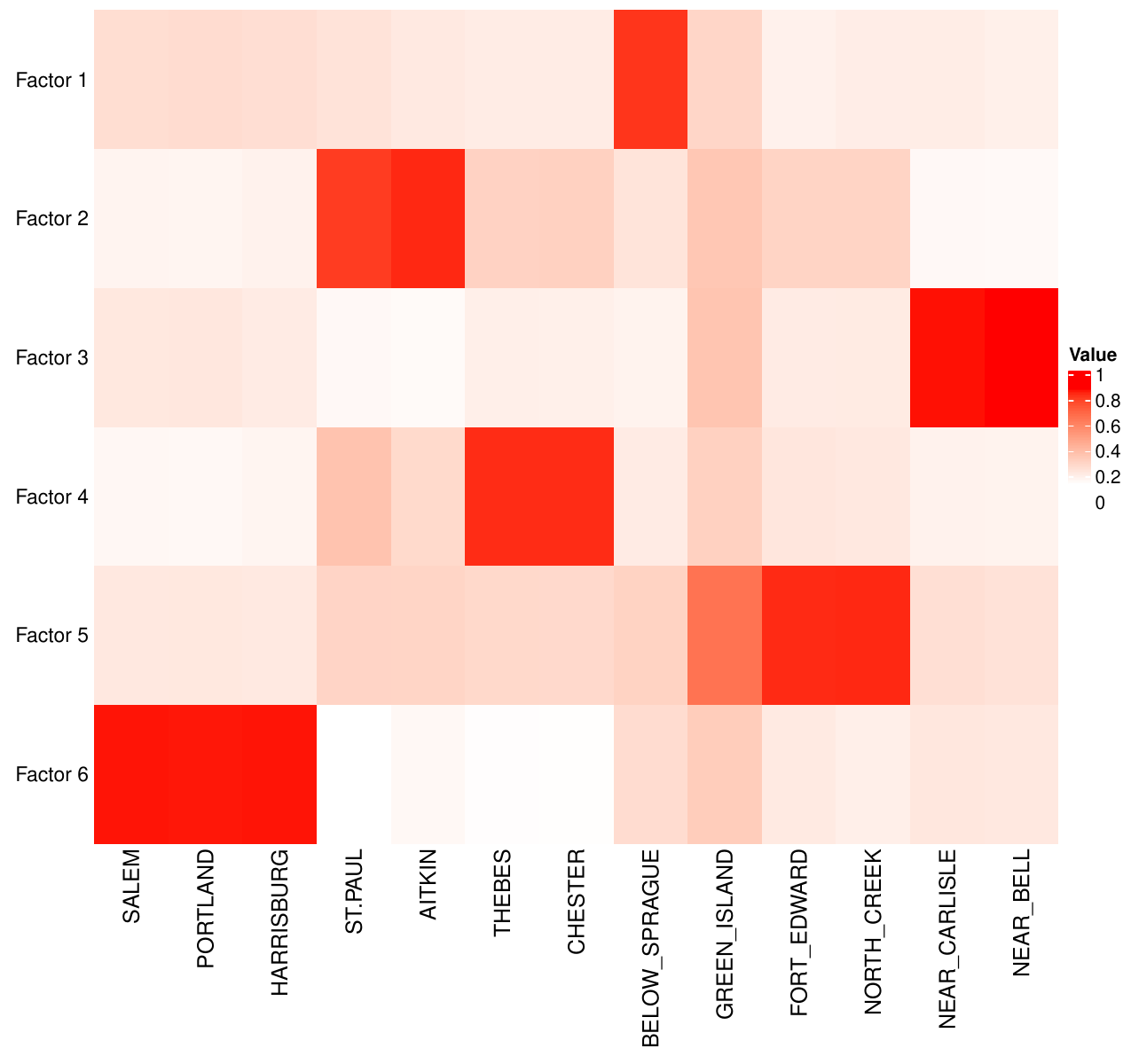} 
        \caption{Estimated $B^\top$ (20\%)}
    \end{minipage}
\end{figure}

\clearpage
\section{Codes}
 The R codes that implement the simulation and real data studies can be found at \url{https://github.com/SyuanD/SphCluster.git}.

\bibliographystyle{plainnat}
\bibliography{trial.bib}
 
\end{document}